\documentclass[10pt]{article}
\usepackage{times, cite}
\usepackage{fullpage, pifont, enumitem}%package to reduce margins
\usepackage[pass]{geometry}
\usepackage[usenames]{color}
\usepackage[colorlinks = true]{hyperref}
\usepackage[english]{babel}
\usepackage{bbm, amsbsy, amsmath, amssymb, amsthm, graphicx, mathtools}
\usepackage[caption=false, font=normalsize, labelfont=sf, textfont=sf]{subfig}
\usepackage[toc, page, header,title]{appendix}
\usepackage{colortbl, booktabs}

\numberwithin{equation}{section}
\newcommand{\blk}[2]{  
	\begin{minipage}{#1\textwidth} 
		\centering \vspace*{0.08cm}
		#2 \vspace*{0.05cm} 	
\end{minipage}}

\newcommand{\hinfty}{{\hat{\infty}}}
\newcommand{\tmem}[1]{{\em #1\/}}

\newcommand{\bl}[1]{{#1}}
\newcommand{\tmop}[1]{\ensuremath{\operatorname{#1}}}
\newcommand  \stack[2]  {\overset{\text{#1}}{#2}}
\def \dif {\text{d}}
\def \eye {\mathbf{I}}
\def \zero {\mathbf{0}}

\def \sign      {\tmop{sign}}

\def \diag      {\operatorname*{diag}}

\def \maximize  {\operatorname*{maximize}}
\def \st        {\operatorname*{subject\ to\ }}

\def \var       {\operatorname*{Var}}
\def \nn        {\nonumber}
\def \E {\mathbb{E}}
\def \R {\mathbb{R}}
\def \P {\mathbb{P}}

\def \C {\mathbb{C}}
\def \TT{\mathbb{T}}
\def \I {\mathbb{I}}
\def \balpha {{\boldsymbol {\alpha}}}
\def \bbeta  {{\boldsymbol {\beta}}}

\def \btheta {{\boldsymbol {\theta}}}

\def \bnu    {{\boldsymbol {\nu}}}
\def \brho {{\boldsymbol {\rho}}}
\def \a {\mathbf{a}}

\def \c {\mathbf{c}}
\def \d {\mathbf{d}}
\def \s {\mathbf{s}}
\def \e {\mathbf{e}}
\def \f {\mathbf{f}}
\def \g {\mathbf{g}}
\def \h {\mathbf{h}}

\def \n {\mathbf{n}}

\def \q {\mathbf{q}}
\def \s {\mathbf{s}}
\def \u {\mathbf{u}}
\def \v {\mathbf{v}}
\def \w {\mathbf{w}}
\def \x {\mathbf{x}}
\def \y {\mathbf{y}}
\def \z {\mathbf{z}}
\def \A {\mathcal{A}}
\def \F {\mathcal{F}}
\def \M {\mathcal{M}}
\def \N {\mathcal{N}}
\def \NN {\mathcal{N}}

\def \D {\mathbf{D}}
\def \W {\mathbf{W}}
\def \mA {\mathbf{A}}
\def \mB {\mathbf{B}}
\def \mZ {\mathbf{Z}}
\def \mS {\mathbf{S}}
\def \mH {\mathbf{H}}
\newtheorem{theorem}{Theorem}[section]
\newtheorem{proposition}{Proposition}[section]
\newtheorem{lemma}{Lemma}[section]
\newtheorem{corollary}{Corollary}[section]
\providecommand{\keywords}[1]{{\em Keywords: }\!\!#1}

\newcommand\blfootnote[1]{%
	\begingroup
	\renewcommand\thefootnote{}\footnote{#1}%
	\addtocounter{footnote}{-1}%
	\endgroup
}

\begin{document}

\title{{Approximate Support Recovery of Atomic Line Spectral Estimation:}\\{ A Tale of Resolution and Precision}}

	\author{{Qiuwei Li \qquad     \qquad Gongguo Tang}
		\\
		 { Department of Electrical Engineering,  Colorado School of Mines,  CO,  USA}
	}
	\date{}
	\maketitle

\blfootnote{\textit{Email addresses}: \texttt{qiuli@mines.edu} (Qiuwei Li),  \texttt{gtang@mines.edu} (Gongguo Tang)}

\blfootnote{This work was supported by the National Science Foundation [CCF-1464205, CCF-1704204].}

\begin{abstract}
	
\noindent
This work investigates the parameter estimation performance of super-resolution line spectral estimation using atomic norm minimization. The focus is on analyzing the algorithm's accuracy of inferring the frequencies and complex magnitudes from noisy observations. When the Signal-to-Noise Ratio is reasonably high and the true frequencies are separated by $O(\frac{1}{n})$,  the atomic norm estimator is shown to localize the correct number of frequencies,  each within a neighborhood of size $O(\sqrt{{\log n}/{n^3}} \sigma)$ of one of the true frequencies. Here $n$ is half the number of temporal samples and $\sigma^2$ is the Gaussian noise variance. The analysis is based on a primal-dual witness construction procedure. The obtained error bound matches the Cram\'er-Rao lower bound up to a logarithmic factor. The relationship between resolution (separation of frequencies) and precision or accuracy of the estimator is highlighted. Our analysis also reveals that the atomic norm minimization can be viewed as a convex way to solve a $\ell_1$-norm regularized,  nonlinear and nonconvex least-squares problem to global optimality.

		\bigskip

\noindent\keywords{
atomic norm,  line spectral estimation,  primal-dual witness construction,   super-resolution,  support recovery
}
\end{abstract}

%\linenumbers
\section{Introduction}\label{sec:intro}

Line spectral estimation,  which aims at approximately inferring the frequency and coefficient parameters from a superposition of complex sinusoids embedded in white noise,  is one of the fundamental problems in statistical signal processing. When the temporal and frequency domains are exchanged,  this classical problem was reinterpreted as the problem of mathematical super-resolution recently~\cite{Candes:2014br,  candes2013super,  fernandez2016super}. This line of work promotes the use of a convex sparse regularizer to solve inverse problems involving spectrally sparse signals,  distinguishing them from classical methods based on root finding and singular value decompositions (e.g.,  Prony's method,  MUSIC,  ESPIRIT,  Matrix Pencil,  etc.). The convex regularizer,  a particular instance of the general atomic norms,  has been shown to achieve optimal performance in signal completion~\cite{Tang:2013fo},  denoising~\cite{Tang:2013gd},  and outlier removal~\cite{Tang:2014outlier,  fernandez2016demixing}. For these signal processing tasks,  either one can recover the spectral signal exactly (and hence extract the true frequencies precisely),  or the error metric is defined using the signal instead of the frequency parameters. The most relevant question of the accuracy of {noisy} frequency estimation has been elusive. This work investigates the parameter estimation performance of super-resolution line spectral estimation using atomic norm minimization. More precisely,  given noisy observations
\begin{align}\label{eqn:noisy_samples}
y(t) = x^\star(t) + w(t),  t = -n,  \ldots,  n
\end{align}
of a spectrally sparse signal
\begin{align}\label{eqn:true}
x^\star(t)&=\sum_{\ell=1}^k c_\ell^\star \exp(i2\pi f_\ell^\star t),  t = -n,  \ldots,  n
\end{align}
with unknown frequencies $T^\star=\left\{f_\ell^\star\right\}_{\ell=1}^k$ and complex amplitudes $\left\{c_\ell^\star\right\}_{\ell=1}^k$, 
we will derive conditions under which the atomic norm formulation will return the correct number of frequencies,  and establish bounds on the frequency and coefficient estimation errors. An informal version of our main result is given in the following theorem,  while a formal statement is presented in Theorem~\ref{thm:main}.

\begin{theorem}[Informal]\label{thm:main:informal}
Suppose we observe $2n+1$ noisy consecutive samples $y(t)=x^\star(t)+w(t)$ of the signal~\eqref{eqn:true} with $w(t)$  being i.i.d. complex Gaussian variables of mean zero and variance $\sigma^2$.  If the unknown frequencies are well-separated,  the Signal-to-Noise Ratio (SNR) is large,  and the dynamic range of the coefficients is small,  then with probability at least $1-\frac{1}{n^2}$,   solving an atomic norm regularized least-squares problem with a large enough regularization parameter will return exactly $k$ estimated frequencies $\{{f}^{\mathrm{glob}}_\ell\}_{\ell=1}^k$ and coefficients $\{{c}^{\mathrm{glob}}_\ell\}_{\ell=1}^k$ that,  when properly ordered, 
satisfy
\begin{align}
\max_{1\leq\ell\leq k} |c_\ell^\star| |{f}^{\mathrm{glob}}_\ell-f^\star_\ell| & = O(\frac{\sqrt{\log n}}{n^{3/2}}\sigma)\label{eqn:freqbound:informal}, \\
\max_{1\leq\ell\leq k}|{c}^{\mathrm{glob}}_\ell-c^\star_\ell|&= O(\sqrt{\frac{\log n}{n}}\sigma)\label{eqn:coefbound:informal}.
\end{align}
\end{theorem}

We  would like to first point out that this {\em frequency estimator} $\{f^{\mathrm{glob}}_\ell\}$ given by the atomic norm regularized least-squares is asymptotically unbiased. The $\ell_1$ norm minimization (atomic norm minimization is  an extension of it) is usually considered biased because it pushes down the solution using the $\ell_1$ norm. In the context of atomic norm minimization,  the estimator for the coefficient vector is indeed biased for the same reason. However,  the frequency estimator,  which is of more interest,  might still be unbiased since it is not pushed down by the atomic norm formulation. Indeed,  our result shows that the frequency estimator is at least asymptotically unbiased. 
\begin{corollary}
Under the same setup as in Theorem~\ref{thm:main:informal}, with  probability at least $1- \frac{1}{n^2}$,  the frequency estimator obtained by the  atomic norm regularized minimization is asymptotic unbiased.
\end{corollary}
\begin{proof}
To see this,  we note that for any $i$,
\begin{align*}
\mathbb{E}[{f}^{\mathrm{glob}}_i]-f^\star_i
\leq\mathbb{E}\{|{f}^{\mathrm{glob}}_i-f^\star_i|\}
&=\int_{\Omega}|{f}^{\mathrm{glob}}_i(\omega)-f^\star_i(\omega)|\dif\omega+\int_{\Omega^c}|{f}^{\mathrm{glob}}_i(\omega)-f^\star_i(\omega)|\dif\omega\\
&\leq O(\frac{\sqrt{\log n}}{c_{\min}^{\star}n^{3/2}}\sigma)+\frac{2}{n^2}
\\
&=o(\frac{1}{n}).
\end{align*}
Here $\Omega$ is the high-probability sample space where our main result \eqref{eqn:freqbound:informal} holds,  $\Omega^c$ is its complement space,  and $c_{\min}^{\star}$ is defined as the smallest magnitude of $\{c_\ell^\star\}$. The second inequality follows from Eq. \eqref{eqn:freqbound:informal},  $\int_\Omega \dif\omega\leq 1$,   $\int_{\Omega^c} \dif\omega\leq \frac{1}{n^2}$,  and the fact that any frequency is defined in $\mathbb{T}=[0, 1]$. 
Therefore,  the frequency estimator is at least asymptotically unbiased. 
\end{proof}
By the asymptotic unbiasedness of our {\em atomic frequency estimator} and considering that the  Cram\'er-Rao bound (CRB)~\cite{Stoica:1989dn} can be viewed as the best squared error bound for any unbiased frequency estimators, we now compare our main result \eqref{eqn:freqbound:informal} (after taking the square) with the CRB, as well as the two most famous classical line spectral estimation methods, i.e., the MUSIC and Maximum Likelihood Estimation (MLE), in Table \ref{tab:comp:classic}.
\begin{table}[h!t]
\centering
	\begin{tabular}{lc}
		\toprule 
	 {\textbf{Method}} 
		& 	\blk{0.3}{ \textbf{Squared-Error Bound} }  
\\	\midrule
		 
		{CRB~\cite{Stoica:1989dn} }&  {$O(\frac{\sigma^2}{{c_{\min}^{\star2}}n^{3}})$}  
	\\[2ex]
		
	{MUSIC~\cite{Stoica:1989dn}	}& {$O(\frac{\sigma^2}{T{c_{\min}^{\star2}}n^3} + \frac{\sigma^4}{T {c_{\min}^{\star4}}n^4})$}	
	\\[2ex]
		
	{MLE~\cite{Stoica:1989dn} }& {$O(\frac{\sigma^2}{T{c_{\min}^{\star2}}n^3} + \frac{\sigma^4}{T {c_{\min}^{\star4}} n^4})$}	
	\\[2ex]
		
	\rowcolor[gray]{0.9}{{{This work~\eqref{eqn:freqbound:informal}}}}&  {$O(\frac{\sigma^2\log n}{{c_{\min}^{\star2}}n^{3}})$}   \\ 
		\bottomrule 
	\end{tabular} 
	\caption {Comparison with the classical line spectral estimation methods.} \label{tab:comp:classic}
\end{table}
We conclude that the squared error bound of the {\em atomic frequency estimator} matches the CRB up to  a logarithmic factor. We also note that the MUSIC and the MLE only have asymptotic mean squared error in the sense that the number of snapshots $T$ has to be infinitely large~\cite{Stoica:1989dn}. We emphasize that our results are non-asymptotic,  which hold for finite-length,  single-snapshot signals  (i.e., $T=1$),  while classical methods such as MUSIC and MLE are not efficient (i.e., approaching CRB) even with an infinite number of snapshots,  as long as the signal length $n$ is finite.

\section{Signal Model and Atomic Norm Regularization} \label{sec:model}
This paper considers the spectral estimation problem: given noisy temporal samples,  how well can we estimate the locations and determine the magnitudes of spectral lines? The signal of interest $x^\star(t)$ as expressed in~\eqref{eqn:true} is composed of only a small number of spectral spikes located in a normalized interval $\TT=[0, 1]$. We abuse notation and call $T^\star = \{f_\ell^\star\}_{\ell=1}^k$ the support of $\x^\star$. The number of frequencies,  $k$,  is referred to as the model order. The goal is to approximately localize these parameters from  a small number $2n+1$  of equispaced noisy samples given in~\eqref{eqn:noisy_samples}. For technical simplicity,  we assume $n = 2M$ is an even number.
The noise components $w(t)$ are i.i.d. centrally symmetric complex Gaussian variables with variance $\sigma^2$. To simplify notation,  we stack the temporal samples into vectors and write the observation model as
\begin{align}\label{eqn:observe}
\y&=\x^\star+\w, 
\end{align}
where $\x^\star:=[x^\star(-n), \ldots, x^\star(n)]^T, ~\y:=[y(-n), \ldots, y(n)]^T$ and $\w^\star:=[w^\star(-n), \ldots, w^\star(n)]^T$.

To estimate the frequency vector $\f^\star:=[f^\star_1, \ldots, f^\star_k]^T$ and the complex coefficient vector $\c^\star:=[c^\star_1, \ldots, c^\star_k]^T$,  we assume $k$ is small and treat $\x^\star$ as a sparse combination of atoms  $\a(f):=[e^{i2\pi(-n)f}, \ldots, e^{i2\pi n f}]^T$ parameterized by frequency $f\in\TT$,  that is, 
\begin{align}\label{eqn:atom_true}
\x^\star = \sum_{\ell = 1}^k c_\ell^\star \a(f^\star_\ell).
\end{align}
To exploit the structure  of $\x^\star$ encoded in the set of atoms $\A := \{\a(f),  f \in \TT\}$,  we follow~\cite{{Chandrasekaran:2010hl},  {Tang:2013fo}} and define the 
associated atomic norm as
\begin{align}\label{eqn:def:atomic:norm}
\|\x\|_{\A}&= \inf \left\{\sum_{\ell}|c_\ell|: \x=\sum_{\ell}c_\ell \a(f_\ell),  \forall f_\ell\in\TT,  c_\ell \in \C \right\}.
\end{align}
The dual norm of the atomic norm,  which is useful both algorithmically and theoretically,  is defined for any vector $\z$ as $\|\z\|_\A^* = \sup_{f\in\TT} | \a(f)^H  \z|$,  where  $^H$ denotes the Hermitian (conjugate transpose) operation. To solve atomic norm minimizations numerically,  the authors of~\cite{Tang:2013gd, bhaskar2013atomic} (see also~\cite{Candes:2014br}) first proposed to reformulate the atomic norm~\eqref{eqn:def:atomic:norm} as an equivalent semidefinite program. Other numerical schemes are studied in~\cite{rao2015forward, tewari2011greedy,  boyd2015alternating,  eftekhari2013greed}.

Given the noisy observation model~\eqref{eqn:observe},  it is natural to denoise $\x^\star$ by solving the atomic norm regularized minimization  program~\cite{bhaskar2013atomic,  Tang:2013gd}:
\begin{align}\label{eqn:primal}
{\x}^{\mathrm{glob}}=\operatorname*{argmin}_\x \frac{1}{2}\|\y-\x\|^2_{\mZ}+\lambda\|\x\|_{\A}.
\end{align}
For technical reasons,  we used a weighted $\ell_2$ norm,  $\|\z\|_{\mZ}:=\sqrt{\z^H \mZ\z}$,  to measure data fidelity. Here $\mZ=\diag(\frac{g_M(\ell)}{M})\in\R^{(4M+1)\times(4M+1)}$ with $g_M(\ell), \ell=-2M, \ldots, 2M$ defined in~\cite{{Tang:2013fo}} as the discrete convolution of two triangular
functions. We remark that,  in practice,  both a standard $\ell_2$ norm $\|\cdot\|_2$ and a weighted  $\ell_2$ norm $\|\cdot\|_{\mZ}$ achieve similarly satisfying performance. In this work,  we use $\|\cdot\|_{\mZ}$ with $\mZ=\diag(\frac{g_M(\ell)}{M})$ mainly for the purpose of introducing the Jackson kernel $K(f_2-f_1):=\a(f_1)^H\mZ\a(f_2)$  so that we can exploit the beautiful decaying properties of the Jackson kernel (see Section~\ref{sec:A:gM} for more details). When we exchange the frequency and temporal domains,  this weighting scheme trusts low-frequency samples more than high-frequency ones,  even though the noise levels are the same. The second term is a regularization term that penalizes solutions with  large atomic norms,  which typically correspond to spectrally dense signals. The regularization parameter $\lambda$,  whose value will be given later,  controls the trade-off between data fidelity and sparsity.

Once ${\x}^{\mathrm{glob}}$ was solved,  we can extract estimates of the frequencies either from the primal optimal solution ${\x}^{\mathrm{glob}}$ or from the corresponding dual optimal solution. Our goal is to characterize conditions such that i) we obtain exactly $k$ estimated frequencies; ii) there is a natural correspondence between the estimated frequencies and the true frequencies,  whose distances can be explicitly controlled; iii) the distances between the corresponding coefficients can also be explicitly bounded.

To formally present the main theorem,  we need to define a few more quantities. It is known that there is a resolution limit of the atomic norm approach in resolving the atoms,  or the frequency parameter $\f^\star$,  even from the noiseless data~\cite{tang2015resolution}. Therefore,  to recover the support of the line spectral signal $\x^\star$,  we need to impose certain separation condition on the distances of the true frequencies. For this purpose,  we define $\Delta(T)=\min_{\{f_\ell, f_m\}\subset T: f_\ell\neq f_m } |f_\ell-f_m|$,  where $|\cdot|$ is understood as the wrap-around distance in $\TT$. For example,  $|0.1-0.9| = 0.2$ under this distance. We also define 1) the dynamic range of the coefficients $B^\star:=\frac{c^\star_{\max}}{c^\star_{\min}}$,  where $c^\star_{\max}$ and $c^\star_{\min}$ denote the maximal and minimal modules of $\{c^\star_\ell\}_{\ell=1}^k$; 2) the normalized noise level $\gamma_0:=\sigma\sqrt{\frac{\log n}{n}}$;  3) the Noise-to-Signal Ratio $\gamma:=\gamma_0/c^\star_{\min}$ and 4) the regularization parameter $\lambda = 0.646X^\star  \gamma_0$ for some positive constant $X^\star$ to be determined later. Now we are ready to present our main result.

\begin{theorem}\label{thm:main}
Suppose we observe $2n+1$ noisy consecutive samples $y_\ell=x^\star_\ell+w_\ell$ of the signal~\eqref{eqn:true} or~\eqref{eqn:atom_true} with $w_\ell$  being i.i.d. complex Gaussian valuables of mean zero and variance $\sigma^2$.   We assume  $n\geq130$ and
\begin{align}
\Delta(T^\star)&\geq2.5009/n, \label{eqn:separation}\\
X^\star {B^\star}\gamma & \leq 10^{-3} \text{\ and\ }  B^\star /{X^\star }\leq 10^{-4}.  \label{eqn:snr}
\end{align}
Then with  probability at least $1- \frac{1}{n^2}$,   the optimal solution  of~\eqref{eqn:primal} has a decomposition
${\x}^{\mathrm{glob}}=\sum_{\ell=1}^k {c}^{\mathrm{glob}}_\ell \a({f}_\ell^{\mathrm{glob}})$ involving exactly $k$ atoms,  whose frequencies and coefficients,  when properly ordered, 
satisfy

\begin{align}
\max_{1\leq\ell\leq k} |c_\ell^\star| |{f}^{\mathrm{glob}}_\ell-f^\star_\ell| &\leq {0.4(X^\star+35.2) \gamma_0}/{n }, \label{eqn:freqbound}\\
\max_{1\leq\ell\leq k}|{c}^{\mathrm{glob}}_\ell-c^\star_\ell|&\leq (X^\star+35.2) \gamma_0 .\label{eqn:coefbound}
\end{align}

\end{theorem}

Several remarks on the conditions follow. Because of the weighting scheme we use in~\eqref{eqn:primal},  our choice of $\lambda$ differs from the standard one in~\cite{bhaskar2013atomic} by a factor $1/n$ and ensures that the weighted dual atomic norm of the noise,  $\|\mZ\w\|_\A^*$,  is less than $\lambda$ with high probability. For technical reasons,  our separation condition~\eqref{eqn:separation} is stronger compared with the previous works~\cite{Candes:2014br, Tang:2013gd, fernandez2013support, candes2013super,  Duval:2015gk, Denoyelle2017, morgenshtern2016super}\footnote{Note that our separation condition is a bit larger when comparing to these recent works in super-resolution,  while there are two other things to be considered. One thing is that most of these works require strong assumptions on the noise in their models (e.g.,   the noise is bounded),  while our work removes such assumptions and hence can deal with the more general Gaussian noise. To make this possible,  we have to develop a new proof strategy involving the two-step construction process of the dual certificate. Another thing is that although some prior works achieve small resolution limit (even comparable to the Relay diffraction limit~\cite{morgenshtern2016super}),  they study a different problem. For example, ~\cite{morgenshtern2016super} considers the signal denoising problem,  that is,   stable recovery of the whole signal $\x$ rather than the parameter estimation  (i.e.,  the source location recovery).  While the focus of our work is the accuracy of parameter estimation in Gaussian noise,  which might be more significant for practical applications such as Radar and single-molecule microscopy,  where precisely locating each target/point source is extremely important. Since the parameter estimation problem is much harder than the denoising problem,  we have to relax a bit the separation condition for ease of analysis.
}. The conditions~\eqref{eqn:snr} wrap several requirements on the problem parameters for the conclusions to hold: the dynamic range of the coefficients $B^\star$,  the Noise-to-Signal Ratio $\gamma$,  and the normalized noise $\gamma_0$ should all be small while the regularization parameter $\lambda$ should be large enough as measured by $X^\star$.

%\note{I added an expectation to $\|\mZ\w\|_{\A}^*$,  otherwise it doesn't make sense because $\|\mZ\w\|_{\A}^*$ is random? Agree? }
%
%\note{Gongguo,  you are right.}

It is worth noting that~\eqref{eqn:snr} implicitly imposes a strong assumption on the Noise-to-Signal Ratio
	\[\gamma\leq 10^{-7}/B^{\star 2}\]
	implying a sufficiently large $n$ (but still finite). For high-level ideas,  there might be two reasons to account for this phenomenon. One is that the problem of line spectral estimation is known to be sensitive to noise.  Another is inherently from our proof regime,  which makes the constants in Eq.~\eqref{eqn:snr} a bit conservative. More precisely,  the ultimate objective is to show the boundedness and interpolation property of the target polynomial (see Proposition~\ref{pro:bip} for more details). Our method is using an ``existing" dual polynomial in~\cite{Candes:2014br} satisfying this property and showing the distance between these two polynomials is sufficiently small. So,  we require the noise level to be small,  since we will see in Lemma~\ref{lem:fix2} that the noise level will influence this distance.

One more remark is that the quantity $35.2\gamma_0$ in our results is  related to the expected dual atomic norm of the weighted Gaussian noise $ \E\|\mZ\w\|_{\A}^*$. By noting the definition $\lambda=0.646X^\star\gamma_0$,  we can rewrite the error bounds~\eqref{eqn:freqbound} and~\eqref{eqn:coefbound} in a more concise way:
\begin{align}
\max_{1\leq\ell\leq k} |c_\ell^\star| |{f}^{\mathrm{glob}}_\ell-f^\star_\ell| &= O\left(\lambda + \E\|\mZ\w\|_{\A}^*\right)/n, \label{eqn:freqbound1}\\
\max_{1\leq\ell\leq k}|{c}^{\mathrm{glob}}_\ell-c^\star_\ell|&= O\left(\lambda + \E\|\mZ\w\|_{\A}^*\right) .\label{eqn:coefbound1}
\end{align}
Eq.~\eqref{eqn:freqbound1} and~\eqref{eqn:coefbound1} imply that the error bounds are determined jointly by the regularization parameter $\lambda$ and the expected dual atomic norm of the weighted Gaussian noise $ \E\|\mZ\w\|_{\A}^*$. Since the regularization parameter $\lambda$ has the same order as $\E\|\mZ\w\|_\A$,   the estimated frequencies and coefficients are guaranteed to have errors of orders $O\left(\E\|\mZ\w\|_{\A}^*/n\right)$ and $O\left(\E\|\mZ\w\|_{\A}^*\right)$,  respectively. 
Remarkably,  using atomic dual norm strategy allows us to deal with the Gaussian noise,  while most prior works~\cite{fernandez2013support, candes2013super,  Duval:2015gk, Denoyelle2017} in approximate support recovery have to build their theoretical foundations on the bounded-noise assumption,  which dramatically narrow down the applications. 

Now we summarize the above comparisons of our result with those state-of-the-art modern support recovery methods in the Table \ref{tab:comp:modern}.
\begin{table}[h!t]
	\centering
	\begin{tabular}{lcccll}
		\toprule 
		\blk{0.15}{\textbf{Paper}} 
		& \blk{0.13}{\textbf{Bounded\\Noise}}	
		& \blk{0.13}{\textbf{Positive\\Measure}}
		&\blk{0.14}{\textbf{Support\\Condition}}
		&\blk{0.1}{\textbf{SNR}}
		&\blk{0.13}{\textbf{Support\\~Recovery}}    
		\\
		\midrule 
		{\cite[Theorem 1.5]{Candes:2014br}} &{Yes}&  {No}&
		$\Delta\ge\frac{2}{n}$&\quad Finite&\quad~~{None}    
		\\[0.4ex]
		{\cite[Theorem 1.2]{candes2013super}} &{No}&{No}& $\Delta\ge\frac{2}{n}$&\quad Finite&\quad~~{None}    
		\\[0.4ex]
		{\cite[Theorem 1]{morgenshtern2016super}} &{No}	&{Yes}& {RRC}&\quad Finite&\quad~~{None}    
		\\[0.4ex]
		{\cite[Theorem 1.2]{fernandez2013support}} &{Yes}&{No}& $\Delta\ge\frac{2}{n}$&\quad Finite&\quad~~{Exist}    
		\\[0.4ex]
		{\cite[Theorem 2]{Tang:2013gd}} &{No}&{No}& $\Delta\ge\frac{2}{n}$&\quad Finite&\quad~~{Exist}    
		\\[0.4ex]
		{\cite[Theorem 2]{Duval:2015gk}}&{Yes}&{No}&NDSC&\quad Infinite&\quad~~{Unique}    
		\\[0.4ex]
		{\cite[Theorem 2]{Denoyelle2017}}&{Yes}&{Yes}&NDSC&\quad Infinite&\quad~~{Unique}    
		\\[0.4ex]
		\rowcolor[gray]{0.9}~{Theorem 1.2}&{No}&{No}&$\Delta\ge\frac{2.5009}{n}$&\quad Finite&\quad~~{Unique} 
		\\
		\bottomrule
	\end{tabular} 
	\caption {Comparison with other modern line spectral estimation/super-resolution methods. The \emph{Positive Measure} column refers to whether the result requires the ground-truth measure to be positive. RRC is short for Rayleigh Regularity condition \cite[Definition 1.1]{morgenshtern2016super},   which generalizes the standard separation condition to clustered support. NDSC stands for the  non-degenerate source condition \cite[Definition 5]{Duval:2015gk}. In the \emph{Support Recovery} column, \emph{None} indicates that the work considers signal recovery instead of support recovery; \emph{Existence} means that the work shows the existence of at least one recovered parameter around each ground-true parameter,  but fails to theoretically eliminate the possibility of spurious recovered parameters; \emph{Uniqueness} shows that around each true parameter there is one and only one recovered parameter.}  \label{tab:comp:modern}
\end{table}

%	 indicates whether the final result addresses the  support recovery or only considers the signal recovery. Among those works  considering the support recovery,  \cite{fernandez2013support, Tang:2013gd} fail to eliminate the existence of multiple recovered supports around certain true support,  while \cite{Duval:2015gk, Denoyelle2017} and this work can fix this issue but at some cost of additional assumptions. 
%	 For example,  \cite{Duval:2015gk, Denoyelle2017} further requires the noise level to be asymptotically small (i.e. an infinite SNR) and the non-degenerate source condition on the supports,  while this work dramatically relaxes these assumptions by using a slightly stronger separation condition.

Finally,  our proof for Theorem~\ref{thm:main} also reveals the connection between the atomic norm minimization~\eqref{eqn:primal} and the following $\ell_1$-norm regularized,  nonlinear and nonconvex least-squares  program:
\begin{align}\label{eqn:pdw}
\operatorname*{minimize}_{\f,  \c} \frac{1}{2}\|\mA(\f)\c - \y\|_{\mZ}^2 + \lambda \|\c\|_1, 
\end{align}
where $\f := [f_1,  \ldots,  f_{k}]^T, ~\c := [c_1,  \ldots,  c_{k}]^T$,  and $\mA(\f):=[\a(f_1), \ldots, \a(f_{k})]$. The program~\eqref{eqn:pdw} is highly nonconvex,  with numerous local minima and saddle points,  so solving it to global optimality is very difficult. Our analysis shows that,  under the conditions of~Theorem \ref{thm:main},  the convex program~\eqref{eqn:primal} shares the same global optimum as the nonconvex program~\eqref{eqn:pdw},  implying that the atomic norm minimization provides a new convex way to solve the nonconvex program to global optimality. We summarize the result in the following corollary,   with the formal proof listed in Appendix~\ref{sec:connection}.
\begin{corollary}\label{cor:main2}
Under the same setup as in Theorem~\ref{thm:main},  with  probability at least $1- \frac{1}{n^2}$,  the frequencies and coefficients estimated by the  atomic norm regularized minimization~\eqref{eqn:primal} constitute a global optimum of the $\ell_1$-regularized nonlinear least-squares program~\eqref{eqn:pdw}.
\end{corollary}

\section{Prior Art and Inspirations}\label{sec:prior:art}

Classical line spectral estimation techniques can be broadly classified into two camps: non-parametric and parametric methods. Non-parametric methods are mainly based on Fourier analysis~\cite{kay1988modern, stoica1997introduction}. Such approaches have low computational complexities and no need for signal models. These methods have limited frequency resolution due to spectral leakage. Parametric methods,  however,  can achieve high resolution for parameter estimation. For example,  Prony's method based on polynomial root-finding~\cite{prony1795,  kahn1992consistency} can resolve arbitrarily close frequencies in the noiseless setting. Yet this method is highly sensitive to noise and would fail even in the small noise regime. As stable versions of Prony's method,  the subspace methods recast the noise-sensitive polynomial root-finding problem into more robust matrix eigenvalue problems. For instance,  the matrix pencil method~\cite{Hua:1990kx} arranges the observations into a matrix pencil whose generalized eigenvalues and eigenvectors contain information about the frequencies; the MUSIC algorithm~\cite{music} and the ESPRIT method~\cite{esprit} decompose the autocorrelation matrix into noise-subspace and signal subspace using eigenvalue decomposition and extract frequency estimates from the signal subspace. Both algorithms were shown to achieve CRB asymptotically~\cite{Stoica:1989dn,  fri} when the signal length $2n+1$ and the number of snapshots approach infinite. However,  these classical methods are not efficient (i.e.,  approaching the CRB) even with an infinite number of snapshots,  as long as the signal length is finite. Also,  all classical parametric methods require knowledge of the model order.

Modern convex optimization based methods formulate line spectral estimation as a linear inverse problem and exploit signal sparsity using $\ell_1$-type regularizations. Such methods are modular,  robust,  and do not require knowledge of model orders. To apply the $\ell_1$ regularization techniques,  the continuous frequency domain is divided into a grid of discrete frequencies. When the true frequencies fall onto the discrete Fourier grid,  work in compressive sensing guarantees optimal recovery performance~\cite{donoho2006compressed, candes2006compressive,  baraniuk2007compressive}. When the frequencies do not fall onto the Fourier grid,  however,  the performance of $\ell_1$ minimization degrades significantly due to basis mismatch~\cite{chi2011sensitivity}. The basis mismatch issue can be mitigated by employing finer grids~\cite{tang2013sparse,  duval2015sparse},  which unfortunately often leads to numerical instability.

Atomic norm regularization avoids basis mismatch by enforcing sparsity directly in the continuous frequency domain. Given a set of atoms,  possibly indexed by continuous parameters,  one constructs an atomic norm in a principled way as a generalization of the $\ell_1$-norm to promote signals with parsimonious representations. Using the notion of descent cones,  the authors of~\cite{Chandrasekaran:2010hl} argued that the atomic norm is the best possible convex proxy for recovering sparse models. For the special line spectral estimation problem,  where the atomic norm is induced by the set of parameterized complex exponentials,  atomic regularizations have been shown to achieve optimal performance for several signal processing tasks. For instance,  atomic norm minimization recovers a spectrally sparse signal from a minimal number of random signal samples~\cite{Tang:2013fo},  identifies and removes a maximal number of outliers~\cite{Tang:2014outlier,  fernandez2016demixing},  and performs denoising with an error approaching the minimax rate~\cite{Tang:2013gd}.
When multiple measurement vectors are available,  a method of exploiting the joint sparsity
pattern of different signals to further improve estimation accuracy is proposed in~\cite{li2016off,yang2014exact,li2018atomic}.
All these works draw inspirations from the dual polynomial construction strategy developed in the pioneer work~\cite{Candes:2014br}. This paper adds to this line of work by showing that the atomic framework produces optimal noisy frequency estimators.

Several closely related works also studied conditions for approximate support recovery from noisy observations. The work~\cite{fernandez2013support} developed error bounds on spectral support recovery for bounded noise. In~\cite{Tang:2013gd},  the authors derived suboptimal bounds for the Gaussian noise model. In~\cite{azais2015spike},  the authors extended this line of research to general measurement schemes beyond Fourier samples using the Beurling-LASSO (B-LASSO) program. The B-LASSO program,  which minimizes a least-squares term plus the measure total variation norm,  is mathematically equivalent to the atomic norm formulation. All these works~\cite{azais2015spike, Tang:2013gd,  fernandez2013support} cannot guarantee the recovery of exactly one frequency in each neighborhood of the true frequencies. In this regard,  the work by Duval and Peyr\'{e}~\cite{Duval:2015gk} showed that as long as the SNR is large enough and the sources are well-separated and satisfy a \emph{non-degenerate source condition},  then total variation norm regularization can recover the correct number of the Diracs with both the coefficient error and the frequency error scale as the $\ell_2$ norm of the noise. Compared with their work,  our result uses the (weighted) dual atomic norm of the noise in place of the $\ell_2$ norm,  which differ by order of $\sqrt{n}$,  allowing our bound to match the CRB up to a logarithmic factor. In addition,   their work relies on a \emph{non-degenerate source condition}~\cite[Definition 5]{Duval:2015gk} that is not proven to hold in the spectral super-resolution setting. In this sense,  the present paper is the first to rigorously establish that in a high SNR regime this approach yields the right number of frequencies.
Further our proof technique based on the primal-dual witness construction is also very different from that employed in~\cite{Duval:2015gk} based on a perturbation analysis of the dual certificate in the noise-free case. In particular,  our analysis reveals the connection between the convex approach and a natural nonlinear least-squares method for spectral estimation. More recently, ~\cite{Denoyelle2017} studies the support recovery for positive measures. For a comparison,  there are several major  differences worth remarking here: 1) in~\cite{Denoyelle2017}  more emphasis is put on the asymptotic analysis,  while the presented work instead deals with non-asymptotic settings with finite signal length; 2)~\cite{Denoyelle2017} requires the underlying noise to have finite $\ell_2$ norm,  which severely restricts the scope of noises satisfying such a property,  excluding the well-known and most common Gaussian noise,  while the presented results allow the underlying noise to be Gaussian; 3)  in addition to requiring a sufficiently large signal-to-noise ratio,  the main result in~\cite{Denoyelle2017} also  relies on a \emph{non-degenerate source condition} that is not proven to hold in the spectral super-resolution setting.

%By exchanging the temporal and frequency domains,  our work is similar to the work~\cite{duval2015exact} which consider the problem of the sparse spikes deconvolution over the space of measures and they obtain that if 1) the regularization parameter $\lambda$ and $\ell_2$ norm of the Gaussian noise satisfy that $\{0\leq\lambda\leq\lambda_0, \ \text{and }\|\w\|_2\leq\alpha\lambda\}$ for $\alpha>0$ and $\lambda>0$ and 2) the target measure satisfies the non-degenerate source condition,  then TV regularization can recover the same number of the Diracs with both the coefficient error and frequency error at the order of $O(\|\w\|_2).$ Here the non-degenerate source condition is corresponding to our separation condition~\eqref{eqn:separation} for the target sources and the requirement for the noise level and $\lambda$ is corresponding to our signal-to-noise ratio condition~\eqref{eqn:snr}. It is known that the $\ell_2$ norm of the Gaussian noise $\w\in(\zero, \sigma^2\eye)$ is at the order of $O(\sqrt n\sigma)$~\cite[Proposition 8.1]{foucart2013mathematical},  which means the error bounds in~\cite{duval2015exact} are at the order of $O(\sqrt n\sigma)$ for both frequency and coefficient. Therefore,  compared to the prior work~\cite{duval2015exact},  our main contribution is improving their error bounds from $O(\sqrt n\sigma)$ to $O(\frac{\sqrt{\log n}}{n^{3/2}}\sigma)$ for frequency estimation and from $O(\sqrt n\sigma)$ to $O(\frac{\sqrt{\log n}}{n}\sigma)$ for coefficient recovery. Such improved error bounds then match the CRB up to a logarithmic factor.

\section{Proof by Primal-Dual Witness Construction}\label{sec:proof}

Duality plays an important role in understanding atomic norm regularized line spectral estimation. Standard Lagrangian analysis shows that the dual problem of
\eqref{eqn:primal} has the following form:
\begin{align}\label{eqn:dual}
{\q}^{\mathrm{glob}}=\operatorname*{argmax}_{\q}&~\frac{1}{2}\|\y\|_{\mZ}^2-\frac{1}{2}\|\y-\lambda\q\|_{\mZ}^2\nn\\
\st&~\|\mZ\q\|_\A^* \leq 1.
\end{align}
 The complex trigonometric polynomial $Q(f):=\a(f)^H\mZ\q$ corresponding to a dual feasible solution $\q$ is called a dual polynomial. The dual polynomial associated with the unique dual optimal solution  ${Q}^{\mathrm{glob}}(f):=\a(f)^H\mZ{\q}^{\mathrm{glob}}$ certifies the optimality of the unique primal optimal solution ${\x}^{\mathrm{glob}}$,  and vice versa. The uniqueness of primal and dual optimal solutions is a consequence of the strong convexity of the objective functions of~\eqref{eqn:primal} and~\eqref{eqn:dual},  respectively. In particular,  the primal-dual optimal solutions are related by ${\q}^{\mathrm{glob}} = (\y - {\x}^{\mathrm{glob}})/\lambda$. We summarize these in the following proposition,  with the proof given in Appendix~\ref{sec:proof:optimalitycondition}:
\begin{proposition}\label{pro:bip}
Let the decomposition $\hat{\x} = \sum_{\ell=1}^{\hat{k}} \hat{c}_\ell \a(\hat{f}_\ell)$ with distinct frequencies $\hat{T}=\{\hat{f}_\ell\} \subset \TT$ and nonzero coefficients $\{\hat{c}_\ell\}$ and set $\hat{\q} = ({\y-\hat{\x}})/{\lambda}$. Suppose the corresponding dual polynomial $\hat{Q} (f) = \a(f)^H\mZ\hat\q$ satisfies the following {\tmem{Bounded Interpolation Property (BIP)}}:
  \begin{align*}
\hat{Q} (\hat{f}_\ell) & =  \sign(\hat{c}_\ell), \ell=1,  \ldots , \hat{k}\ \ ( \text{\em Interpolation} );
\\
|\hat{Q} (f)|  & < 1 {} ,   \forall f\notin \hat{T}\ \  (\text{\em Boundedness} );
  \end{align*}
  then $\hat{\x}$ and $\hat{\q}$ are the unique primal-dual optimal solutions to~\eqref{eqn:primal} and~\eqref{eqn:dual},  that is,  $\hat{\x} = {\x}^{\mathrm{glob}}$ and $\hat{\q} = {\q}^{\mathrm{glob}}$. Here the operation $\sign(c) := c/|c|$ for a nonzero complex number and applies entry-wise to a vector.
\end{proposition}

Proposition~\ref{pro:bip} gives a way to extract the frequencies from the dual optimal solution -- one can simply identify the frequencies where the dual polynomial corresponding to the dual optimal solution achieves magnitude $1$.  The uniqueness of the dual solution for~\eqref{eqn:primal} makes the construction of a dual certificate much harder compared with the line spectral signal completion problem~\cite{Tang:2013fo} and demixing problem~\cite{Tang:2014outlier,  fernandez2016demixing}. For the latter two problems,  while the primal optimal solution is unique,  the dual optimal solutions are non-unique. One usually chooses one dual solution that is easier to analyze (e.g.,  the one with minimal energy). For the support recovery problem,  we need to simultaneously construct the primal and dual solutions,  which witness the optimality of each other. In the compressive sensing literature,  this construction process is called the \emph{primal-dual witness construction}~\cite{wainwright2009sharp}. In sparse recovery problems,  a candidate primal solution is relatively easy to find,  since when the noise is relatively small,  the support of the recovered signal would not change. So one only needs to solve a LASSO problem restricted to the true support to determine the candidate coefficients,  as was done in~\cite{wainwright2009sharp}. For the optimization~\eqref{eqn:primal},  due to the continuous nature of the atoms,  even a bit of noise would drive the support away from the true one. So to construct a candidate primal solution (hence a candidate dual solution),  we need to simultaneously seek for the candidate support $\{\hat{f}_\ell\}$ and the candidate coefficients $\{\hat{c}_\ell\}$.

\subsection{Proof Outline}\label{sec:outline}

We use the $\ell_1$-regularized,  nonlinear and nonconvex program~\eqref{eqn:pdw},  which we copy below,  to find plausible candidates for $\{\hat{f}_\ell\}$ and $\{\hat{c}_\ell\}$:
\begin{align*}
\operatorname*{minimize}_{\f,  \c} \frac{1}{2}\|\mA(\f)\c - \y\|_{\mZ}^2 + \lambda \|\c\|_1, 
\end{align*}
where $\f = [f_1,  \ldots,  f_k]^T, ~\c = [c_1,  \ldots,  c_k]^T$ and $\mA(\f)=[\a(f_1), \ldots, \a(f_k)]$. Note that we have effectively fixed the number of estimated frequencies $\hat{k}$ in Proposition~\ref{pro:bip} to be $k$.  But unlike in compressive sensing we cannot fix $\f = \f^\star$ to solve for $\c$ only as was done in~\cite{wainwright2009sharp}. The program~\eqref{eqn:pdw} is highly nonconvex,  with numerous local minima,  local maxima,  and saddle points. So solving it to global optimality is hard even in theory. We are primarily interested in its local minimum $(\{\hat{f}_\ell\},  \{\hat{c}_\ell\})$ in a neighborhood of the true frequencies and coefficients $(\f^\star,  \c^\star)$. To find this local minimum,  we will run gradient descent to~\eqref{eqn:pdw} using $(\f^\star,  \c^\star)$ as initialization. We will argue that under conditions presented in Theorem~\ref{thm:main},  each $\hat{f}_\ell$ and $\hat{c}_\ell$ stay close to $f^\star_\ell$ and $c^\star_\ell$ as given in~\eqref{eqn:freqbound} and~\eqref{eqn:coefbound},  respectively. The major tool we use is the contraction mapping theorem. As shown in Corollary~\ref{cor:main2},  the local minimum found in this manner is actually a global optimum of~\eqref{eqn:pdw}. 

The rest of arguments consist of showing that $\hat{\x} = \sum_{\ell =1}^{k} \hat{c}_\ell \a(\hat{f}_\ell)$ with $\{\hat{f}_\ell\}$ and $\{\hat{c}_\ell\}$ constructed as described above satisfies the Bounded Interpolation Property of Proposition~\ref{pro:bip}. 
The Interpolation property is automatically satisfied due to the construction process and the main challenge is to show the Boundedness property $|\hat{Q}(f)| < 1,  \forall f\notin \hat{T}$. 
The harder part is showing the Boundedness property. For ease of interpretation we first collect the definitions of the most important variables  that will be used throughout the proof,  and then introduce the {\em main logic} and the {\em two-step construction process} of the proof.

\begin{center}
	\begin{tabular}{ll}
		\toprule 
\blk{0.08}{\textbf{Symbol}} & \multicolumn{1}{c}{\bf Definition}		
\\ \midrule
$(\f^\lambda,  \c^\lambda)$ &The local minima of   $\operatorname*{minimize}_{\f,  \c} \frac{1}{2}\|\mA(\f)\c - \x^\star\|_{\mZ}^2 + \lambda \|\c\|_1$ that is  closest to $(\f^\star, \c^\star)$
\\[0.1ex]
$(\hat\f, \hat\c)$ &The local minima of   $\operatorname*{minimize}_{\f,  \c} \frac{1}{2}\|\mA(\f)\c - \y\|_{\mZ}^2 + \lambda \|\c\|_1$ that is closest to $(\f^\lambda,  \c^\lambda)$
\\[0.1ex]
$\x^\lambda$ &The primal solution  defined by the local minima $(\f^\lambda,  \c^\lambda)$ via $\x^\lambda:=\sum_{\ell=1}^k {c}^\lambda_\ell \a({f}^\lambda_\ell)$
\\[0.1ex]	
$\hat\x$ &The primal solution defined by the local minima $(\hat\f,  \hat\c)$ via $\hat\x:=\sum_{\ell=1}^k \hat{c}_\ell \a(\hat{f}_\ell)$
\\[0.1ex]
$\q^\lambda$ &The dual solution corresponding to the primal solution $\x^\lambda$,  that is,  $\q^\lambda:=({\x^\star-\x^\lambda})/{\lambda}$
\\[0.1ex]	
$\hat\q$ &The dual solution corresponding to the primal solution $\hat\x$,  that is,  $\hat\q:=({\y-\hat\x})/{\lambda}$
\\[0.1ex]	
$\q^\star$& $\q^\star:=\lim\limits_{\lambda\to 0} \q^\lambda$,  satisfying the Boundedness and Interpolation property for $(\f^\star, \c^\star)$
\\
\bottomrule
\end{tabular} 
\end{center}

\noindent{\bf Main Logic:}    
Firstly,  identifying that $Q^\star(f):=\a(f)^H\mZ\q^\star$ satisfies the Boundedness property  with some similar arguments used in~\cite{Candes:2014br}.  Secondly,  establishing that $\hat\q$ and $\q^\star$ are sufficiently close (so are $\hat{Q}(f):=\a(f)^H\mZ\hat\q$ and $Q^\star(f)=\a(f)^H\mZ\q^\star$). Therefore $\hat{Q}(f)$ also satisfies the Boundedness property. It turns out that directly showing the closeness of  $\hat\q$ and $\q^\star$  is difficult. That is why we introduce the intermediate dual variable $\q^\lambda$ and use the {\em two-step construction process},  i.e., first showing $\q^\star$ is close to $\q^\lambda$ and then showing $\q^\lambda$ is close to $\hat\q$.

\medskip
\noindent{\bf Two-step Construction Process:}
We will first find a local minimum   $(\f^\lambda,  \c^\lambda)$ of $\frac{1}{2}\|\mA(\f)\c - \x^\star\|_{\mZ}^2 + \lambda \|\c\|_1$ around $(\f^\star,  \c^\star)$,  where one should note we replaced the noisy signal $\y$ in~\eqref{eqn:pdw} by the noise-free signal $\x^\star$. We will then run gradient descent to~\eqref{eqn:pdw} using $(\f^\lambda,  \c^\lambda)$ as initialization. The intermediate quantities $(\f^\lambda,  \c^\lambda)$ will serve as a bridge between $(\f^\star,  \c^\star)$ and $(\hat{\f}, \hat{\c})$ to make the proof easier.  The key is noting that $\hat{Q}(f) = \a(f)^H\mZ\hat\q$ is close to $Q^\lambda (f) = \a(f)^H\mZ\q^\lambda$,  where $\q^\lambda = (\x^\star - \x^\lambda)/\lambda$ and $\x^\lambda = \sum_{\ell = 1}^k c_\ell^\lambda \a(f^\lambda_\ell)$,  and $Q^\lambda (f)$ is close to $Q^\star (f) = \a(f)^H\mZ\q^\star$. Here $\q^\star = \lim_{\lambda \rightarrow 0} \q^\lambda$ is a dual certificate used to certify the atomic decomposition of $\x^\star$. The former claim can be showed using the closeness of $(\f^\lambda,  \c^\lambda)$ and $(\hat{\f}, \hat{\c})$. The later claim,  however,  must take advantage of the fact that $\q^\star = \lim_{\lambda \rightarrow 0} \q^\lambda = - \frac{\mathrm{d}}{\mathrm{d} \lambda} \x^\lambda|_{\lambda = 0}$ and apply the triangle inequality to
\begin{align*}
Q^\lambda (f) - Q^\star(f) = \frac{1}{\lambda}\int_0^\lambda  \a(f)^H\mZ\left( \frac{\mathrm{d}}{\mathrm{d} t}\x^0- \frac{\mathrm{d}}{\mathrm{d} t}\x^t \right)\mathrm{d}t, 
\end{align*}
where $\frac{\mathrm{d}}{\mathrm{d} t}\x^0 = \lim_{\lambda \rightarrow 0} \frac{\mathrm{d}}{\mathrm{d} t}\x^\lambda := \frac{\mathrm{d}}{\mathrm{d} t}\x^\star$.
The closeness of $(\f^\lambda,  \c^\lambda)$ and $({\f}^\star, {\c}^\star)$ ensures that the derivatives in the integrand are also close. Finally,  we exploit  the properties of $Q^\star(f)$ which are similar to those established in~\cite{Candes:2014br} to complete the proof.

\subsection{A Formal Proof: Applying the Contraction Mapping Theorem}
\begin{theorem}[Contraction Mapping Theorem]\label{thm:fix} Given a Banach space $\mathcal{B}$ equipped with a norm $\|\cdot\|$,  a bounded closed set $\NN\subset\mathcal{B}$ and a map $\Theta:\NN\rightarrow \mathcal{B}$,  if $\Theta(\NN)\subset\NN$ (the non-escaping property) and there exists $\rho\in(0, 1)$ such that $\|\Theta(\v)-\Theta(\w)\|\leq\rho\|\v-\w\|$ for each $\v, \w\in\NN$ (the contraction property),  then there exists a unique $\v^\star\in\NN$ such that $\Theta(\v^\star)=\v^\star.$
\end{theorem}
This classical result helps to find a candidate solution for the construction of a valid dual certificate. To see this we first choose the bounded closed set $\NN$ to be a small region around the target joint frequency-coefficient vector $\btheta^\star:=(\f^\star, \u^\star, \v^\star)$ (where $\u^\star$ and $\v^\star$ denote respectively the real and imaginary parts of $\c^\star$). Let the fixed point map $\Theta$ be the gradient map of~\eqref{eqn:pdw}. The key is to determine the size of $\NN$ in which the non-escaping and the contraction properties of the fixed point map $\Theta$ hold. Then,  the contraction mapping theorem implies that iteratively performing the gradient map $\Theta$ from any initial point in $\NN$ would produce a candidate solution that still lies in $\NN$ (by the non-escaping property) and hence is close to $\btheta^\star$ (since $\NN$ is small). Finally relating the fixed point equation to the BIP property shows that such a candidate solution generates a valid dual certificate.

In order to apply the contraction mapping theorem to our problem,  we choose the norm in Theorem~\ref{thm:fix} to be a weighted $\ell_\infty$ norm $\|\cdot\|_{\hinfty}$ given by $\|(\f, \u, \v)\|_\hinfty:= \|(\mS\f, \u, \v)\|_{{\infty}}$ with $\mS:=\sqrt{|K''(0)|}\diag(|\c^\star|)$ and $K(\cdot)$ is the Jackson kernel (refer to Appendix~\ref{sec:A} for an introduction). This weighted $\ell_\infty$ norm is used as a metric function to define the neighborhood $\NN$ around $\btheta^\star$. The choice of the weighting matrix $\mS$ ensures that the larger a coefficient $c^\star_i$ is,  the smaller the neighborhood in the direction of the frequency $f_i$. In addition,  since $\sqrt{|K''(0)|}$ is of order $O(n)$,  the frequency neighborhood is smaller than the coefficient neighborhood by the same order. Next,  we choose the fixed point map $\Theta$ to be a weighted gradient map of~\eqref{eqn:pdw}
\begin{align}
{\Theta}(\btheta)  := \btheta-\W^\star\nabla\left(\frac{1}{2}\|\mA(\f)\c-\y\|_{\mZ}^2+\lambda\|\c\|_1\right)\label{eqn:gradient_descent}, 
\end{align}
where the gradient $\nabla$ is taken with respect to the parameter $\btheta = (\f,  \u,  \v)$ and the weighting matrix
\begin{align}\label{eqn:Wstar}
\W^\star=
\begin{bmatrix}
\mS^{-2} &&\\
&\eye_k&\\
&&\eye_k
\end{bmatrix}.
\end{align}
Scaling the gradient vector by $\W^\star$ ensures that the Jacobian matrix of the second term in~\eqref{eqn:gradient_descent} is close to the identity matrix,  which makes it easier to show the contraction property of ${\Theta}$.

\subsubsection{Two-step Construction Process}
As discussed in Section~\ref{sec:outline},  we divide the construction process into two steps. We first analyze the fixed point map $\Theta^\lambda$ obtained by replacing the noisy observation vector $\y$ in~\eqref{eqn:gradient_descent} by the noise-free signal $\x^\star$. We determine a region around $\btheta^\star$,  say $\NN^\star$,  such that both the contraction and non-escaping properties of $\Theta^\lambda$ are satisfied in $\NN^\star$. Then by the contraction mapping theorem,  iterating the gradient map $\Theta^\lambda$ in $\NN^\star$ initialized by $\btheta^\star$ generates a unique fixed point $\btheta^\lambda:=(\f^\lambda, \u^\lambda, \v^\lambda)$. These results are summarized in the following lemma:

\begin{lemma}[The First Fixed Point Map]\label{lem:fix1}
Let the first fixed point map be the weighted gradient map of the nonconvex program~\eqref{eqn:pdw} with the noisy signal $\y$ replaced by the noise-free signal $\x^\star$:
\begin{align}\label{eqn:Map:1}
\Theta^\lambda(\btheta) & := \btheta-\W^\star\nabla\left(\frac{1}{2}\|\mA(\f)\c-\x^\star\|_{\mZ}^2+\lambda\|\c\|_1\right), 
\end{align}
where the gradient $\nabla$ is taken with respect to the parameter $\btheta = (\f,  \u,  \v)$. Let the regularization parameter $\lambda$ vary in $[0,  0.646X^\star  \gamma_0]$. Define a neighborhood $\NN^\star:=\left\{\btheta:\|\btheta-\btheta^\star\|_{\hinfty}\leq {X^\star}\gamma_0/{\sqrt2}\right\}$. Suppose that the separation condition~\eqref{eqn:separation} and the SNR condition~\eqref{eqn:snr} hold. Then the map $\Theta^\lambda$ has a unique fixed point $\btheta^\lambda\in \NN^\star$ satisfying $\Theta^\lambda(\btheta^\lambda) = \btheta^\lambda$. Furthermore,  according to the implicit function theorem,  $\btheta^\lambda$ is a continuously differentiable function of $\lambda$ whose derivative is given by
\begin{align}\label{eqn:thetadiff}
\frac{\mathrm{d}}{\mathrm{d} \lambda} \btheta^\lambda = - (\nabla^2 G^\lambda(\btheta^\lambda))^{-1}\frac{\partial }{\partial \lambda}\nabla G^\lambda(\btheta^\lambda).
\end{align}
Finally,  when $\lambda$ turns to zero,  the fixed point $\btheta^\lambda$ converges to $\btheta^\star$,  i.e.,  $\lim_{\lambda \rightarrow 0} \btheta^\lambda=\btheta^\star$, and therefore $\lim_{\lambda \rightarrow 0} \x^\lambda=\x^\star$.
\end{lemma}
\begin{proof}[Proof of Lemma~\ref{lem:fix1}]
See Appendix~\ref{sec:D}.
\end{proof}

We now turn to the gradient map $\Theta$ in~\eqref{eqn:gradient_descent} defined in a region $\NN^\lambda$ around $\btheta^\lambda$. Similar to the first step,  we show the contraction and non-escaping properties of $\Theta$ in $\NN^\lambda$,  which imply that iterating the gradient map $\Theta$ initialized by $\btheta^\lambda$ produces a unique fixed point $\hat{\btheta}:=(\hat\f, \hat\u, \hat\v)$.

\begin{lemma}[The Second Fixed Point Map]\label{lem:fix2}
Let the second fixed point map be the weighted gradient map of the nonconvex program~\eqref{eqn:pdw}:
\begin{align}\label{eqn:Map:2}
\Theta(\btheta) & = \btheta-\W^\star\nabla\left(\frac{1}{2}\|\mA(\f)\c-\y\|_{\mZ}^2+\lambda\|\c\|_1\right)
\end{align}
and the region $\NN^\lambda:=\left\{\btheta:\|\btheta-\btheta^\lambda\|_{\hinfty}\leq 35.2\gamma_0/{\sqrt2} \right\}$. Set the regularization parameter $\lambda$ as $ 0.646X^\star  \gamma_0$ in~\eqref{eqn:Map:2}. Suppose that the separation condition~\eqref{eqn:separation} and the SNR condition~\eqref{eqn:snr} hold. Then with probability at least $1- \frac{1}{n^2}$,   $\Theta(\btheta)$ has a unique fixed point $\hat{\btheta}$ living in $\NN^\lambda$.
\end{lemma}
\begin{proof}[Proof of Lemma~\ref{lem:fix2}]
See Appendix~\ref{sec:E}.
\end{proof}

The radius of the second contraction region $\N^\lambda$ is determined by a high probability bound on the dual atomic norm of the Gaussian noise and ensures that $\N^\lambda$ is a non-escaping set for $\Theta(\btheta)$. So far,  we have identified the neighborhoods where the two fixed points $\btheta^\lambda$ and $\hat{\btheta}$ live in,  which is the key to show the validity of the dual certificates later. Figure~\ref{fig:two:regions} illustrates the main results of Lemma~\ref{lem:fix1} and Lemma~\ref{lem:fix2}.

\begin{figure}[h!t]
  \centering
  \includegraphics[width=0.35\textwidth]{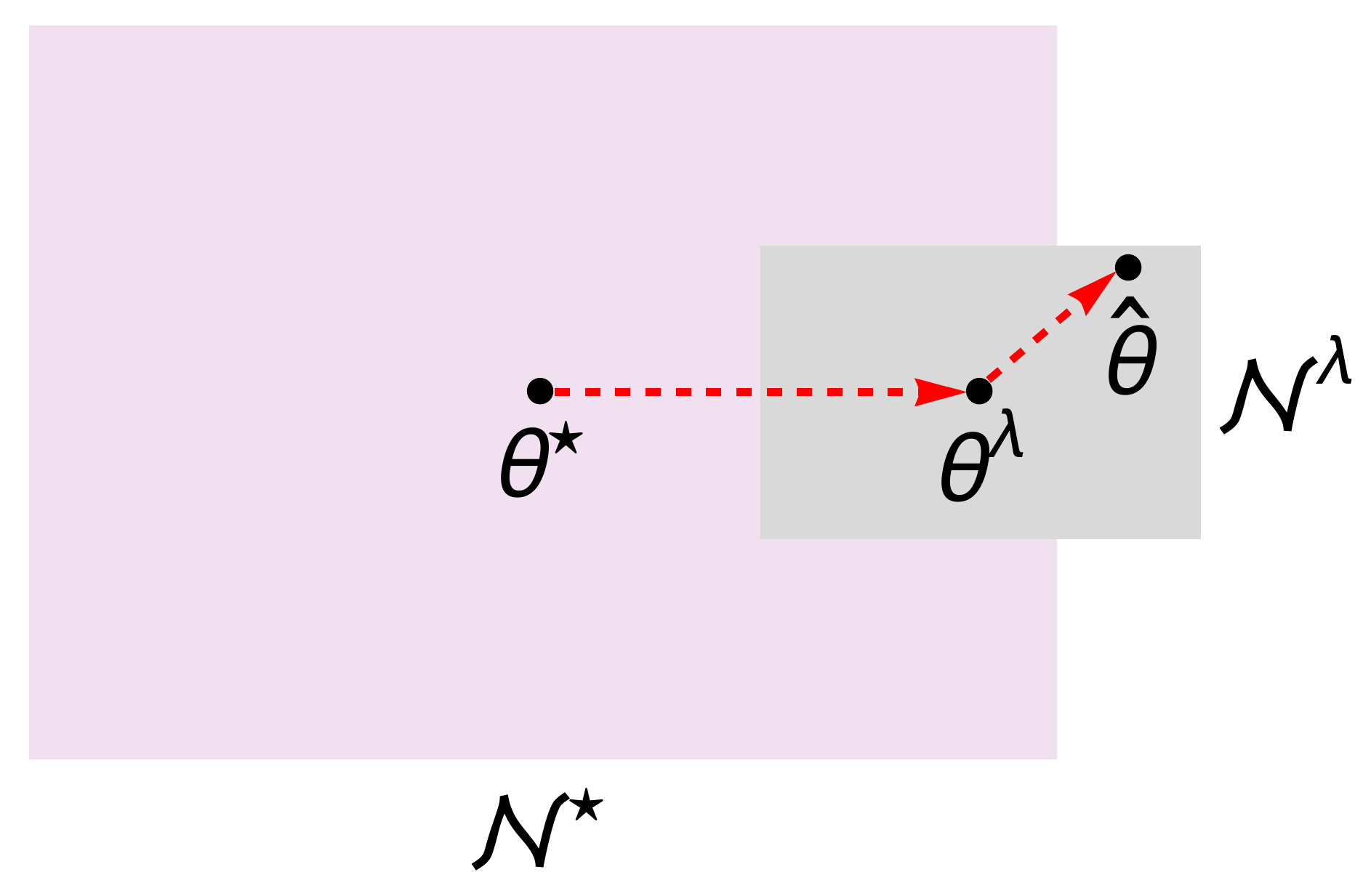}\\
  \caption{Use the true parameter vector $\btheta^\star$ as an initialization and run the first weighted gradient map~\eqref{eqn:Map:1} to obtain the first fixed point $\btheta^\lambda\in\N^\star$. Run the second weighted gradient map~\eqref{eqn:Map:2} initialized by $\btheta^\lambda$ to get the second fixed point $\hat{\btheta}\in\N^\lambda$.  The closeness of $\hat{\btheta}$ and $\btheta^\star$ is determined by the sizes of the two neighborhoods $\N^\star$ and $\N^\lambda$,  whose precise forms are given in Lemmas~\ref{lem:fix1} and~\ref{lem:fix2},  respectively.
   }\label{fig:two:regions}
\end{figure}

\paragraph{Road Map:} Define two pre-certificates using the two fixed points as $\q^\lambda :={(\x^\star-\x^\lambda)}/{\lambda}$ and $\hat{\q} := {(\y-\hat{\x})}/{\lambda}$ with the corresponding pre-dual polynomials denoted by $Q^\lambda(f)$ and $\hat{Q}(f)$. Here $\x^\lambda = \sum_{\ell=1}^k {c}_\ell^\lambda \a({f}_\ell^\lambda)$ and $\hat{\x} = \sum_{\ell =1}^k \hat{c}_\ell \a(\hat{f}_\ell)$. Let $\q^\star = \lim_{\lambda \rightarrow 0} \q^\lambda$.
The remaining steps are to:
\begin{enumerate}[label=\arabic*)]
\item Show that $\q^\star$ is a valid dual certificate that certifies the atomic decomposition of $\x^\star$,  i.e.,  $Q^\star(f) = \a(f)^H\mZ\q^\star$ satisfies $Q^\star(f^\star_\ell) = \sign(c^\star_\ell),  \ell = 1,  \ldots,  k$ and $|Q^\star(f)| < 1,  \forall f \notin T^\star$;
\item Use Lemma~\ref{lem:fix1} to bound the pointwise distance between $Q^\star(f)$ and $Q^\lambda(f)$;
\item Use Lemma~\ref{lem:fix2} to bound the pointwise distance between $Q^\lambda(f)$ and $\hat Q(f)$.
\end{enumerate}

\subsubsection{Showing \texorpdfstring{$\q^\star$}{Lg}  is a Dual Certificate}
To show that $\q^\star$ is a dual certificate,  it is sufficient to show that $Q^\star(f)$ satisfies the Bounded Interpolation Property of Proposition~\ref{pro:bip}. The Interpolation property is automatically satisfied due to the construction process,  and we will show the Boundedness property using the arguments of~\cite{Candes:2014br}. In particular,  fix an arbitrary point $f_0^\star\in T^\star$ as the reference point,  and let $f_{-1}^\star$ be the first frequency in $T^\star$ that lies on the left of $f_0^\star$ while $f_{1}^\star$ be the first frequency in $T^\star$ that lies on the right. Here ``left'' and ``right'' are directions on the complex circle $\TT$. We remark that the analysis depends only on the relative locations of $\{f_\ell^\star\}$. Hence,  to simplify the arguments,  we assume that the reference point $f_0^\star$ is at $0$ by shifting the frequencies if necessary.
Then we divide the region between $f_0^\star = 0$ and $f_{1}^\star/2$ into three parts: Near Region $\N:=[0, 0.24/n]$,  Middle Region $\M:=[0.24/n, 0.75/n]$ and Far Region $\F:=[0.75/n,  f_{1}^\star/2]$. Also their symmetric counterparts are defined as $-\N:=[-0.24/n,  0]$,  $-\M:=[-0.75/n, -0.24/n]$,  and $-\F:=[f_{-1}^\star/2,  -0.75/n]$.  
We first show that the dual polynomial has strictly negative curvature $|Q^\star(f)|''< 0$ in $\N=[0, 0.24/n]$ and $|Q^\star(f)|< 1$ in $\M\cup\F=[0.24/n,  f_{1}^\star/2]$,  implying $|Q^\star(f)|< 1$ in $\N\cup\M\cup\F\backslash\{f_0^\star\}$ by exploiting $|Q^\star(f_0^\star)|=1$ and $|Q^\star(f_0^\star)|'=0$.  Then using the same symmetric arguments as in~\cite{Candes:2014br},  we claim that  $|Q^\star(f)|< 1$ in $(-\N)\cup(-\M)\cup(-\F)\backslash\{f_0^\star\}$. Combining these two results with the fact that the reference point $f_0^\star$ is chosen arbitrarily from $T^\star$ (and shifted to $0$),  we establish that the Boundedness property of $Q^\star(f)$ holds in the entire $\TT \backslash T^\star$.

\begin{lemma}[$\q^\star$ is a dual certificate]\label{lem:q0:dual:certificate} The dual polynomial $Q^\star(f)$ satisfies both the Interpolation and  Boundedness properties with respect to the coefficients $\{c^\star_\ell\}$ and the frequencies $\{f_\ell^\star\}$. In addition,  $Q^\star(f)$ satisfies first
\begin{align*}
\begin{matrix*}[l]
Q^\star_R(f) \geq \bl{0.887594}, 
&{Q_R^\star}''(f) \leq -2.24483 n^2,  \\
|Q^\star_I(f)| \leq \bl{0.0183836}, 
&|{Q^\star}''_I(f)| \leq 0.113197  n^2,  \\
|{Q^\star}'(f)| \leq0.821039 n, 
&|{Q^\star}''(f)| \leq3.40320 n^2, 
\end{matrix*}
\end{align*}
and
\begin{align*}
 {Q}^\star_R(f){Q_R^\star}(f)''+|{{Q}^\star}(f)'|^2+|{Q^\star_I}(f)||{{Q}^\star_I}(f)''| \leq -1.316313n^2 < 0
\end{align*}
for $f \in \N$, 
implying $|{Q}^\star(f)|'' <0$ in $\N$,  and
second, 
\begin{align*}
|Q^\star(f)| &\leq 0.927615, ~f\in\M, \\
|Q^\star(f)| &\leq 0.734123, ~f\in\F.
\end{align*}
Here the subscripts $R$ and $I$ denote respectively the real and complex parts of $Q^\star(f)$.
Thus $\q^\star$ is a valid dual certificate to certify the atomic decomposition $\x^\star = \sum_{\ell = 1}^k c_\ell^\star \a(f_\ell^\star)$ such that $\|\x^\star\|_\A = \sum_{\ell = 1}^k |c_\ell^\star|$.
\end{lemma}
\begin{proof}[Proof of Lemma~\ref{lem:q0:dual:certificate}]
See Appendix~\ref{sec:F}.
\end{proof}

%The remaining arguments involve exploiting the closeness of $\btheta^\star$ and $\btheta^\lambda$ shown in Lemma~\ref{lem:fix1} to bound the uniform distance between $Q^\star(f)$ and $Q^\lambda(f)$ and using the closeness of $\btheta^\lambda$ and $\hat\btheta$ given by Lemma~\ref{lem:fix2} to bound the uniform distance between $Q^\lambda(f)$ and $\hat{Q}(f)$.
%To show the former,  recognize that
%\begin{align*}
%Q^\lambda (f) - Q^\star(f) = \a(f)^HZ(\q^\lambda - \q^\star) = \a(f)^H\mZ\left(\frac{\x^\star - \x^\lambda}{\lambda} + \frac{\mathrm{d} }{\mathrm{d} \lambda} \x^\lambda\big|_{\lambda=0}\right)  =
%\frac{1}{\lambda}\int_0^\lambda \a(f)^H\mZ\left( \frac{\mathrm{d}}{\mathrm{d} t}\x^\star- \frac{\mathrm{d}}{\mathrm{d} t}\x^t \right)\mathrm{d}t.
%\end{align*}
%Then by developing upper bounds of RHS of the above equation,  we bound the uniform distance between $Q^\lambda(f)$ and $Q^\star(f)$ in Lemma~\ref{lem:Q:lambda:closeto:Q:0},  whose proof is given in Appendix~\ref{sec:G}.

Next lemma,  with the proof given in Appendix~\ref{sec:G},  exploits the closeness of $\btheta^\star$ and $\btheta^\lambda$ shown in Lemma~\ref{lem:fix1} to bound the pointwise distance between $Q^\star(f)$ and $Q^\lambda(f)$.

\begin{lemma}[$Q^\lambda(f)$ is close to $Q^\star(f)$]\label{lem:Q:lambda:closeto:Q:0}
Under the settings of Lemma~\ref{lem:fix1},  let $Q^\lambda(f)$ and $Q^\star(f)$ be the dual polynomials corresponding to $\btheta^\lambda$ and $\btheta^\star$,  respectively.
Then the distances between $Q^\lambda(f)$ and $Q^\star(f)$ and their various derivatives are uniformly bounded:
\begin{align*}
\begin{matrix*}[l]
|{Q^\star}(f)-{Q^\lambda}(f)| \leq   {28.7343}X^\star{B^\star} \gamma, ~f\in\N, &
|{Q^\star}(f)-{Q^\lambda}(f)| \leq   {\bl{39.3557}}X^\star{B^\star} \gamma, ~f\in\M, 
\\
|{Q^\star}'(f)-{Q^\lambda}'(f)|\leq  {44.4648}nX^\star{B^\star}\gamma, ~f\in\N, &
|{Q^\star}(f)-{Q^\lambda}(f)| \leq   {66.1596}X^\star{B^\star} \gamma, ~f\in\F, 
\\
|{Q^\star}''(f)-{Q^\lambda}''(f)|\leq  \bl{140.808}n^2X^\star{B^\star} \gamma, ~f\in\N.&
\end{matrix*}
\end{align*}
\end{lemma}

In the following,  we will control the pointwise distance between $Q^\lambda(f)$ and $\hat Q(f)$ by taking advantage of the closeness of
$\btheta^\lambda$ and $\hat\btheta$ given by Lemma~\ref{lem:fix2}.
The key is to observe that
\[\hat{\q}-\q^\lambda=\frac{(\y-\hat{\x})-(\x^\star-\x^\lambda)}{\lambda}
                           =\frac{\w}{\lambda}+\frac{\x^\lambda-\hat{\x}}{\lambda}\]
implying
\begin{align}
|Q^\lambda(f)-\hat{Q}(f)|\leq \frac{|\a(f)^H\mZ\w|}{\lambda} +\frac{| \a(f)^H\mZ(\x^\lambda-\hat{\x})| }{\lambda}.\label{eqn:sec:4:Qlamba-Qhat}
\end{align}
This separates the distance between $Q^\lambda(f)$ and $\hat{Q}(f)$ into two parts: one is $|{\a(f)^H\mZ\w}/{\lambda}|$ determined by the dual atomic norm of the Gaussian noise $\w$,  which is upperbounded in Appendix~\ref{sec:B}; the other is $|{ \a(f)^H\mZ(\x^\lambda-\hat{\x}) }/{\lambda}|$ that can be upperbounded by the dual atomic norm of $\x^\lambda-\hat{\x}$.
We summarize the final result in Lemma~\ref{lem:Q:hat:lambda:closeto:Q:lambda},  where the proof is given in Appendix~\ref{sec:H}.

\begin{lemma}[$\hat{Q}(f)$ is close to $Q^\lambda(f)$]\label{lem:Q:hat:lambda:closeto:Q:lambda}
Under the settings of Lemma~\ref{lem:fix2},  let $\hat{Q}$ and $Q^\lambda$ be the dual polynomials corresponding to $\hat{\btheta}$ and $\btheta^\lambda$,  respectively. Then the pointwise distances between $Q^\lambda(f)$ and $\hat{Q}(f)$ and their derivatives are bounded:
\begin{align*}
\begin{matrix*}[l]
		|{\hat{Q}}(f)-{Q^\lambda}(f)| \leq 82.5975  B^\star/{X^\star}, ~f\in\N, &
		|{\hat{Q}}(f)-{Q^\lambda}(f)| \leq  \bl{114.323}  B^\star/{X^\star}, ~f\in\M, 
		\\
		|{\hat{Q}}(f)'-{Q^\lambda}'(f)|\leq 180.283 n  B^\star/{X^\star}, ~f\in\N, &
		|{\hat{Q}}(f)-{Q^\lambda}(f)| \leq 162.903  B^\star/{X^\star}, ~f\in\F, 
		\\
		|{\hat{Q}}(f)''-{Q^\lambda}''(f)|\leq 758.404n^2  B^\star/{X^\star}, ~f\in\N.&
\end{matrix*}
\end{align*}
\end{lemma}
%\begin{proof}[Proof of Lemma~\ref{lem:Q:hat:lambda:closeto:Q:lambda}]
%See Appendix~\ref{sec:H}.
%\end{proof}

\subsubsection{Proof of Theorem~\ref{thm:main}}
By combining Lemmas~\ref{lem:q0:dual:certificate}, ~\ref{lem:Q:lambda:closeto:Q:0},  and~\ref{lem:Q:hat:lambda:closeto:Q:lambda},  we are now ready
to prove Theorem~\ref{thm:main}.

Basically,  we will show that $\hat{\btheta}$ constructed from the two-step gradient descent procedure and ${\btheta}^{\mathrm{glob}}:=({\f}^{\mathrm{glob}}, {\u}^{\mathrm{glob}}, {\v}^{\mathrm{glob}})$ are the same point. Then the error bounds follow from the closeness of $\hat{\btheta}$ and $\btheta^\star$. First,  we show that the signal $\hat{\x}=\sum_{\ell=1}^k \hat{c}_\ell\a(\hat{f}_\ell)$ and $\hat{\q}=(\y-\hat{\x})/\lambda$ constructed from the second fixed point $\hat{\btheta}$ form primal and dual optimal solutions of~\eqref{eqn:primal}. It suffices to show that the dual polynomial $\hat{Q}(f)=\a(f)^H\mZ\hat\q$ satisfies the Bounded Interpolation Property of Proposition~\ref{pro:bip}.

\bigskip
\noindent{\bf 1) Showing the Interpolation property.}
\medskip
\\
The Interpolation property has the following equivalences:
\begin{align}
\hat{Q} (\hat{f}_\ell) =  \sign(\hat{c}_\ell), \ell=1,  \ldots , {k} \nn
\iff &\ \a(\hat{f}_\ell)^H\mZ(\y-\hat{\x})=\lambda\sign(\hat{c}_\ell),  \ell=1, \ldots, k\nn\\
\iff &\ \a(\hat{f}_\ell)^H\mZ(\y-\mA(\hat{\f})\hat{\c})=\lambda\sign(\hat{c}_\ell),  \ell=1, \ldots, k\nn\\
\iff &\ \mA(\hat\f)^H\mZ(\y-\mA(\hat{\f})\hat{\c})=\lambda\hat\c./|\hat\c|. \label{eqn:interpolation:QW}
\end{align}
From Lemma~\ref{lem:fix2},  $\hat{\btheta}$ is the fixed point solution of the map $\Theta(\btheta)=\btheta-\W^\star\nabla G(\btheta)$, 
i.e.,  $\Theta(\hat{\btheta})=\hat\btheta$,  implying $\nabla G(\hat\btheta)=\zero$ due to the invertibility of $\W^\star$. Invoking the explicit expression for $\nabla G(\btheta)$ developed in Appendix~\ref{sec:C},  we get
\begin{align}\label{eqn:gradient:G}
\nabla G(\hat\btheta)=
\begin{bmatrix*}[l]
\R\{{(\mA'(\hat\f)\diag(\hat\c))^H\mZ(\mA(\hat\f)\hat\c-\y)}\}
\\
\R\{{\mA(\hat\f)^H\mZ(\mA(\hat\f)\hat\c-\y)}+\lambda \hat\c./ |\hat\c |\}
\\
\I\{{\mA(\hat\f)^H\mZ(\mA(\hat\f)\hat\c-\y)}+\lambda \hat\c./ |\hat\c |\}
\end{bmatrix*}
=
\begin{bmatrix}
\zero
\\
\zero
\\
\zero
\end{bmatrix}.
\end{align}
Then the Interpolation property~\eqref{eqn:interpolation:QW} follows from the last two row blocks of~\eqref{eqn:gradient:G}.

\bigskip
\noindent{\bf 2) Showing the Boundedness property.}
\medskip
\\
Following the same arguments preceding Lemma~\ref{lem:q0:dual:certificate},  it is sufficient to show $|\hat{Q}(f)|<1$ in $\N\cup\M\cup\F\backslash \{\hat{f}_0\}$.

%we first show that the dual polynomial has strictly negative curvatures $|\hat{Q}(f)|''< 0$ in $\N$ and $|\hat{Q}(f)|< 1$ in $\M\cup\F$. Then using the same symmetric arguments in~\cite{Candes:2014br},    we claim that $|\hat{Q}(f)|''< 0$ in $(-\N)$ and $|\hat{Q}(f)|< 1$ in $(-\M)\cup(-\F)$. Then the desired result follows from
%$|\hat{Q}(\hat{f}_0)|=1$,  $|\hat{Q}(\hat{f}_0)|'=0$ and $\hat{f}_0\in(-\N)\cup(\N)$,  since $\max_k|\hat{f}_k-f^\star_k|\ll 0.24/n$ implied by~\eqref{eqn:snr},  Lemma~\ref{lem:fix1} and Lemma~\ref{lem:fix2}.

First,  since $\hat{f}_0$ might be located in $-\N$ or $\N$,   we bound $|\hat{Q}(f)|$  for $f\in(-\N)\cup\N$. The second-order Taylor expansion of $|\hat{Q}(f)|$ at $f = \hat{f}_0$ states
\begin{align}
|\hat{Q}(f)| &= |\hat{Q}(\hat{f}_0)| + (f-\hat{f}_0) |\hat{Q}(\hat{f}_0)|' + \frac{1}{2}(f-\hat{f}_0)^2 |\hat{Q}(\xi)|''\nonumber\\
& = 1 + (f-\hat{f}_0) |\hat{Q}(\hat{f}_0)|' + \frac{1}{2}(f-\hat{f}_0)^2 |\hat{Q}(\xi)|''\text{\ for\ some\ } \xi \in (-\N)\cup\N, 
\end{align}
where for the second line we used a consequence of the interpolation property.
We argue that
\[
|\hat{Q}(\hat{f}_0)|' = \frac{\hat{Q}_R(\hat{f}_0)\hat{Q}_R(\hat{f}_0)'+\hat{Q}_I(\hat{f}_0)\hat{Q}_I(\hat{f}_0)'}{|\hat{Q}(\hat{f}_0)|} = \frac{\R\{\hat{c}_0\}\hat{Q}_R(\hat{f}_0)'+\I\{\hat{c}_0\}\hat{Q}_I(\hat{f}_0)'}{|\hat{c}_0||\hat{Q}(\hat{f}_0)|} = 0.
\] The last equality is a consequence of the first row block of~\eqref{eqn:gradient:G} since $\R\{\hat{c}_0\}\hat{Q}_R(\hat{f}_0)'+\I\{\hat{c}_0\}\hat{Q}_I(\hat{f}_0)' = \R\{\hat{c}_0^H \a(\hat{f}_0)^H\mZ(\y - \mA(\hat{\f})\hat{\c})\}$. Therefore,  it suffices to show that $|\hat{Q}(f)|'$  has strictly negative derivative in the symmetric Near Region
$f\in(-\N)\cup\N$. By the symmetric arguments,  it suffices to show this in $\N$.  Since
\begin{align*}
|\hat{Q}(f)|'' = -\frac{(\hat{Q}_R(f)\hat{Q}_R(f)'+\hat{Q}_I(f)\hat{Q}_I(f)')^2}{|\hat{Q}(f)|^3} + \frac{\hat{Q}_R(f){\hat{Q}_R}(f)''+|{\hat{Q}}(f)'|^2+|\hat{Q}_I(f)||{\hat{Q}_I}(f)''|}{|\hat{Q}(f)|}, 
\end{align*}
we only need to show that
\begin{align*}
\hat{Q}_R(f){\hat{Q}_R}(f)''+|{\hat{Q}}(f)'|^2+|\hat{Q}_I(f)||{\hat{Q}_I}(f)''|<0, 
\end{align*}
which can be obtained by applying Lemma~\ref{lem:q0:dual:certificate},  Lemma~\ref{lem:Q:lambda:closeto:Q:0},  Lemma~\ref{lem:Q:hat:lambda:closeto:Q:lambda} and the triangle inequality to control these three terms $\hat{Q}_R(f){\hat{Q}_R}(f)''$,  $|{\hat{Q}}(f)'|^2$ and $|\hat{Q}_I(f)||{\hat{Q}_I}(f)''|$,  respectively.
%For example,  to bound the first term $\hat{Q}_R(f){\hat{Q}_R}(f)''$,  the key is relating the noisy and regularized version $\hat{Q}_R(f)$ with the noise-free and regularization-free version $Q^\star_R(f)$ by considering the middle stage which is the noise-free but regularized version $Q^\lambda_R(f)$. Such a middle stage serves a bridge for us to separate the total error into two parts: one is between the noisy and regularized version and the noise-free but regularized version which is revealed by Lemma~\ref{lem:Q:hat:lambda:closeto:Q:lambda},  the other is related with the noise-free and regularization-free version and the noisy and regularized version,  which can be easily controlled by Lemma~\ref{lem:Q:lambda:closeto:Q:0}. The rest of arguments would consist of  applying the triangle inequality to get the final result. The similar arguments are suitable for controlling the second and the third terms.

More precisely,    the first term can be bounded by
\begin{align}\label{eqn:first:term}
&\hat{Q}_R(f){\hat{Q}_R}(f)''\nn\\
\leq&     Q^\star_R(f){Q_R^\star}(f)''+|\hat{Q}_R(f)-Q^\star_R(f)||{\hat{Q}_R}(f)''-{Q_R^\star}(f)''|
+|Q^\star_R(f)||{\hat{Q}_R}(f)''-{Q_R^\star}(f)''|+|\hat{Q}_R(f)-Q_R^\star(f)||{Q_R^\star}(f)''|
\nn\\
\leq& (\bl{0.887594})( -2.24483 n^2)
+( 28.7343X^\star{B^\star} \gamma+82.5975B^\star/{X^\star})(\bl{140.808}n^2X^\star{B^\star} \gamma+ 758.404n^2B^\star/{X^\star})
\nn\\
&+(1)(\bl{140.808}n^2X^\star{B^\star} \gamma
+ 758.404n^2B^\star/{X^\star})
+( 28.7343X^\star{B^\star} \gamma+82.5975B^\star/{X^\star})3.40320 n^2
\nn\\
\leq&  -\bl{1.64194} n^2, 
\end{align}
where we have used the SNR condition~\eqref{eqn:snr}: $X^\star B^\star\gamma\leq10^{-3}, B^\star/{X^\star}\leq10^{-4}$ in the last line. We now bound the second term 
\begin{align}\label{eqn:second:term}
|{\hat{Q}}(f)'|^2
=& |{\hat{Q}}(f)'-{Q^\star}'(f)|^2+|{Q^\star}'(f)|^2+ 2|{ {Q}^\star}(f)'||{\hat{Q}}(f)'-{Q^\star}'(f)|
\nn\\
\leq& (44.4648nX^\star{B^\star}\gamma+ 180.283 nB^\star/{X^\star} )^2+(0.821039 n)^2
+2(0.821039 n)(44.4648nX^\star{B^\star}\gamma+ 180.283 nB^\star/{X^\star} )
\nn\\
\leq&0.780629 n^2.
\end{align}
Finally,   the third term can be bounded by
\begin{align}\label{eqn:third:term}
&|\hat{Q}_I(f)||{\hat{Q}_I}(f)''|\nn\\
\leq&( |Q_I^\star(f)|+|\hat{Q}(f)-Q^\star(f)|)(|{Q_I^\star}''(f)|+|{\hat{Q}}(f)''-{Q^\star}''(f)|)
\nn\\
\leq&(\bl{0.0183836}+( 28.7343X^\star{B^\star} \gamma+82.5975B^\star/{X^\star}) )0.113197  n^2
+(\bl{140.808}n^2X^\star{B^\star} \gamma+ 758.404n^2B^\star/{X^\star})
\nn\\
\leq& \bl{0.222917} n^2.
\end{align}
From~\eqref{eqn:first:term}, ~\eqref{eqn:second:term} and~\eqref{eqn:third:term},  we have
\begin{align*}
\hat{Q}_R(f){\hat{Q}_R}(f)''+ |{\hat{Q}}(f)'|^2+   |\hat{Q}_I(f)||{\hat{Q}_I}(f)''| \leq (-\bl{1.64194}+0.780629+\bl{0.222917})n^2 <0, 
\end{align*}
implying that $|\hat{Q}(f)|''<0$ in $\N$. This completes showing $|\hat{Q}(f)|''<0$ in $(-\N)\cup\N$ and
\begin{align}
|\hat{Q}(f)| < 1,  \quad\text{for }f \in (-\N)\cup\N\backslash\{\hat{f}_0\}. \label{eqn:Near:control}
\end{align}
%By symmetric arguments as in~\cite{Candes:2014br},  the above arguments also hold for $f\in-\N$. Hence, 
%\begin{align}
%|\hat{Q}(f)|''  < 0, \text{\ for }f\in-\N\cup\N.   \label{eqn:Near:control}
%\end{align}
%Further,  since $\hat{Q}(\hat{f}_0)=1$ and  $|\hat{f}_0-f^\star_0|\leq 0.4(X^\star+35.2)\gamma/n\ll 0.24/n$ by Lemma~\ref{lem:fix1} and Lemma~\ref{lem:fix2}, 
%we have $\hat{f}_0\in-N\cup\N$. Then invoking equation~\eqref{eqn:Near:control} and applying the second order Taylor expansion theorem to $|\hat{Q}(f)|$
%at $\hat{f}_0$ (where its  gradient vanishes),   we finally arrive at $|\hat{Q}(f)|\leq 1$ for $f\in-\N\cup\N$.
Next,  we bound $|\hat{Q}(f)|$ in Middle Region 
\begin{align}
|\hat{Q}(f)|
\leq& |Q^\star(f)|+ |{Q^\star}(f)-{{Q}^\lambda}(f)| +|{\hat{Q}}(f)-{{Q}^\lambda}(f)|
\nn\\
\leq&  0.927615  +( \bl{39.3557}X^\star{B^\star} \gamma+\bl{114.323}B^\star/{X^\star})
\nn\\
\leq& 0.978403
<1, \quad\text{for }f\in\M.\label{eqn:Middle:control}
\end{align}
Finally,  we arrive at an upper bound of $|\hat{Q}(f)|$ in Far Region:
\begin{align}
|\hat{Q}(f)|  \leq& |Q^\star(f)|+ |{Q^\star}(f)-{{Q}^\lambda}(f)| +|{\hat{Q}}(f)-{{Q}^\lambda}(f)|
\nn\\
\leq& 0.734123+ (66.1596X^\star{B^\star}  \gamma+162.903B^\star/{X^\star})
\nn\\
\leq& 0.81658
<1,  \quad\text{for }f\in\F. \label{eqn:Far:control}
\end{align}
From~\eqref{eqn:Near:control}, ~\eqref{eqn:Middle:control} and~\eqref{eqn:Far:control},  we obtain that $\hat{Q}(f)$ satisfies the BIP property and hence $\hat{\q}$ is a valid dual certificate that certifies the optimality of $\hat{\x} = \sum_{\ell = 1}^k \hat{c}_\ell\a(\hat{f}_\ell)$. The uniqueness of the decomposition as also certified by $\hat{\q}$ implies that $\{\hat{f}_\ell\}_{\ell=1}^k = \{{f}^{\mathrm{glob}}_\ell\}_{\ell=1}^k$ and $\{\hat{c}_\ell\}_{\ell=1}^k = \{{c}_\ell^{\mathrm{glob}}\}_{\ell=1}^k$,  i.e.,  $\hat{\btheta}$ and ${\btheta}^{\mathrm{glob}}$ are the same point.

As the final step,  using Lemma~\ref{lem:fix1},  Lemma~\ref{lem:fix2} and the triangle inequality,  we have
\[\|\hat\btheta-\btheta^\star\|_{\hinfty}\leq \|\hat{\btheta}-\btheta^\lambda\|_\hinfty+\|\btheta^\lambda-\btheta^\star\|_\hinfty  \leq (X^\star+35.2)\gamma_0/\sqrt{2}.\] Then the desired results follow from the definition of the norm $\|\cdot\|_{\hinfty}$ and the fact that $\sqrt{|K''(0)|} \geq 3.289n^2$ for $n\geq130$ by~\eqref{eqn:sec:A:tau} and hence
$1/\sqrt{2 |K''(0)|}\leq 1/ \sqrt{2 (3.289)}/n\leq 0.3899/n\leq 0.4/n.$
{\hfill$\square$}

\section{Numerical Experiments}\label{sec:experiment}

We present numerical results to support our theoretical findings. In particular,  we first examine the phase transition curve of the rate of success in Figure~\ref{fig:phase}. In preparing Figure~\ref{fig:phase},  $k$ complex coefficients $c^\star_1, \ldots, c^\star_k$ were generated uniformly from the unit complex circle such that $c^\star_{\min}=c^\star_{\max}=1$ hence $B^\star=1$. We also generated $k$ normalized frequencies $f^\star_1, \ldots, f^\star_k$ uniformly chosen from $[0, 1]$ such that every pair of frequencies are separated by at least $2.5/n$. Then the signal $\x^\star$ was formed according to~\eqref{eqn:atom_true}. We created our observation $\y$ by adding Gaussian noise of mean zero and variance $\sigma^2$ to the target signal $\x^\star$. Let $\lambda=x\gamma_0$ (recall that $\lambda=0.646X^\star\gamma_0$ in Theorem~\ref{thm:main}  and hence $x=0.646X^\star$). We varied $x$ and the Noise-to-Signal Ratio $\gamma$. For each fixed $(x,  \gamma)$ pair,  $20$ instances of the spectral line signals were generated. We then solved~\eqref{eqn:primal} for each instance and extracted the frequencies and coefficients. We declared success for an instance if i) the recovered frequency vector is within $\gamma/2n$ $\ell_\infty$ distance of the true frequency vector $\f^\star$,  and ii) the recovered coefficient vector is within $2\lambda$ $\ell_\infty$ distance of the true frequency vector $\c^\star$. The rate of success for each algorithm is the proportion of successful instances.
\begin{figure}[h!t]
  \centering
  \includegraphics[width=0.4\textwidth]{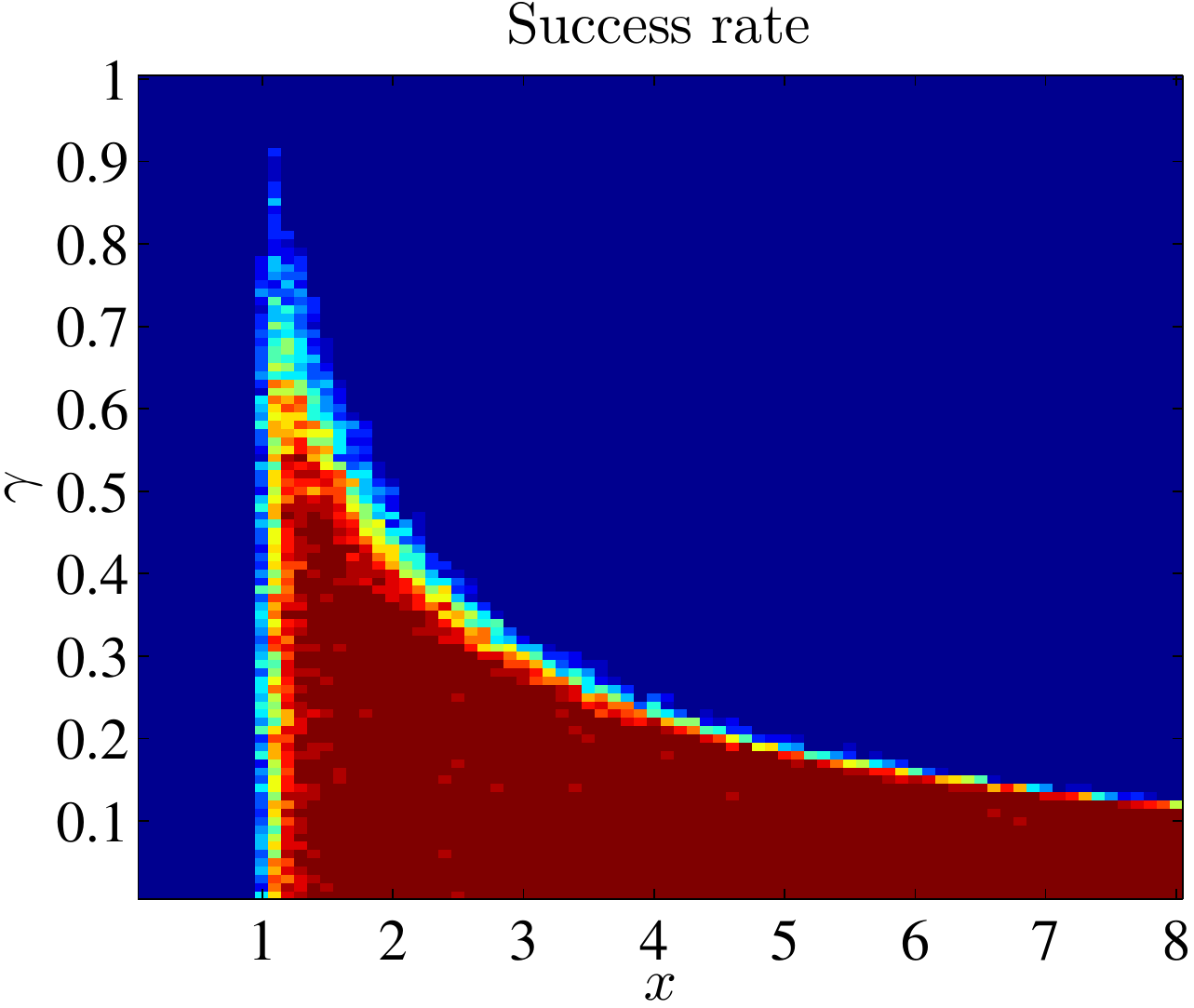}\\
  \caption{Rate of success for line spectral estimation by solving the atomic norm regularized program~\eqref{eqn:primal}.}\label{fig:phase}
\end{figure}

From Figure~\ref{fig:phase},  we observe that solving~\eqref{eqn:primal} is unable to identify the sinusoidal parameters if  $x\leq 1$ and the performance of the method is unstable when $x$ is around $1$. When $x$ is set to be slightly larger than $1$,  however,  we
almost always succeed in finding good estimates of the sinusoidal parameters as long as $x\gamma\leq c$ for some small constant $c$.
This matches the findings in Theorem~\ref{thm:main}.  Figure~\ref{fig:phase} also shows the constants in Theorem~\ref{thm:main}  are a bit conservative.

We also run simulations to compare the mean-squared error for our frequency estimate with those for MUSIC and the MLE,  as well as the CRB. The simulation results are listed in Figure~\ref{fig:crb}. We emphasize that the MLE is initialized using the true frequencies and coefficients,  which are not available in practice. We focus on the case of two unknown frequencies and examine the effect of separation. We observe that the atomic norm minimization method always outperforms MUSIC,  with increased performance gap when the frequencies become closer. While the MLE performs the best,  its initialization is not practical.

\begin{figure}
  \centering
\includegraphics[width=0.32\textwidth]{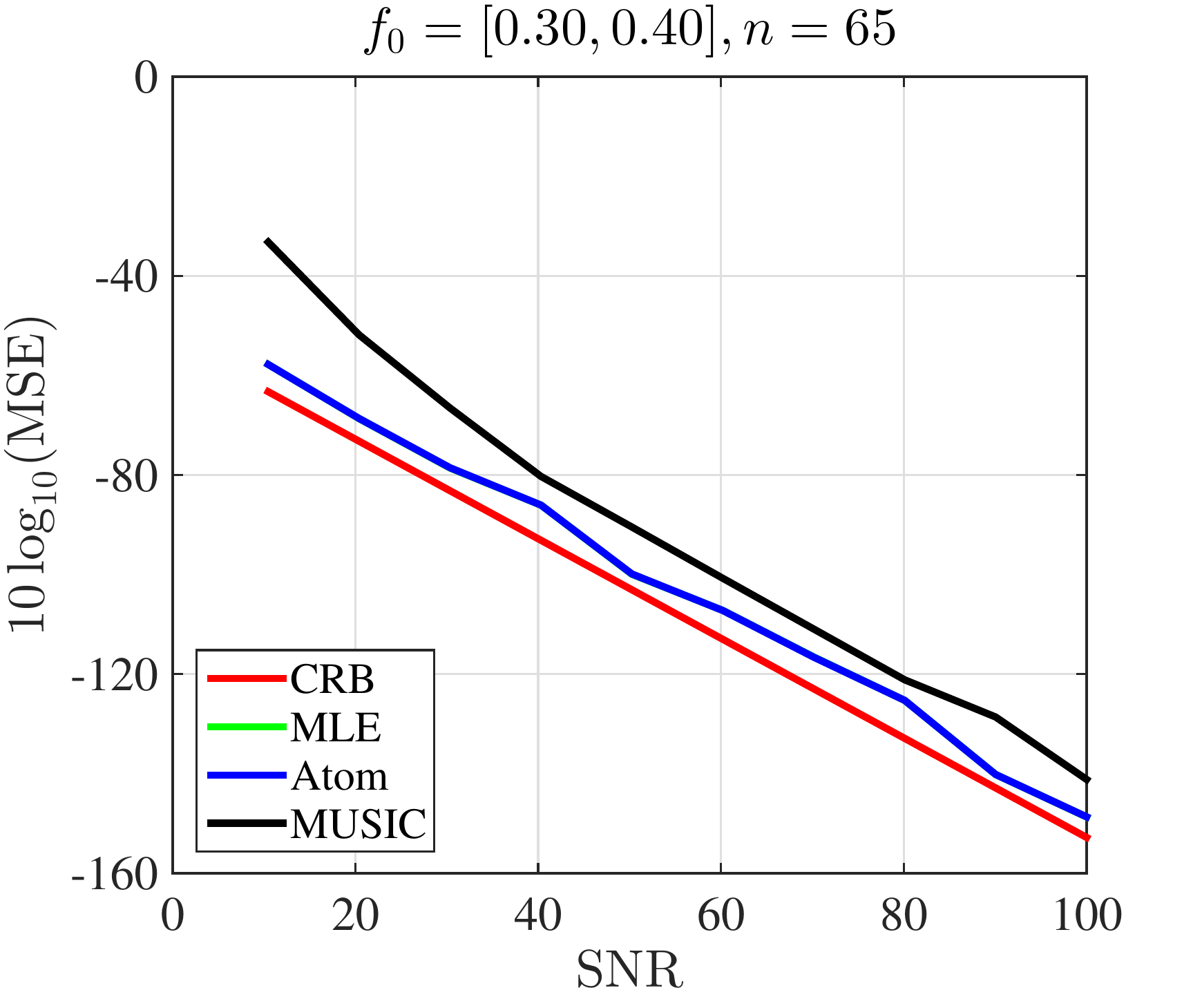}
 \includegraphics[width=0.32\textwidth]{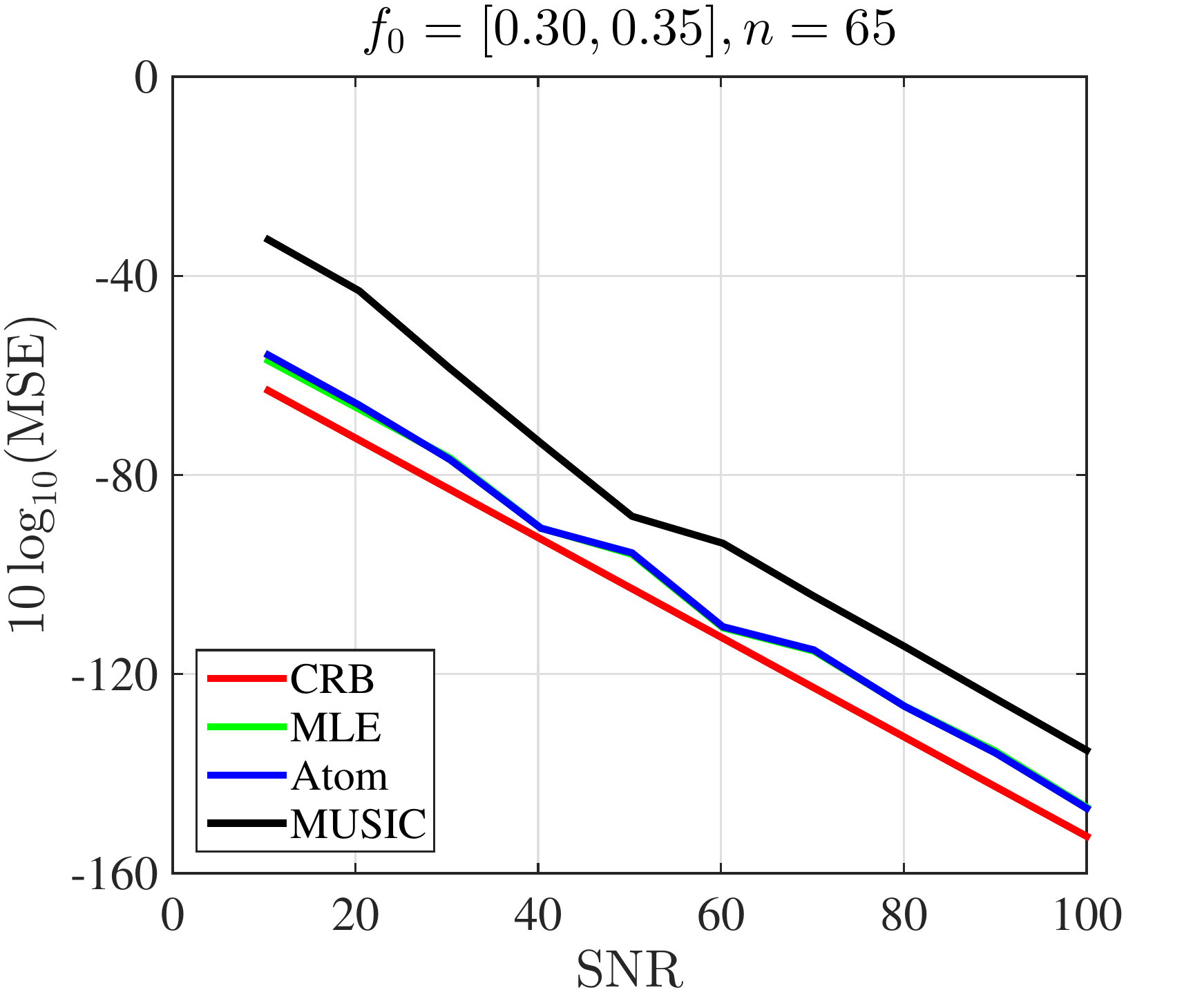}
\includegraphics[width=0.32\textwidth]{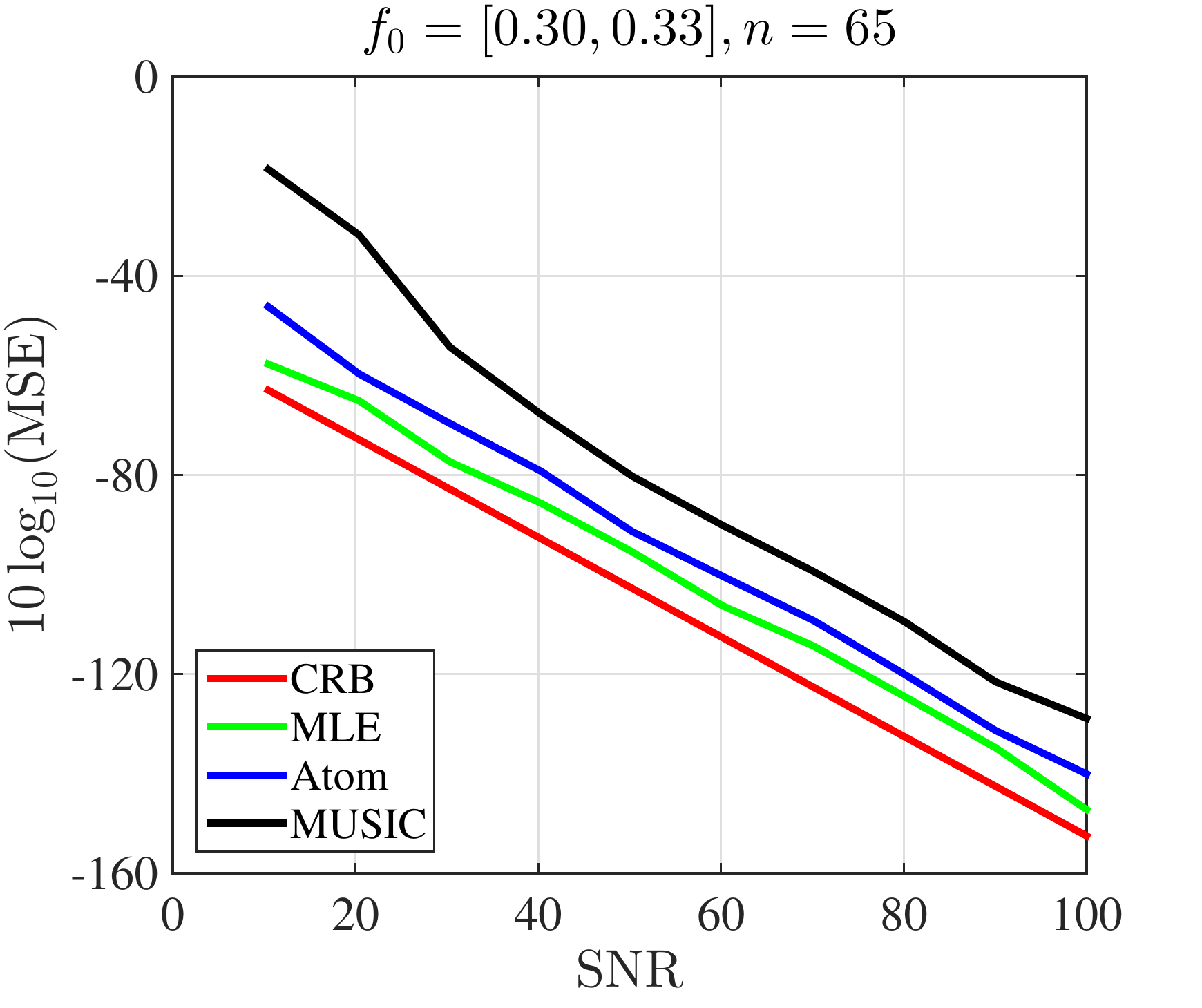}
  \caption{Performance comparison: Atomic norm minimization~\eqref{eqn:primal} (labeled as ``Atom"),  MUSIC,  MLE initialized by the true parameters,  and the CRB.}\label{fig:crb}
\end{figure}

\section{Conclusions}\label{sec:conclusion}

This work considers the problem of approximately estimating the frequencies and coefficients of a superposition of complex sinusoids in white noise. By using a primal-dual witness construction,  we have established theoretical performance guarantees for atomic norm minimization algorithm in line spectral parameter estimation. The obtained error bounds match the Cram\'er-Rao lower bound up to a logarithmic factor. The relationship between resolution (separation of frequencies) and precision or accuracy of the estimator is highlighted. Our analysis also reveals that the atomic norm minimization can be viewed as a convex way to solve a $\ell_1$-norm regularized,  nonlinear and nonconvex least-squares problem to global optimality.

\begin{appendices}
\section{Jackson Kernel}\label{sec:A}
For any integer $M > 0$,  the Jackson kernel,  also known as the squared Fej\'er kernel,  is defined by \cite[Eq. (IV.2)]{Tang:2013fo} or \cite[Eq. (2.3) with $M=f_c/2+1$]{Candes:2014br}
\begin{align}\label{eqn:sec:A:fejer:kernel}
K(f)=\left[ \frac{\sin(\pi Mf)}{M\sin(\pi f)} \right]^4.
\end{align}
The Jackson kernel shows up in the construction of dual polynomials that satisfy the  Boundedness and Interpolation properties. The choice of the Jackson kernel is due to its nice properties as easily seen from its graph: it attains one at the peak,  and quickly decrease to zero. Cand\`{e}s and Fernandez-Granda showed in~\cite{Candes:2014br} that as long as the frequencies composing a signal satisfy certain separation condition,  then a dual polynomial can be constructed as a linear combination of shifted copies of the Jackson kernel and its first-order derivative to certify that the decomposition achieves the signal's atomic norm.

We use $K'(f),  K''(f),  K'''(f)$ to denote respectively the first,  second,  and third order derivatives of the Jackson kernel and more generally $K^{(\ell)}(f)$ the $\ell$th order derivative. We will frequently use the second order derivative of the Jackson kernel evaluated at zero $K''(0)$,  whose value is~\cite[Above Eq. (IV.5)]{Tang:2013fo}
\begin{align*}
K''(0)=-\frac{4\pi^2(M^2-1)}{3}=-\frac{\pi^2(n^2-4)}{3}.
\end{align*}
Here we used the convention that $n=2M.$
Then its absolute value $|K''(0)|$ (denoted by $\tau$) falls into the interval
\begin{align*}
|K''(0)|\in\left[\left(\frac{\pi^2}{3}-\frac{4\pi^2}{3(130)^2}\right)n^2,  \left(\frac{\pi^2}{3}\right)n^2 \right], \quad\text{for } n \geq 130.
\end{align*}
For ease of exposition,  we give an explicit lower bound on $|K''(0)|$ (which is valid for any $n\geq130$):
\begin{align}\label{eqn:sec:A:tau}
\tau:=|K''(0)|\geq \left(\frac{\pi^2}{3}-\frac{4\pi^2}{3(130)^2}\right)n^2\geq 3.289n^2, \quad\text{for }n\geq130.
\end{align}
At a high-level,  the purpose of introducing $\tau=|K''(0)|$ is to normalize the second order derivative of the Jackson kernel to $1$ at $f = 0$.

\subsection{Decomposing the Jackson Kernel}\label{sec:Jackson}
 The Jackson kernel admits the following decomposition~\cite{Tang:2013fo}
\begin{align*}
K(f_2-f_1) & = \left[\frac{\sin(\pi M(f_2-f_1))}{M\sin(\pi(f_2-f_1))}\right]^4
= \a(f_1)^H\mZ\a(f_2)
=\frac{1}{M}\sum_{j=-2M}^{2M}g_M(j)e^{-i2\pi j (f_2-f_1) }, 
\end{align*}
where $M= {n}/{2}$ and $\mZ$ is an $n\times n$ diagonal matrix whose diagonal entries are given by $[\mZ]_{\ell\ell}=\frac{g_M(\ell)}{M}$
with
\begin{align}\label{sec:A:gM}
g_M(\ell)=\frac{1}{M}\sum_{k=\max(\ell-M, -M)}^{\min(\ell+M, M)}\left(1-\bigg|\frac{k}{M}\bigg|\right)
\left(1-\bigg|\frac{\ell-k}{M}\bigg|\right) \geq 0,  \ell=-2M, \ldots, 2M, 
\end{align}
the convolution of two discrete triangle functions scaled by $1/M$. The weighting function $g_M(\ell)$ attains its peak at zero and
\begin{align*}
g_M(0)
&=\frac{1}{M}\sum_{k=-M}^{M}\left(1- \left| \frac{k}{M} \right|\right)^2
=\frac{2}{3} + \frac{2}{M^2}
\stack{\ding{172}}{\leq} \frac{2}{3}+ \frac{2}{65^2}\leq  0.667, 
\end{align*}
where \ding{172} holds  for $M\geq 65$ (or $n\geq 130$) by noting that $2/M^2$ is a decreasing function of $M$.
Using the definition of $\mZ$,  we bound $\|\mZ\|_{\infty, \infty}$ as
\begin{equation}\label{eqn:sec:A:gM}
\begin{aligned}
\|\mZ\|_{\infty, \infty}
&= \max_{-2M\leq j\leq 2M}\frac{ g_M(j)}{M}
= \frac{g_M(0)}{M}
\leq   \frac{0.667}{M}, \quad\text{for~}n\geq 130.
\end{aligned}
\end{equation}

\subsection{Decomposing the Jackson Kernel Matrices}
We frequently use matrices formed by sampling the Jackson kernel and its derivatives at appropriate frequencies. Given a finite set of frequencies $T=\{f_\ell\}_{\ell=1}^k$ (or its vector form $\f\in\R^k$),  we define
\begin{equation}\label{eqn:sec:A:kernelmatrix}
\begin{aligned}
\D_0(\f):&=[K(f_m-f_n)]_{1\leq n\leq k, 1\leq m\leq k}
=\mA(\f)^H\mZ\mA(\f);
\\
\D_1(\f):&=[K'(f_m-f_n)]_{1\leq n\leq k, 1\leq m\leq k}
=\mA(\f)^H\mZ\mA'(\f)=-{\mA'(\f)}^H\mZ\mA(\f);\\
\D_2(\f):&=[K''(f_m-f_n)]_{1\leq n\leq k, 1\leq m\leq k}
=-{\mA'(\f)}^H\mZ\mA'(f)={\mA''(\f)}^H\mZ\mA(\f)=\mA(\f)^H\mZ\mA''(\f), 
\end{aligned}
\end{equation}
where
\begin{align*}
\mA(\f)   := [\a(f_1), \ldots, \a(f_k)], \
\mA'(\f)   := i2\pi\diag(\n)\mA(\f), \
A{''}(\f):= (i2\pi\diag(\n))^2\mA(\f)
\end{align*}
with $\n=[-n,  -n+1, \ldots,  0, \ldots,  n-1,  n]^T$. More generally,  the kernel matrix $\D_\ell(\f):=[K^{(\ell)}(f_m-f_n)]_{1\leq n\leq k, 1\leq m\leq k}$ satisfies the factorization
\begin{align}\label{eqn:sec:A:kernelmatrix:general}
\D_\ell(\f)=(-1)^j\mA^{(j)}(\f)^H\mZ\mA^{(\ell-j)}(\f), \quad\text{for~}j\leq\ell, 
\end{align}
where $\mA^{(\ell)}(\f)$ represents  the $\ell$th order derivative of the matrix $\mA(\f)$:
$$\mA^{(\ell)}(\f)=(i2\pi\diag(\n))^\ell \mA(\f).$$
Similarly,  we define the cross kernel matrices with respect to the frequency pair $(\f^1, \f^2)$ or $(\{f^1_\ell\}, \{f^2_\ell\})$ as
\begin{align*}
\D_\ell(\f^1, \f^2) = [K^{(\ell)}(f_m^2-f_n^1)]_{1\leq n\leq k, 1\leq m\leq k}, \quad\text{for}~\ell=0, 1, 2.
\end{align*}
We can also express $\D_\ell(\f^1, \f^2) $ in factorization forms
\begin{align}\label{eqn:sec:A:kernelmatrix:cross}
\D_\ell(\f^1, \f^2)=(-1)^j\mA^{(j)}(\f^1)^H\mZ\mA^{(\ell-j)}(\f^2),\quad\text{for~}j\leq\ell.
\end{align}

\subsection{Bounding the Jackson Kernel}

The following lemma provides a set of bounds on the $\ell$th derivative of the Jackson kernel for $\ell\in\{0, 1, 2, 3, 4\}$.

\begin{lemma}[Bounds on $|K^{(\ell)}|$] \label{lem:bounds:drivative:K}
For $\ell\in\{0, 1, 2, 3, 4\}$,  let $K^{(\ell)}$ be the $\ell$th derivative of $K$ ($K=K^{(0)}$).
Define $s(f)$ as a symmetric and periodic function with period 1 and $s(f) = \frac{1}{Mf(3-4f^2)}$  for $f\in(0, 1/2]$.
Then for $f\in(0, 1/2]$,  we have
\begin{align*}
|K^{(0)}(f)|&\leq B_0(f):=    s^4(f), \\
|K^{(1)}(f)|&\leq B_1(f):=   2\pi M s^4(f)\left(\frac{3\sqrt3}{8}+ 2s(f)\right), \\
|K^{(2)}(f)|&\leq B_2(f):=  (2 \pi M)^2 s^4(f) \left(1 + \frac{3\sqrt3}{2} s(f)+ 5 s^2(f)\right), \\
|K^{(3)}(f)|&\leq B_3(f):=  (2 \pi M)^3 s^4(f) \left(c_1 + 6 s(f) + \frac{45\sqrt3}{8} s^2(f) + 15s^3(f)\right), \\
|K^{(4)}(f)|&\leq B_4(f):=  (2 \pi M)^4 s^4(f) \left(\frac{5}{2} + c_2 s(f) + 30 s^2(f) + \frac{45\sqrt3}{2}s^3(f) + \frac{105}{2} s^4(f)\right), 
\end{align*}
where
\begin{align*}
c_1&=
\frac{1}{2} \left(\sin \left(2 \tan ^{-1}\left(\sqrt{\frac{1}{5} \left(\sqrt{129}+12\right)}\right)\right)-2 \sin \left(4 \tan ^{-1}\left(\sqrt{\frac{1}{5} \left(\sqrt{129}+12\right)}\right)\right)\right), 
\\
c_2&=
-4 \sin \left(2 \tan ^{-1}\left(\sqrt{\frac{1}{5} \left(\sqrt{129}+12\right)}\right)\right) \left(4 \cos \left(2 \tan ^{-1}\left(\sqrt{\frac{1}{5} \left(\sqrt{129}+12\right)}\right)\right)-1\right).
\end{align*}
Furthermore, 
$B_\ell(f)$ is decreasing in $f$ on $(0, 1/2]$ and
$B_\ell(\Omega-f)+B_\ell(\Omega+f)$   is increasing in $f$ for any positive $\Omega$ such that $\Omega>f$ and $\Omega+f\leq 1/2$.
\end{lemma}

\begin{proof}
We need the following elementary bound on the sine function for $f \in [0,  \frac{1}{2}]$:
\begin{align}
\sin(\pi f) & \geq f(3-4f^2).
\end{align}
Clearly,  a consequence is $\frac{1}{M|\sin(\pi f)|} \leq s(f),  f \in [-\frac{1}{2},  \frac{1}{2}]\backslash \{0\}$. We use this fact together with explicit expressions for $K^{(\ell)}(f)$ to develop upper bounds. 

When $\ell = 0$, 
$$|K(f)| = \left| \frac{\sin(\pi Mf)}{M\sin(\pi f)} \right|^4 \leq s^4(f).$$

When $\ell = 1$, 
\begin{align*}
K^{(1)}(f) = \frac{2\pi M}{(M\sin(\pi f))^4}\frac{1}{M}
\left( -2 \cot (\pi  f) \sin ^4(\pi  f M)+ 2 \sin ^3(\pi  f M) \cos (\pi  f M)M\right)
\end{align*}
implying
$$
|K^{(1)}(f)|  \leq  2\pi M s^4(f)\left(\frac{3\sqrt3}{8}+ 2s(f)\right), 
$$
since
$
\max_f |2\cos(\pi  f M)\sin(\pi  f M)^3|\leq \frac{3 \sqrt{3}}{8}.
$

When $\ell = 2$, 
\begin{align*}
&K^{(2)}(f)  = \frac{(2 \pi M)^2}{(M\sin(\pi f))^4} \frac{1}{M^2}
 \times\\
& \bigg((2 \cos (2 \pi  f)+3) \csc ^2(\pi  f) \sin ^4(\pi  f M)
 -8 \cot (\pi  f) \sin ^3(\pi  f M) \cos (\pi  f M) M
 +
 \sin ^2(\pi  f M) (2 \cos (2 \pi  f M)+1) M^2
 \bigg)
\end{align*}
implying
$$
|K^{(2)}(f)|   \leq (2 \pi M)^2 s^4(f) \left(1 + \frac{3\sqrt3}{2} s(f)+ 5 s^2(f)\right), 
$$
where we used $\max_f |8 \sin ^3(\pi  f M) \cos (\pi  f M)|=\frac{3 \sqrt{3}}{2}$ and
$\max_f |\sin ^2(\pi  f M) (2 \cos (2 \pi  f M)+1)|=1.$

When $\ell = 3$, 
\begin{align*}
K^{(3)}(f) = \frac{(2 \pi M)^3}{(M\sin(\pi f))^4} \frac{1}{M^3}\times
\bigg(
&-(4 \cos (2 \pi  f)+11) \cot (\pi  f) \csc ^2(\pi  f) \sin ^4(\pi  f M)\\
&+6 (2 \cos (2 \pi  f)+3) \csc ^2(\pi  f) \sin ^3(\pi  f M) \cos (\pi  f M)M\\
&-6 \cot (\pi  f) \sin ^2(\pi  f M) (2 \cos (2 \pi  f M)+1)\sin (4 \pi  f M)M^2
-\frac{1}{2} \sin (2 \pi  f M)M^3
\bigg)
\end{align*}
implying
$$
|K^{(3)}(f)|   \leq (2 \pi M)^3 s^4(f) \left(c_1 + 6 s(f) + \frac{45\sqrt3}{8} s^2(f) + 15s^3(f)\right), 
$$
by recognizing the following upper bounds:
\begin{center}
\begin{tabular}{ll} 
	$\max_{f\in[0, 1/2]} |(4 \cos (2 \pi  f)+11) \cos (\pi  f)|=15$, &  
	$\max_f |6 \sin ^2(\pi  f M) (2 \cos (2 \pi  f M)+1)|=6$, 
    \\ 
	$\max_{f\in[0, 1/2]} |6 (2 \cos (2 \pi  f)+3)|=30$	, &  
	$\max_f |\sin (4 \pi  f M)-(1/2) \sin (2 \pi  f M)|= c_1, $
   \\ 
	$\max_f |\sin ^3(\pi  f M) \cos (\pi  f M)|=  3 \sqrt{3}/{16}.$ &
\end{tabular} 
\end{center}

When $\ell=4$, 
\begin{align*}
K^{(4)}(f)=\frac{(2 \pi M)^4}{(M\sin(\pi f))^4} \frac{1}{M^4}\times
\bigg(
&\frac{1}{2} (49 \cos (2 \pi  f)+4 \cos (4 \pi  f)+52) \csc ^4(\pi  f) \sin ^4(\pi  f M)\\
&-8 (4 \cos (2 \pi  f)+11) \cot (\pi  f) \csc ^2(\pi  f) \sin ^3(\pi  f M) \cos (\pi  f M)M\\
&+6 (2 \cos (2 \pi  f)+3) \csc ^2(\pi  f) \sin ^2(\pi  f M) (2 \cos (2 \pi  f M)+1)M^2\\
&-4 \cot (\pi  f) \sin (2 \pi  f M) (4 \cos (2 \pi  f M)-1)M^3\\
&+(2 \cos (4 \pi  f M)-\frac{1}{2} \cos (2 \pi  f M))M^4
\bigg)
\end{align*}
implying
$$
|K^{(4)}(f)|   \leq (2 \pi M)^4 s^4(f) \left(\frac{5}{2} + c_2 s(f) + 30 s^2(f) + \frac{45\sqrt3}{2}s^3(f) + \frac{105}{2} s^4(f)\right), 
$$
which follows from the following upper bounds:
\begin{align*}
\begin{matrix*}[l]
\max_{f\in[0, 1/2]}\frac{1}{2} (49 \cos (2 \pi  f)+4 \cos (4 \pi  f)+52)=105/2, &  
\max_f |\sin ^2(\pi  f M) (2 \cos (2 \pi  f M)+1)|=1, 
\\ 
\max_{f\in[0, 1/2]}|8 (4 \cos (2 \pi  f)+11) \cos (\pi  f)|=120, &  
\max_f |4 \sin (2 \pi  f M) (4 \cos (2 \pi  f M)-1)|=c_2, 
\\ 
\max_f |\sin ^3(\pi  f M) \cos (\pi  f M)|=3\sqrt3/16, &  
\max_f |2 \cos (4 \pi  f M)- {1}/{2} \cos (2 \pi  f M)|=5/2, 
\\ 
\max_{f\in[0, 1/2]} |6 (2 \cos (2 \pi  f)+3)|= 30.&  
\end{matrix*}
\end{align*}
 
Finally,  $s(f)$ is nonnegative and is decreasing in $(0,  1/2]$ since $s'(f)$ is negative on $(0,  1/2)$. Therefore,  the $k$th power $s^k(f)$ is decreasing in $(0,  1/2])$,  which further implies that $B_\ell(f), \ell=0, 1, 2, 3, 4$ is decreasing in $(0,  1/2].$
In addition,  since $s(f)$ is convex in $(0, 1/2]$,  $s^k(f)$ is also convex as a consequence of the composition rule of convex and monotonic functions. Combining the convex and decreasing property of $s^k(f)$ on $(0, 1/2]$ and then applying arguments similar to those in~\cite[Lemma 2.6]{Candes:2014br},  we conclude that
$B_\ell(\Omega-f)+B_\ell(\Omega+f)$ is increasing in $f$ for any positive $\Omega$ such that $\Omega>f$ and $\Omega+f\leq 1/2$.

\end{proof}

\subsection{Bounding the Sums of   the Jackson Kernel }

Without loss of generality,  we assume $0 \in T$ and develop bounds on  $\sum_{f_i\in T\backslash\{0\}}|K^{(\ell)}(f-f_i)|,  \ell\in\{0, 1, 2, 3, 4\}$ when $f$ lives in a neighborhood around $0$. It is easy to verify the following lemma based on the properties of $|K^{(\ell)}(f)|,  \ell=0, 1, 2, 3, 4$ in  Lemma~\ref{lem:bounds:drivative:K}. The proof parallels that of~\cite[Lemma 2.7]{Candes:2014br} and is omitted here.
\begin{lemma}\label{lem:bounds:sum:K}
Suppose $0\in T$ and $f_+$ is the smallest positive frequency in $T$. Let $\Delta:=\Delta(T)\geq\Delta_{\min}$ and $f\in[0, \bar f]$ where $\bar{f} \leq f_+/2$.  Then for $\ell\in\{0, 1, 2, 3, 4\}$, 
\begin{align}
\sum_{f_i\in T\backslash\{0\}}|K^{(\ell)}(f-f_i)|
\leq F_\ell(\Delta,  f)
:=&F^+_\ell(\Delta, f)+F^-_\ell(\Delta, f)\leq F_\ell(\Delta_{\min},  \bar f)
\nn
\end{align}
with
\begin{align*}
F^+_\ell(\Delta, f)&=\max\left\{\max_{\Delta\leq \xi \leq3\Delta_{\min}}|K^{(\ell)}(f-\xi )|, B_\ell(3\Delta_{\min}-f)\right\}+\sum_{j=2}^{\lfloor\frac{1}{2\Delta_{\min}}\rfloor}B_\ell(j\Delta_{\min}-f), \\
F^-_\ell(\Delta, f)&=\max\left\{\max_{\Delta\leq \xi \leq3\Delta_{\min}}|K^{(\ell)}(\xi)|, B_\ell(3\Delta_{\min})\right\}+\sum_{j=2}^{\lfloor\frac{1}{2\Delta_{\min}}\rfloor}B_\ell(j\Delta_{\min}+f).
\end{align*}
$F_\ell(\Delta, f)$ is decreasing in $\Delta$.  When  $\Delta$ is fixed as $\Delta_{\min}$,  $F_\ell(\Delta_{\min}, f)$  is increasing in $f$.
\end{lemma}

The following lemma provides bounds on $\sum_{f_i\in T}|K^{(\ell)}(f-f_i)|$ for $\ell\in\{0, 1, 2, 3, 4\}$ and is a direct consequence of the decreasing property of $B_\ell(\cdot)$.
\begin{lemma}\label{lem:bounds:sum:K:far}
Suppose $0\in T$,  $f_+$ is the smallest positive frequency in $T$  and $f\in[\underline f,  f_+-\bar f]$.  Then for $\ell\in\{0, 1, 2, 3, 4\}$, 
\begin{align}
\sum_{f_i\in T}|K^{(\ell)}(f-f_i)|
&\leq W_\ell(\underline{f}, \bar{f}):=  \sum_{j=0}^{\lfloor\frac{1}{2\Delta_{\min}}\rfloor} B_\ell(j\Delta_{\min}+\underline{f})+\sum_{j=0}^{\lfloor\frac{1}{2\Delta_{\min}}\rfloor} B_\ell(j\Delta_{\min}+\bar{f}).
\nn
\end{align}
\end{lemma}

\subsection{Numerical Bounds on the Jackson Kernel Sums}\label{sec:numerical:bound}

Suppose $0\in T$. Then by Lemma~\ref{lem:bounds:sum:K} we can   bound $\sum_{f_j\in T\backslash\{0\}}|K^{(\ell)}(f-f_j)|$ for $f\in[0, \bar f]$ as:
$$ \sum_{f_j\in T\backslash\{0\}}|K^{(\ell)}(f-f_j)|\leq F_\ell(\Delta_{\min}, \bar f).$$
We list the values of $F_\ell(\Delta_{\min}, \bar f)$ for different $\bar f$ in Table~\ref{tbl:bounds:1}.

We can use Lemma~\ref{lem:bounds:sum:K:far} to bound $\sum_{f_j\in T}|K^{(\ell)}(f-f_j)|$ for $f\in[\underline{f},   f_+-\bar{f}]$ as
$$
\sum_{f_j\in T} |K^{(\ell)}(f-f_j)| \leq W_\ell(\underline{f}, \bar{f}).
$$
We list the values of $W_\ell(\underline{f}, \bar{f})$ for different $\underline{f}, \bar{f}$ in Table~\ref{tbl:bounds:3}.

Finally,  we  list
several numerical upper bounds on $|K^{(\ell)}(f)|$ and $K''(f)$ over different intervals in Table~\ref{tbl:bounds:2},  which directly follow from
\cite[equations (2.21)-(2.24)]{Candes:2014br} and numerical computations.

%%%%%%%%%%%%%%%%%%%%%%%%%%%%%%%%%%%%%%%%%%%%%%%%%%%%%%%%%%%%%%
 \begin{table*}[h!t]
\renewcommand{\arraystretch}{1.3}
\label{tbl:bounds:1}
\centering
\begin{tabular}{|l||c|c|c|c|c|}
\hline
$f$ &$F_{0}(2.5/n, f)$ &$F_{1}(2.5/n, f)$ &$F_{2}(2.5/n, f)$ &$F_{3}(2.5/n, f)$ & $F_{4}(2.5/n, f)$
\\
\hline
0 &0.00755 &$0.01236n$ &$0.05610n^{2}$ & $0.28687n^{3}$ & $1.48634n^{4}$ \\ %\hline
$0.002/n$ &0.00755 &$0.01236n$ &$0.05610n^{2}$ & $0.28687n^{3}$ & $1.48634n^{4}$ \\ %\hline
$0.24/n$ &0.00757 &$0.01241n$ &$0.05637n^{2}$ & $0.28838n^{3}$ & $1.67097n^{4}$ \\ %\hline
$0.2404/n$ &0.00757 &$0.01241n$ &$0.05637n^{2}$ & $0.28838n^{3}$ & $1.67100n^{4}$ \\ %\hline
$0.75/n$ &0.00772 &$0.01450n$ & $0.12639n^{2}$ & $1.07987n^{3}$ & $6.57069n^{4}$ \\ %\hline
$0.7504/n$ &0.00772 &$0.01454n$ &$0.12675n^{2}$ & $1.08211n^{3}$ & $6.57595n^{4}$
\\
\hline
\end{tabular}
\caption{Numerical upper bounds on $F_\ell(2.5/n, f)$.  }
\end{table*}

 \begin{table}[h!t]
\renewcommand{\arraystretch}{1.3}
\label{tbl:bounds:3}
\centering
\begin{tabular}{|l|l||c|c|c|}
\hline
$f_{1}$ &$f_{2}$ &$W_{0}(f_{1}, f_{2})$ &$W_{1}(f_{1}, f_{2})$ &$W_{2}(f_{1}, f_{2})$
\\
\hline
$0.7496/n$ &$1.25/n$ &0.71059 &$5.2265n$ &$48.0330n^{2}$ \\ %\hline
$0.75/n$ &$1.25/n$ &0.70859 &$5.2084n$ & $47.8388n^{2}$
\\
\hline
\end{tabular}
\caption{Numerical upper bounds on $W_\ell(f_1, f_2).$}
\end{table}

 \begin{table*}[h!t]
\renewcommand{\arraystretch}{1.3}
\label{tbl:bounds:2}
\centering
\begin{tabular}{|l||c|c|c|c|c|c|}
\hline
$f$ &$|K(f)|$ & $|K'(f)|$ &$|K''(f)|$ &$|K'''(f)|$ &$|K''''(f)|$ &$K''(f)$
\\
\hline
$[0, 0.002/n]$ & 1 & $0.00658n$ &$3.290n^{2}$ &$0.0649394n^{3}$ & & \\ %\hline
$[0, 0.24/n]$ & 1 & ${0.789569}n$ &$3.290n^{2}$ &$7.79273n^{3}$ & &$-2.35084 n^{2}$ \\ %\hline
$[0, 0.2404/n]$ & 1 & ${0.790885}n$ &$3.290n^{2}$ &$7.80572 n^{3}$ & ${{29.2227}}n^{4}$ & \\ %\hline
$[0.2396/n, 0.7504/n]$ & {0.90951} & $2.46872n$ &$3.290n^{2}$ & & &
\\
\hline
\end{tabular}
\caption{Numerical upper bounds on $|K^{(\ell)}(f)|$ and $K''(f)$.}
\end{table*}

%%%%%%%%%%%%%%%%%%%%%%%%%%%%%%%%%%%%%%%%%%%%%%%

\subsection{Controlling the Jackson Kernel Matrices}\label{sec:A:5}
In this section,  we derive several consequences of the joint frequency-coefficient vector $\btheta = (\f, \u, \v)$ living in the neighborhood $\N^\star$ of the true joint frequency-coefficient vector $\btheta^\star=(\f^\star, \u^\star, \v^\star)$. Recall that $\N^\star$ contains all $\btheta$ that is close to $\btheta^\star$ in the $\ell_{\hinfty}$ norm:
\begin{align}\label{eqn:theta:near:theta}
\N^\star = \{\btheta: \|\btheta-\btheta^\star\|_{\hinfty}\leq {X^\star\gamma_0}/{\sqrt 2}\}.
\end{align}
Recall that the weighted $\ell_\infty$ norm $\|\cdot\|_{\hinfty}$ is defined by $\|(\f, \u, \v)\|_\hinfty:= \|(\mS\f, \u, \v)\|_{{\infty}}$ with $\mS:=\sqrt{\tau}\diag(|\c^\star|)$.
We remark that all the results in this section still hold for the bigger neighborhood $\hat{\N}$ defined by replacing $X^\star$ with $\hat{X}=X^\star+35.2$. Indeed,  for the results to hold,  the key requirement on $\btheta$ is $\|\f-\f^\star\|_{\infty}\leq0.002/n$. This condition holds for both regions because as we will show later
$$\|\f-\f^\star\|_\infty\leq
\begin{cases}
0.4X^\star\gamma &\text{for }\btheta\in\N^\star, \\
0.4\hat{X}\gamma &\text{for }\btheta\in\hat{\N}.
\end{cases}$$
Invoking the SNR condition~\eqref{eqn:snr},  we conclude that the two upper bounds are much smaller than $0.002/n$ in both cases.

Our first set of results bound the distances between the parameters in   $\btheta^\star$  and  $\btheta$:
for each $j = 1,  \ldots,  k$
\begin{equation}
\begin{aligned}
\frac{|c_j-c_j^\star|}{|c^\star_j|}
&\stack{\ding{172}}{\leq}  X^\star  \gamma, 
\\
\frac{|c_j|}{|c_j^\star|}
&\stack{\ding{173}}{\leq}  1+X^\star  \gamma, 
\\
\left|\left( {|c_j|}/{|c_j^\star|}\right)^2-1\right|&\stack{\ding{174}}{\leq}  X^\star   \gamma(2+X^\star   \gamma), 
\\
|f_j-f_j^\star|
& \stack{\ding{175}}\leq  {X^\star   \gamma}/{ \sqrt{2{ \tau }}}\stack{\ding{176}}{\leq} 0.4X^\star   \gamma/n.
\label{eqn:parameter:close}
\end{aligned}
\end{equation}
For \ding{172} to hold,  first note  $\|\u-\u^\star\|_\infty\leq X^\star\gamma_0/\sqrt{2}$ and
$\|\v-\v^\star\|_\infty\leq X^\star\gamma_0/\sqrt{2}$ by~\eqref{eqn:theta:near:theta}. Also note
\begin{align*}
\|\c-\c^\star\|_\infty^2
&=\max_\ell  [(u_\ell-u^\star_\ell)^2+(v_\ell-v_\ell^\star)^2]\\
&\leq \max_\ell (u_\ell-u^\star_\ell)^2+ \max_\ell(v_\ell-v_\ell^\star)^2\\
&=\|\u-\u^\star\|_\infty^2+\|\v-\v^\star\|_\infty^2
\leq 2(X^\star\gamma_0/\sqrt{2})^2=(X^\star\gamma_0)^2.
\end{align*}
Finally \ding{172} follows since $\max_j  {|c_j-c_j^\star|}/{|c^\star_j|}\leq\|\c-\c^\star\|_\infty/ c^\star_{\min}$ and $\gamma=\gamma_0/{c_{\min}^\star}$. After we show \ding{172}, 
\ding{173} follows from  $|c_j|/|c_j^\star|=|c_j-c_j^\star+c_j^\star|/|c_j^\star| $ and  the triangle inequality. \ding{174} follows from
$|( |c_j|/|c_j^\star|)^2-1|=|(|c_j|/|c_j^\star|+1)(|c_j|/|c_j^\star|-1)|$.
\ding{175} follows from the definition of the $\ell_\hinfty$ norm:
\begin{align*}
& \|S(\f-\f^\star)\|_\infty\leq  {X^\star\gamma_0}/{\sqrt2}\\
\Longrightarrow&\|\sqrt{\tau}\diag(|\c^\star|)(\f-\f^\star)\|_\infty\leq  {X^\star\gamma_0}/{\sqrt2}\\
\Longrightarrow& |f_j-f^\star_j|  \leq  {X^\star\gamma_0/|c_j^\star|}/\sqrt{2 \tau}, \ \forall j\\
\Longrightarrow& |f_j-f^\star_j|  \leq  {X^\star\gamma_0/c^\star_{\min}}/\sqrt{2 \tau}= {X^\star\gamma}/\sqrt{2 \tau}, \ \forall j.
\end{align*}
Finally \ding{176} holds due to the fact that $\tau\geq 3.289n^2$ for $n\geq130$ by~\eqref{eqn:sec:A:tau} and hence
$$1/\sqrt{2\tau}\leq 1/ \sqrt{2 (3.289)}/n\leq 0.3899/n\leq 0.4/n.$$

Next,  we present the second class of results that quantify the well-conditionedness of the Jackson kernel matrices $\D_\ell(\f), \ \ell=0, 1, 2$.
Such results are instrumental to dual certificate construction~\cite{Candes:2014br}. Since the minimal separation $\Delta(T)$ is a key quantity affecting the well-conditionedness,  we first show that
those frequencies $T^\lambda:=\{f^\lambda_\ell\}$ and $\hat{T}:=\{\hat{f}_\ell\}$  in Lemma~\ref{lem:fix1} and Lemma~\ref{lem:fix2}  satisfy a  separation condition,  provided $T^\star = \{f^\star_\ell\}$ satisfy a slightly stronger separation condition. The proof is given in Appendix~\ref{sec:proof:resolution}.

\begin{lemma}\label{lem:resolution}
 Let the separation condition~\eqref{eqn:separation} and the SNR condition~\eqref{eqn:snr} hold. Then both the frequencies $T^\lambda=\{f^\lambda_\ell\}$ returned by the first fixed point map~\eqref{eqn:Map:1}  and the frequencies $\hat{T}=\{\hat{f}_\ell\}$ generated by the second fixed point map~\eqref{eqn:Map:2} have minimal separations at least $2.5/n$. Furthermore,  the intermediate frequencies defined by $\tilde{T}=\{\tilde{f}_\ell\}_{\ell=1}^k$ with each $\tilde{f}_\ell \in [f_\ell^\star, f_\ell^\lambda]$ or $[f_\ell^\lambda,  f_\ell^\star]$ and the second intermediate frequencies $\tilde{T}^\lambda:=\{\tilde{f}_\ell\}_{\ell=1}^k$ with each $\tilde{f}_\ell\in[f_\ell^\lambda, \hat{f}_\ell]$ or $[\hat{f}_\ell,  f_\ell^\lambda]$ also have minimal separations at least $2.5/n$:
\begin{align*}
\min\{\Delta(T^\lambda),  \Delta(\tilde{T}),  \Delta(\hat{T}),  \Delta(\tilde{T}^\lambda)\} \geq 2.5/n.
\end{align*}
\end{lemma}

Now we are ready to  provide numerical bounds related to the well-conditionedness of the Jackson kernel matrices $\D_\ell(\f), \ \ell=0, 1, 2$:
\begin{equation}
\begin{aligned}
\|\eye-\D_0(\f)\|_{\infty, \infty}
&\stack{\ding{172}}{\leq}  F_0(2.5/n, 0)\stack{\ding{175}}{\leq}  0.00755, 
\\
\left\| \D_1(\f)/\sqrt{{\tau}}\right\|_{\infty, \infty}
&\stack{\ding{173}}{\leq}    {F_1(2.5/n, 0)}/\sqrt{{\tau}}\stack{\ding{176}}{\leq}0.01236n/\sqrt{\tau}\leq 0.00682, 
\\
\left\|\eye-\left(- \D_2(\f)/\tau\right)\right\|_{\infty, \infty}
&\stack{\ding{174}}{\leq}   {F_2(2.5/n, 0)}/{\tau}\stack{\ding{177}}{\leq}  0.05610n^2/\tau\leq  0.0171, 
\label{eqn:kernel:matrix:small:norm}
\end{aligned}
\end{equation}
where \ding{172},  \ding{173} and \ding{174} follow because the diagonal entries of these kernel matrices are given by $[\D_0(\f, \f)]_{\ell, \ell}=K(0)=1$,  $[\D_1(\f, \f)]_{\ell, \ell}=K'(0)=0$ and $[\D_2(\f, \f)]_{\ell, \ell}=K''(0)=-\tau$~\cite[Section IV.A]{Tang:2013fo}. Hence,  it suffices to compute
$\sum_{f_i\in T\backslash\{\zeta\}}|K^{(\ell)}(\zeta-f_i)|$ for $\zeta\in T$ which can be bounded by $F_\ell(2.5/n, 0)$ according to Lemma~\ref{lem:bounds:sum:K} since $\Delta(T)\geq 2.5/n$
by Lemma~\ref{lem:resolution}. The inequalities \ding{175},  \ding{176} and \ding{177} follow from the upper bounds on $F_\ell(2.5/n, 0)$ in Table~\ref{tbl:bounds:1} and the fact that $\tau\geq 3.289n^2$ for $n\geq130$ by~\eqref{eqn:sec:A:tau}.

To control the $\ell_{\infty,  \infty}$ distance between two kernel matrices,  say $\D_0(\f)$ and $D_0(\f,  f^\star)$,  we  apply the mean value theorem and Lemma~\ref{lem:bounds:sum:K}:
\begin{equation}
\begin{aligned}
\|\D_0(\f)-\D_0(\f, \f^\star)\|_{\infty,  \infty}
&\stack{\ding{172}}{=} \|\D_0(f_1, \f)-\D_0(f_1, \f^\star)\|_1
\\
&\leq \sum_\ell |K(f_\ell-f_1)-K(f^\star_\ell-f_1)| \\
&\stack{\ding{173}}{=}    \sum_\ell |K'(\tilde{f}_\ell-f_1)(f_\ell-f^\star_\ell)|
\\
&\leq (|K'(\tilde{f}_1-f_1)|+\sum_{\ell\neq 1} |K'(\tilde{f}_\ell-f_1)| )\|\f-\f^\star\|_\infty
\\
&\stack{\ding{174}}{\leq}  (F_1\left( {2.5}/{n},  {0.002}/{n}\right)+\max_{f\in[0, 0.002/n]}|K'(f)|) \|\f-\f^\star\|_\infty
\\
&\stack{\ding{175}}{\leq} (0.01236n+0.00658n)(0.4X^\star\gamma/n)
= 0.00758X^\star\gamma, 
 \label{eqn:bound:diff:kernel:c}
\end{aligned}
\end{equation}
where \ding{172} follows since by rearranging indices if necessary,  we can assume without loss of generality that the maximum absolute row sum of $\D_0(\f)-\D_0(\f, \f^\star)$ happens at the first row; \ding{173} holds because we applied the mean value theorem for some $\tilde{f}_\ell$ between $f_\ell$ and $f_\ell^\star$; \ding{174} follows from the monotonic property of $F_\ell(2.5/n, f)$ in Lemma~\ref{lem:bounds:sum:K} by taking into account that  $\Delta(\tilde T)\geq 2.5/n$ (by Lemma~\ref{lem:resolution})  and $\|\tilde \f-\f\|_\infty\leq\|\f^\star-\f\|_\infty\leq {0.4X^\star\gamma}/{n}\leq {0.002}/{n}$. \ding{175} follows from the upper bounds on $F_1(2.5/n, 0.002/n)$ in Table~\ref{tbl:bounds:1} and $\max_{f\in[0, 0.002/n]}|K'(f)|$ in Table~\ref{tbl:bounds:2}.

Applying the similar arguments as the step \ding{174},  we can get a more general result as follows 
\begin{lemma}\label{lem:control:middle}
Let an arbitrary cluster of points $T:=\{{f_j}\}$ satisfy the separation condition of $\Delta(T)\geq 2.5/n$.
Assume $\underline{f}\leq|f-{f}_r|\leq \bar{f}$ for an arbitrary $f_r\in T$. Then, 
\begin{align}
\sum_j |K^{(\ell)}({f}_j-f)|\leq F_\ell(2.5/n, \bar{f})+\max_{f\in[\underline{f}, \bar{f}]}|K^{(\ell)}(f)|. \label{eqn:rmk:1}
\end{align}
\end{lemma}

To control $\|\D_\ell(\f, \f^\star)-\D_\ell(\f)\|_{\infty,  \infty}$ in a similar manner for $\ell = 1,  2$,  we note that $\|\btheta-\btheta^\star\|_{\hinfty}\leq  {X^\star\gamma_0}/{\sqrt 2}$ and both $T$ and $\tilde{T}$ are well-separated: $\Delta(T)\geq 2.5/n$ and $\Delta(\tilde{T})\geq2.5/n$ with $\tilde{T}$ composed of certain ``middle" frequencies  $\tilde{f}_\ell\in[f_\ell, f_\ell^\star]$ or $[f_\ell^\star,  f_\ell]$.
Then using Lemma~\ref{lem:control:middle},  we upperbound $\|\D_\ell(\f, \f^\star)-\D_\ell(\f)\|_{\infty,  \infty}$  as follows 
\begin{equation}
\begin{aligned}
\frac{1}{\sqrt{\tau}}\|\D_1(\f, \f^\star)-\D_1(\f)\|_{\infty,  \infty}
\stack{\ding{172}}{\leq} &   {1}/{\sqrt{\tau}}\left(F_2( {2.5}/{n},  {0.002}/{n}\right)+\max_{f\in[0, 0.002/n]}|K''(f)|) \|\f-\f^\star\|_\infty
\\
\stack{\ding{173}}\leq &   ({1}/{\sqrt{3.289n^2}})(0.05610n^2+3.290n^2)(0.4X^\star\gamma/n)
\leq 0.73802 X^\star\gamma, 
\label{eqn:bound:diff:kernel:b0}
\end{aligned}
\end{equation}
where \ding{172}  follows by Lemma~\ref{lem:control:middle} and \ding{173} follows from the fact that $\tau\geq 3.289n^2$ for $n\geq130$ in~\eqref{eqn:sec:A:tau} and by
combining the upper bound on $F_2(2.5/n, 0.002/n)$  in Table~\ref{tbl:bounds:1}  and the upper bound on $\max_{f\in[0, 0.002/n]}|K{''}(f)|$  in Table~\ref{tbl:bounds:2}. Similarly,  following from  Lemma~\ref{lem:control:middle} and the mean value theorem,   by combining  the upper bound on $F_3(2.5/n, 0.002/n)$  in Table~\ref{tbl:bounds:1}  and the upper bound on $\max_{f\in[0, 0.002/n]}|K^{(3)}(f)|$  in Table~\ref{tbl:bounds:2},  we have
\begin{equation}
\begin{aligned}
\frac{1}{ {\tau}}\|\D_2(\f, \f^\star)-\D_2(\f)\|_{\infty,  \infty}
\leq& \frac{1}{ {\tau}}(F_3( {2.5}/{n},  {0.002}/{n})+\max_{f\in[0, 0.002/n]}|K'''(f)|) \|\f-\f^\star\|_\infty
\\
\leq& ({1}/{ {3.289n^2}})(0.28687n^3+0.0649394n^3)(0.4X^\star\gamma/n)
= 0.04279X^\star\gamma.
\label{eqn:bound:diff:kernel:b1}
\end{aligned}
\end{equation}

To control $\|\D_\ell(\f^\star)-\D_\ell(\f)\|_{\infty,  \infty} $,  we rewrite $\D_\ell(\f^\star)-\D_\ell(\f)$ as $$\D_\ell(\f^\star)-\D_\ell(\f)=\D_\ell(\f^\star)-\D_\ell(\f^\star, \f)+\D_\ell(\f^\star, \f)-\D_\ell(\f).$$
Then,  the desired results follow from the triangle inequality of the $\ell_{\infty, \infty}$ norm:
\begin{equation}
\begin{aligned}
\|\D_0(\f^\star)-\D_0(\f)\|_{\infty,  \infty}
&\leq\|\D_0(\f^\star)-\D_0(\f^\star, \f)\|_{\infty,  \infty}+\|\D_0(\f^\star, \f)-\D_0(\f)\|_{\infty,  \infty}
\\
&\stack{\ding{172}}{\leq}2(0.00758X^\star\gamma)
=0.01516X^\star\gamma, 
\label{eqn:bound:diff:kernel:b2}
\end{aligned}
\end{equation}
where \ding{172} follows from\eqref{eqn:bound:diff:kernel:c} and an exchange of the roles of $\f$ and $\f^\star$;
\begin{equation}
\begin{aligned}
\frac{1}{\sqrt{\tau}}\|\D_1(\f^\star)-\D_1(\f)\|_{\infty,  \infty}
&\leq \frac{1}{\sqrt{\tau}}\|\D_1(\f^\star)-\D_1(\f^\star, \f)\|_{\infty,  \infty}+ \frac{1}{\sqrt{\tau}}\|\D_1(\f^\star, \f)-\D_1(\f)\|_{\infty,  \infty}
\\
&\stack{\ding{172}}{\leq}2(0.73802 X^\star\gamma)
=1.47604X^\star\gamma, 
\label{eqn:bound:diff:kernel:b3}
\end{aligned}
\end{equation}
where \ding{172} follows from~\eqref{eqn:bound:diff:kernel:b0};
\begin{equation}
\begin{aligned}
\frac{1}{ {\tau}}\|\D_2(\f^\star)-\D_2(\f)\|_{\infty,  \infty}
\leq& \frac{1}{ {\tau}}\|\D_2(\f^\star)-\D_2(\f^\star, \f)\|_{\infty,  \infty}+ \frac{1}{ {\tau}}\|\D_2(\f^\star, \f)-\D_2(\f)\|_{\infty,  \infty}
\\
\stack{\ding{172}}{\leq}& 2(0.04279X^\star\gamma)
=0.08558X^\star\gamma ,  \label{eqn:bound:diff:kernel:b4}
\end{aligned}
\end{equation}
where \ding{172} follows from~\eqref{eqn:bound:diff:kernel:b1}.

Then following from Eq.~\eqref{eqn:parameter:close}, ~\eqref{eqn:bound:diff:kernel:b0}-\eqref{eqn:bound:diff:kernel:b1} and~\eqref{eqn:bound:diff:kernel:b2}-\eqref{eqn:bound:diff:kernel:b4},  and together with the sub-multiplicative property of the $\ell_{\infty, \infty}$ norm,  we have
\begin{equation} \label{eqn:A:DC-DC1}
\begin{aligned}
\left\|\frac{1}{\sqrt{{\tau}}}\diag(1./|\c^\star|)[\D_1(\f, \f^\star)\c^\star-\D_1(\f)\c]\right\|_\infty
\leq &\frac{1}{\sqrt{\tau}}\|\D_1(\f, \f^\star)-\D_1(\f)\|_{\infty, \infty}\|1./\c^\star\|_\infty\|\c^\star\|_\infty
\\&+\frac{1}{\sqrt{\tau}}\|\D_1(\f)\|_{\infty, \infty}\|1./\c^\star\|_\infty\|\c-\c^\star\|_\infty
\\
\leq & (0.73802 X^\star\gamma)B^\star+(0.01236n/\sqrt\tau)B^\star X^\star\gamma
=0.75038B^\star X^\star\gamma, 
\end{aligned}
\end{equation}
where the last but one line follows from $\|1./\c^\star\|_\infty\leq1/{c_{\min}^\star}$ and $\gamma=\gamma_0/{c_{\min}^\star}$. Here and throughout the rest of the paper,   we use $1./\x$,  $1./|\x|$,  $\y./\x$,  $|\y|./|\x|$,  $\x\odot\y$ and $\frac{1}{\x}$,  $\frac{1}{|\x|}$,  $\frac{\y}{\x}$,  $\frac{|\y|}{|\x|}$ in the sense of pointwise arithmetic operations,  here $\x, \y$ stand for any vectors of the same length.

We apply similar arguments to develop the following bound 
\begin{equation} \label{eqn:A:DC-DC2}
\begin{aligned}
\left\|\frac{1}{ {{\tau}}}\diag(1./|\c^\star|)[\D_2(\f, \f^\star)\c^\star-\D_2(\f)\c]\right\|_\infty
\leq&\frac{1}{ {\tau}}\|\D_2(\f, \f^\star)-\D_2(\f)\|_{\infty, \infty}\|1./\c^\star\|_\infty\|\c^\star\|_\infty
\\&
+\frac{1}{ {\tau}}\|\D_2(\f)\|_{\infty, \infty}\|1./\c^\star\|_\infty\|\c-\c^\star\|_\infty
\\
\leq& (0.08558 X^\star \gamma)B^\star+(1.05610)B^\star(X^\star\gamma )
\leq 1.14168B^\star X^\star\gamma.
\end{aligned}
\end{equation}

%\note{Sometimes we use $\|\cdot\|_{\infty,  \infty}$ to denote the matrix $\ell_{\infty,  \infty}$ norm,  while sometimes we use $\|\cdot\|_\infty$. Qiuwei,  can you change all of them to $\|\cdot\|_{\infty,  \infty}$ so that the notations are consistent?}

\section{Bounding the Dual Atomic Norm of Gaussian Noise}\label{sec:B}
In this section,  we develop an upper bound on the dual atomic norm of the weighted Gaussian noise $\mZ\w\sim\N(\mathbf{0}, \sigma^2\mZ^2)$ for the positive definite diagonal matrix $\mZ$ with $[\mZ]_{\ell, \ell}=\frac{g_M(\ell)}{M}$.  First
following~\cite[C.4 with $N\geq4\pi(2n+1)$]{bhaskar2013atomic},  we get
\begin{align}
\sup_{f\in\TT} \left|\a(f)^H\mZ\w\right|&\leq 2\max_{m=0,  \ldots,  N-1}|S_m|,  \label{eqn:sec:B:1}
\end{align}
where $\{S_m\}_{m=0}^{N-1}$ are $N$ equispaced samples of the continuous function $\a(f)^H\mZ\w$ defined on $\TT=[0, 1]$:
\begin{align*}
S_m:
&=\a(\frac{m}{N})^H\mZ\w
=\sum_{\ell=-n}^n \frac{g_M(\ell)}{M}w_\ell e^{-i2\pi\ell\frac{m}{N}}.
\end{align*}
Since $\{w_\ell\}$ are i.i.d. Gaussian variables with mean zero and variance $\sigma^2$,  we have that each $S_m$ is a Gaussian variable with mean zero and variance given by $\var(S_m):=\sum_{\ell=-n}^n \left(\frac{g_M(\ell)}{M}\right)^2\sigma^2.$ The main idea next is first to compute an upper bound (denoted by $\overline{\Pi}$) on the variance $\var(S_m)$  and then apply the Gaussian upper deviation inequality~\cite[Eq.  (7.8)]{foucart2013mathematical}
\begin{align}\label{eqn:sec:B:tail}
\P\left[|S_m|\geq t \sqrt{\overline{\Pi}}\right]\leq e^{-t^2/2}
\end{align}
 to get a high-probability upper bound on $|S_m|$. To evaluate $\overline{\Pi}$,  it is instructive to first note
\begin{align*}
g_M(\ell)=\frac{1}{M}\sum_{k=\max(\ell-M, -M)}^{\min(\ell+M, M)}\left(1-\bigg|\frac{k}{M}\bigg|\right)
\left(1-\bigg|\frac{\ell-k}{M}\bigg|\right), 
\end{align*}
with $\ell=-2M, \ldots, 2M, $
which is the convolution of two triangle functions:
\begin{align}\label{sec:B:QW:triangle}
g_M(\ell)= \frac{1}{M}\tmop{Tri}_M(\ell)\ast \tmop{Tri}_M(\ell), ~\ell=-2M, \ldots, 2M.
\end{align}
Here the triangle function is defined by $\tmop{Tri}_M(\ell):=1-\frac{|\ell|}{M}, \ \ell=-M, \ldots, M$ and $\ast$ represents the convolution operator. Apparently
$\var(S_m)$ is the squared $\ell_2$ norm of the vector $\g_M:=[g_M(-2M), \ldots, g_M(2M)]^T$ scaled by $ {\sigma^2}/{M^2}$. Since by Eq.~\eqref{sec:B:QW:triangle},  $\g_M$ is the convolution of two (the same) triangular vectors $\h_M:=[\tmop{Tri}_M(-M), \ldots, \tmop{Tri}_M(M)]$ and then scaled by $1/M$,  we obtain an upper bound on $\var(S_m)$ by applying Young's inequality $\|\f\ast\g\|_r\leq\|\f\|_p\|\g\|_q$ where $r^{-1}=p^{-1}+q^{-1}-1$ and setting $r=2, p=2, q=1$:
\begin{equation}\label{eqn:sec:B:Pi}
\begin{aligned}
\var(S_m)&=\sum_{\ell=-n}^n \left(\frac{g_M(\ell)}{M}\right)^2\sigma^2
{=}\frac{\sigma^2}{M^2}\|\g_M\|_2^2
%\\
%&
{=}\frac{\sigma^2}{M^4}\|\h_M\ast \h_M\|_2^2
{\leq}\frac{\sigma^2}{M^4}\|\h_M\|_1^2 \|\h_M\|_2^2.
\end{aligned}
\end{equation}
Therefore,  to bound $\var(S_m)$,  it remains to bound $\|\h_M\|_1^2$ and $\|\h_M\|_2^2:$
\begin{align}
\|\h_M\|_1^2 =\left(\sum_{\ell=-M}^{M} \left(1-\frac{|\ell|}{M}\right) \right)^2=M^2
\qquad\text{and}\qquad
\|\h_M\|_2^2 =\sum_{\ell=-M}^{M} \left(1- \frac{|\ell|}{M}\right)^2 = \frac{2 M}{3}+\frac{1}{3 M}.\label{eqn:sec:B:Pi:1}
\end{align}
Then plug~\eqref{eqn:sec:B:Pi:1} into~\eqref{eqn:sec:B:Pi},  we obtain an upper bound on $\var(S_m)$ as
\begin{equation}\label{eqn:sec:B:QW1}
\begin{aligned}
\var(S_m)&\leq \frac{\sigma^2}{M^4}M^2\left(\frac{2 M}{3}+\frac{1}{3 M}\right)
=\sigma^2\left( \frac{2}{3M}+\frac{1}{M^3} \right)
\stack{\ding{172}}{=}\sigma^2\left( \frac{4}{3n}+\frac{8}{n^3} \right) \stack{\ding{173}}{\leq} 1.334\sigma^2 /n, \quad\text{for }n\geq130, 
\end{aligned}
\end{equation}
where \ding{172} follows from $n=2M$ and \ding{173} follows since $ {8}/{n^2}$ is a decreasing sequence of $n$ implying the maximal happens at $n=130$. Thus we can choose $\overline\Pi=1.334\sigma^2/n.$ Plugging $\overline\Pi=1.334\sigma^2/n$ into the Gaussian tail bound~\eqref{eqn:sec:B:tail},  we get
\begin{align}\label{eqn:sec:B:tail:2}
 \P\left[|S_m|\geq t\sqrt{1.334}\sigma/\sqrt{n}\right]\leq e^{-t^2/2}
\end{align}
for all $m=0, \ldots, N-1.$

Applying the union bound yields
\begin{align}\label{eqn:sec:B:QW2}
\P\left[\sup_{f\in\TT} \left|\a(f)^H\mZ\w\right|\geq 2t\sqrt{1.334}\sigma/\sqrt n\right] \leq \P\left[\max_{m=0\ldots N-1}|S_m|\geq t\sqrt{1.334}\sigma/\sqrt n\right]\leq Ne^{-t^2/2}, 
\end{align}
where the first inequality follows from~\eqref{eqn:sec:B:1}. Setting $t=\sqrt{8\log n}$ in the above gives
\begin{align}
\P\left[\sup_{f\in\TT} \left|\a(f)^H\mZ\w\right|\geq \underbrace{4\sqrt{2}\sqrt{1.334}}_{\leq 6.534}\sqrt{\log n/n}\sigma\right]\leq\frac{8\pi(2n+1)}{n^4}
{\leq} \frac{1}{n^2}, 
\end{align}
where the last inequality holds for $n\geq130.$ Therefore,  we obtain that
\begin{align}\label{eqn:sec:B:2}
\P\left[\sup_{f\in\TT} \left|\a(f)^H\mZ\w\right|\leq 6.534\sqrt{\log n/n}\sigma\right]\geq 1- \frac{1}{n^2}, \quad\text{for }n\geq130.
\end{align}

To bound $\sup_{f\in\TT} \left|\a'(f)^H\mZ\w\right|$ and $\sup_{f\in\TT} \left|\a''(f)^H\mZ\w\right|$,  a natural approach is to exploit the relations between $\a(\f)$ and its derivatives  $\a'(\f)$,  $\a{''}(\f)$:
\begin{align*}
\a'(f)&=(i2\pi\diag(\n)))\a(f)
\text{ and }
\a''(f)=(i2\pi\diag(\n))^2\a(f).
\end{align*}
Similarly,  define $S_m'$ and $S_m''$ as the $m$th equispaced sample of $\a'(f)^H\mZ\w$ and $\a''(f)^H\mZ\w$,  respectively:
\begin{align*}
S_m'&=\a'(m/N)^H\mZ\w=\a(m/N)^H(-i2\pi\diag(\n))\mZ\w, 
\\
S_m''&=\a''(m/N)^H\mZ\w=\a(m/N)^H(-i2\pi\diag(\n))^2\mZ\w.
\end{align*}
Hence $S_m'\sim\N(0, \var(S_m'))$
and $S_m''\sim\N(0, \var(S_m''))$ with
\begin{align*}
\var(S_m')&=\sum_{\ell=-n}^n \left( {2\pi\ell g_M(\ell)}/{M}\right)^2\sigma^2)
\leq (2\pi n)^2\left(\sum_{\ell=-n}^n \left( {g_M(\ell)}/{M}\right)^2\sigma^2\right)
%\\
%&
\stack{\ding{172}}{\leq}(2\pi n)^2 1.334\sigma^2/n, 
\\
\var(S_m'')&=\sum_{\ell=-n}^n \left( {(2\pi\ell)^2 g_M(\ell)}/{M}\right)^2\sigma^2)\leq (2\pi n)^4\left(\sum_{\ell=-n}^n \left( {g_M(\ell)}/{M}\right)^2\sigma^2\right)
%\\
%&
\stack{\ding{173}}{\leq}(2\pi n)^4 1.334\sigma^2/n, 
\end{align*}
where \ding{172} and \ding{173} follow from~\eqref{eqn:sec:B:QW1}.
Applying the Gaussian deviation inequality to $S_m', S_m''$ yields
\begin{align*}
\P\left[|S_m'|\geq t2\pi\sqrt{1.334}\sqrt{n}\sigma\right]&\leq 2e^{-t^2/2}
\qquad\text{and}\qquad
\P\left[|S_m''|\geq t4\pi^2\sqrt{1.334}n\sqrt{n}\sigma\right]\leq 2e^{-t^2/2}.
\end{align*}
Then applying the same arguments as~\eqref{eqn:sec:B:QW2},  we get for $n\geq130$, 
\begin{equation}\label{eqn:sec:B:QW3}
\begin{aligned}
&\P\left[\sup_{f\in\TT} \left|\a'(f)^H\mZ\w\right|\leq\underbrace{ 8\sqrt{2}\pi\sqrt{1.334}}_{\leq41.052}\sqrt{n\log n}\sigma\right]\geq 1- \frac{1}{n^2}, 
\\
&\P\left[\sup_{f\in\TT} \left|\a''(f)^H\mZ\w\right|\leq\underbrace{ 16\sqrt{2}\pi^2\sqrt{1.334}}_{\leq257.94}n\sqrt{n\log n}\sigma\right]\geq 1- \frac{1}{n^2}.
\end{aligned}
\end{equation}

Finally,   we invoke that
\begin{align*}
\mA(\f)=[\a(f_1), \ldots, \a(f_k)], 
\qquad
\mA'(\f)=[\a'(f_1), \ldots, \a'(f_k)], 
\qquad
\mA''(\f)=[\a''(f_1), \ldots, \a''(f_k)], 
\end{align*}
and recognize that $\sup_{f\in\TT} \left|\a^{(\ell)}(f)^H\mZ\w\right|$  is an upper bound on $\|\mA^{(\ell)}(\f)^H\mZ\w\|_\infty$ to get
\begin{align*}
\|\mA^{(\ell)}(\f)^H\mZ\w\|_\infty
&{=}\max_{f\in\{f_j\}}|\a^{(\ell)}(f)^H\mZ\w|
%\\
%&
{\leq} \sup_{f\in\TT} |\a^{(\ell)}(f)^H\mZ\w|.
\end{align*}
Together with~\eqref{eqn:sec:B:2}, ~\eqref{eqn:sec:B:QW3} and the definition $\gamma_0=\sigma\sqrt{\frac{\log n}{n}}$,  we obtain that the following inequalities hold for $n\geq130$ with probability at least  $1- \frac{1}{n^2}$:
\begin{equation}
\label{eqn:noise:bound1}
\begin{aligned}
\|\mA(\f)^H\mZ\w\|_\infty &\leq \sup_{f\in\TT} \left|\a(f)^H\mZ\w\right|\leq 6.534\gamma_0, 
\\
\|\mA'(\f)^H\mZ\w\|_\infty &\leq \sup_{f\in\TT} \left|\a'(f)^H\mZ\w\right|\leq 41.052 n\gamma_0, 
\\
\|\mA''(\f)^H\mZ\w\|_\infty &\leq \sup_{f\in\TT} \left|\a''(f)^H\mZ\w\right|\leq 257.94 n^2\gamma_0.
\end{aligned}
\end{equation}

As a consequence,  we claim that the following inequalities hold  for $n\geq130$  with probability at least  $1- \frac{1}{n^2}$:
\begin{align}
\|\diag(1./|\c^\star|)\mA'(\f)^H\mZ\w\|_\infty/\sqrt\tau
&
\stack{\ding{172}}{\leq} \|\diag(1./|\c^\star|)\|_{\infty, \infty}\|\mA(\f)^H\mZ\w\|_\infty/\sqrt\tau
\nn\\
&\leq \frac{1}{\sqrt{3.289n^2}}\frac{1}{{c_{\min}^\star}}(41.052n\gamma_0)
\leq22.64\gamma, 
\label{eqn:noise:bound2.1}
\\[1ex]
\|\diag(\c./|\c^\star|^2)\mA''(\f)^H\mZ\w\|_\infty/\tau
&\stack{\ding{173}}{\leq} \|\diag(\c./|\c^\star|)\|_{\infty, \infty}\|\diag(1./|\c^\star|)\|_{\infty, \infty}\|\mA''(\f)^H\mZ\w\|_\infty/\tau
\nn\\
&\leq \frac{1}{3.289n^2}(1+X^\star\gamma)\frac{1}{{c_{\min}^\star}}(257.94n^2\gamma_0)
\leq 78.43(1+X^\star\gamma)\gamma, 
\label{eqn:noise:bound2.2}
\end{align}
where \ding{172} follows from that $\|\mA\x\|_\infty\leq\|\mA\|_{\infty, \infty}\|\x\|_\infty$ by the definition of the $\ell_{\infty, \infty}$ norm  and the fact $\tau\geq 3.289n^2$ for $n\geq130$ by~\eqref{eqn:sec:A:tau}. \ding{173} follows from the sub-multiplicative property of the $\ell_{\infty, \infty}$ norm that $\|\mA\mB\x\|_\infty\leq\|\mA\|_{\infty, \infty}\|\mB\|_{\infty, \infty}\|\x\|_\infty$ and $\|\diag(\c./|\c^\star|)\|_{\infty, \infty}=\max_\ell|c_\ell|/|c^\star_\ell|\leq(1+X^\star\gamma)$ which follows from the assumption $\|\btheta-\btheta^\star\|_{\hinfty}\leq {X^\star\gamma_0}/{\sqrt 2}$ and the derived results~\eqref{eqn:parameter:close}.

\section{Gradient and Hessian for the Nonconvex Program (\ref{eqn:pdw})}\label{sec:C}
Recall that the objective function ${G}$ of the program~\eqref{eqn:pdw} is
\begin{align*}
{G}(\f, \c)=\frac{1}{2}\|\mA(\f)\c-\y\|_{Z}^2+\lambda\|\c\|_1.
\end{align*}
We denote $\c=\u+i\v$ for $\u\in\R^k$ and $\v\in\R^k.$

\subsection{Gradient}
Let the operators $\R\{\cdot\}$ and $\I\{\cdot\}$ take respectively the real and imaginary parts of a complex number or vector.
The gradient of ${G}(\f, \c)$ with respect to $\btheta:=(\f, \u, \v)\in\R^{3k}$ is defined by
\begin{align}\label{eqn:gradient}
\nabla {G}(\btheta)
=
\begin{bmatrix}
 {\partial {G}}/{\partial \f}
\\[0.8ex]
 {\partial {G}}/{\partial \u}
\\[0.8ex]
 {\partial {G}}/{\partial \v}
\end{bmatrix}
\stack{\ding{172}}{=}
\begin{bmatrix*}[l]
 {\partial {G}}/{\partial \f}
\\[1ex]
2\R\{  {\partial {G}}/{\partial \bar\c}\}
\\[1ex]
2\I\{ {\partial {G}}/{\partial \bar\c}\}
\end{bmatrix*}
&\stack{\ding{173}}{=}
\begin{bmatrix*}[l]
\R\left\{{(\mA'(\f)\diag(\c))^H\mZ(\mA(\f)\c-\y)}\right\}
\\[1ex]
\R\left\{{\mA(\f)^H\mZ(\mA(\f)\c-\y)}+\lambda \c./ |\c |\right\}
\\[1ex]
\I\left\{{\mA(\f)^H\mZ(\mA(\f)\c-\y)}+\lambda \c./ |\c |\right\}
\end{bmatrix*}
\nonumber\\[1ex]
&\stack{\ding{174}}{=}
\begin{bmatrix*}[l]
\R\left\{{\diag( {\c})^H(-\D_1(\f)\c+\D_1(\f, \f^\star)\c^\star-\mA'(\f)^H\mZ\w})\right\}
\\[1ex]
\R\left\{{\D_0(\f)\c-\D_0(\f, \f^\star)\c^\star-\mA(\f)^H\mZ\w}+\lambda \c./ |\c |\right\}
\\[1ex]
\I\left\{{\D_0(\f)\c-\D_0(\f, \f^\star)\c^\star-\mA(\f)^H\mZ\w}+\lambda \c./ |\c |\right\}
\end{bmatrix*}, 
\end{align}
where
\ding{172} holds for ${G}\in\R.$
\ding{173} follows from $\diag(\text{d}\f)\c=\diag(\c)\text{d}{\f}$ and $\text{d}|c|=\frac{\bar{c}\text{d}c+c\text{d}\bar{c}}{2|c|}$.
\ding{174} follows from the kernel matrix factorization formulas~\eqref{eqn:sec:A:kernelmatrix:general}-\eqref{eqn:sec:A:kernelmatrix:cross} and by taking into account that  $\y=\x^\star+\w=\mA(\f^\star)\c^\star+\w.$

\subsection{Hessian}

The symmetric Hessian matrix $\nabla^2{G}(\btheta)$ is given by
\begin{align*}
\nabla^2{G}(\btheta)&=
\begin{bmatrix}
\frac{\partial^2 {G}}{\partial \f\partial \f} & \frac{\partial^2 {G}}{\partial \f\partial \u} & \frac{\partial^2 {G}}{\partial \f\partial \v}\\
\frac{\partial^2 {G}}{\partial \u\partial \f} & \frac{\partial^2 {G}}{\partial \u\partial \u} & \frac{\partial^2 {G}}{\partial \u\partial \v}\\
\frac{\partial^2 {G}}{\partial \v\partial \f} & \frac{\partial^2 {G}}{\partial \v\partial \u} & \frac{\partial^2 {G}}{\partial \v\partial \v}\\
\end{bmatrix}
:=
\begin{bmatrix}
\mH_{\f\f} & \mH_{\f\u}& \mH_{\f\v}\\
\mH_{\u\f} & \mH_{\u\u}& \mH_{\u\v}\\
\mH_{\v\f} & \mH_{\v\u}& \mH_{\v\v}\\
\end{bmatrix}
\end{align*}
with
\begin{align}\label{eqn:hessian}
\begin{bmatrix}
\mH_{\f\f}
\\[0.8ex]
\mH_{\f\u}
\\[0.8ex]
\mH_{\f\v}
\\[0.8ex]
\mH_{\u\u}
\\[0.8ex]
\mH_{\v\v}
\\[0.8ex]
\mH_{\u\v}
\end{bmatrix}
&\stack{\ding{172}}{=}
\begin{bmatrix*}[l]
\R\{(\mA'(\f)\mathbf{\Lambda})^H\mZ\mA'(\f)\mathbf{\Lambda}+\diag((\mA''(\f)\mathbf{\Lambda})^H\mZ(\mA(\f)\c-\y))\}
\\[1ex]
\R\{(\mA'(\f)\mathbf{\Lambda})^H\mZ\mA(\f)+\diag(\mA'(\f)^H\mZ(\mA(\f)\c-\y))\}
\\[1ex]
\I\{-(\mA'(\f)\mathbf{\Lambda})^H\mZ\mA(\f) +\diag(\mA'(\f)^H\mZ(\mA(\f)\c-\y))\}
\\[1ex]
\mA(\f)^H\mZ\mA(\f)+\lambda\diag(\v^2\odot|\c|.^{-3})
\\[1ex]
\mA(\f)^H\mZ\mA(\f)+\lambda\diag(\u^2\odot|\c|.^{-3})
\\[1ex]
-\lambda\diag(\u\odot\v\odot|\c|.^{-3})
\end{bmatrix*}
\nonumber\\[1ex]
&\stack{\ding{173}}{=}
\begin{bmatrix*}[l]
\R\{-\mathbf{\Lambda}^H\D_2(\f)\mathbf{\Lambda}-\diag(\mathbf{\Lambda}^H\mA''(\f)^H\mZ\w)
-\diag(\mathbf{\Lambda}^H(\D_2(\f, \f^\star)\c^\star-\D_2(\f)\c))\}
\\[1ex]
\R\{-\mathbf{\Lambda}^H\D_1(\f)-\diag(\mA'(\f)^H\mZ\w)
+\diag(\D_1(\f, \f^\star)\c^\star)-\diag(\D_1(\f)\c)\}
\\[1ex]
\I\{\mathbf{\Lambda}^H\D_1(\f)-\diag(\mA'(\f)^H\mZ\w)
+\diag(\D_1(\f, \f^\star)\c^\star)-\diag(\D_1(\f)\c)\}
\\[1ex]
\D_0(\f)+\lambda\diag(\v^2\odot|\c|.^{-3})
\\[1ex]
\D_0(\f)+\lambda\diag(\u^2\odot|\c|.^{-3})
\\[0.8ex]
-\lambda\diag(\u\odot\v\odot|\c|.^{-3})
\end{bmatrix*}, 
\end{align}
where we denoted $\mathbf{\Lambda}:=\diag(\c)$ to simplify notation. \ding{172} follows from direct computation and \ding{173} follows from the matrix decomposition formulas~\eqref{eqn:sec:A:kernelmatrix:general}-\eqref{eqn:sec:A:kernelmatrix:cross} and by taking into account that  $\y=\x^\star+\w=\mA(\f^\star)\c^\star+\w.$

Remarkably,  if we replace the noisy signal $\y$ in the objective function of the nonconvex program~\eqref{eqn:pdw} with the noise-free signal $\x^\star$ to get
\begin{align*}
{G^\lambda}(\f, \c)=\frac{1}{2}\|\mA(\f)\c-\x^\star\|_{Z}^2+\lambda\|\c\|_1, 
\end{align*}
then its gradient and Hessian matrix can be obtained from those of ${G}(\f, \c)$ by setting the noise $\w$ to zero.

\section{Proof of Lemma~\ref{lem:fix1}}\label{sec:D}
\begin{proof} 
The underlying fixed point map is
\begin{align*}
\Theta^\lambda(\btheta) = \btheta-\W^\star\nabla G^\lambda(\btheta), 
\end{align*}
where $G^\lambda$ is defined as the objective function of the nonconvex  program~\eqref{eqn:pdw} with the noisy signal $\y$ replaced by the noise-free signal $\x^\star$:
$${G}^\lambda(\btheta) =\frac{1}{2}\|\mA(\f)\c-\x^\star\|_{\mZ}^2+\lambda\|\c\|_1.$$
By Theorem~\ref{thm:fix},  to show the existence and uniqueness of a point $\btheta^\lambda\in\NN^\star$ such that $\Theta^\lambda(\btheta^\lambda)=\btheta^\lambda$, 
the key is to show that  $\Theta^\lambda$ satisfies  the non-escaping condition and the contraction condition:
\\[1ex]
\boxed{\vbox{
\vspace*{-0.2cm}
\begin{description}
\item[(i)] $\Theta^\lambda(\NN^\star)\subset\NN^\star$;
\item[(ii)]There exists $\rho\in(0, 1)$ such that $\|\Theta^\lambda(\v)-\Theta^\lambda(\w)\|_\hinfty\leq\rho\|\v-\w\|_\hinfty$ for any $\v, \w\in\NN^\star$.
\end{description}
\vspace*{-0.2cm}
}
}

\subsection{Showing the Contraction Property}
For $\v, \w\in\NN^\star$,  we have
\begin{align*}
\| \Theta^\lambda(\v)-\Theta^\lambda(\w)\|_\hinfty
&\stack{\ding{172}}{=}  \left\|\int_{0}^1[\nabla\Theta^\lambda(t\v+(1-t)\w)](\v-\w)\mathrm{d}t \right\|_\hinfty
\stack{\ding{173}}{\leq}  \maximize_{\btheta\in \NN^\star} \|\nabla\Theta^\lambda(\btheta)\|_{\hinfty, \hinfty}   \|\v-\w\|_\hinfty, 
\end{align*}
where \ding{172} follows from the integral form of the mean value theorem for vector-valued functions (see~\cite[Eq. (A.57)]{nocedal2006numerical});
\ding{173} follows from the sub-multiplicative property of $\|\cdot\|_{\hinfty, \hinfty}$ and the fact that $t\v+(1-t)\w\in\NN^\star$ for $t\in[0, 1]$.
Thus,  it suffices to show
\begin{align*}
\maximize_{\btheta\in \NN^\star} \|\nabla\Theta^\lambda(\btheta)\|_{\hinfty, \hinfty} <1, 
\end{align*}
where the matrix $\ell_{\hinfty, \hinfty}$ norm is defined  by (following from the definition of the $\ell_\hinfty$ norm)
\begin{align*}
\|\mathbf{A}\|_{\hinfty, \hinfty}&=
\left\|\begin{bmatrix}\mA_{11}&\mA_{12} &\mA_{13} \\ \mA_{21}& \mA_{22}&\mA_{23}\\\mA_{31}&\mA_{32}&\mA_{33}\end{bmatrix}\right\|_{\hat{\infty}, \hat{\infty}}
:=
\left\|\begin{bmatrix}\mS\mA_{11}\mS^{-1}&\mS\mA_{12} &\mS\mA_{13} \\ \mA_{21}\mS^{-1}& \mA_{22}&\mA_{23}\\ \mA_{31}\mS^{-1} &\mA_{32}&\mA_{33}\end{bmatrix}\right\|_{\infty, \infty}, 
\end{align*}
with $\mS=\sqrt{\tau}\diag(|\c^\star|).$ Together with
\begin{align*}
\W^\star= 
\begin{bmatrix}
\mS^{-2} &&\\
&\eye_k&\\
&&\eye_k
\end{bmatrix}, 
\end{align*}
we therefore obtain that
\begin{equation}
\begin{aligned}
\|\W^\star\mathbf{A}\|_{\hinfty, \hinfty}
&=
\left\|
\begin{bmatrix}
\mS^{-1} \mA_{11}\mS^{-1}&\mS^{-1} \mA_{12} &\mS^{-1} \mA_{13} \\ \mA_{21}\mS^{-1}& \mA_{22}&\mA_{23}\\ \mA_{31}\mS^{-1} &\mA_{32}&\mA_{33}
\end{bmatrix}
\right\|_{\infty, \infty}=\|{\W^\star}^{\frac{1}{2}}\mathbf{A}{\W^\star}^{\frac{1}{2}}\|_{\infty, \infty}
:=  \|\Upsilon(\mathbf{A}) \|_{\infty, \infty}, \label{eqn:gamma}
\end{aligned}
\end{equation}
where the linear operator $\Upsilon(\cdot):={\W}^{\star\frac{1}{2}}(\cdot){\W}^{\star\frac{1}{2}}$.
The Jacobian of the fixed point map $\Theta^\lambda$ is given by
\begin{align}\label{eqn:grad:theta}
\nabla \Theta^\lambda(\btheta)=\eye-\W^\star\nabla^2G^\lambda(\btheta), 
\end{align}
where the symmetric Hessian matrix $\nabla^2 G^\lambda(\btheta)$ can be obtained from $\nabla^2 {G}(\btheta)$ by setting the noise $\w$ to zero:
\begin{align*}
\nabla^2 G^\lambda(\btheta)=\begin{bmatrix} \mH_{\f\f} & \mH_{\f\u} & \mH_{\f\v}\\ \mH_{\u\f}&\mH_{\u\u}&\mH_{\u\v}\\\mH_{\v\f}&\mH_{\v\u}&\mH_{\v\v}  \end{bmatrix}.
\end{align*}
Due to the symmetric structure of the Hessian matrix,  it suffices to know the expressions for the following block matrices (see Eq.~\eqref{eqn:hessian}):
\begin{center}
\begin{tabular}{ll}
$\mH_{\f\f}= \R\{-\mathbf{\Lambda}^H\D_2(\f)\mathbf{\Lambda}-\diag(\mathbf{\Lambda}^H(\D_2(\f, \f^\star)\c^\star-\D_2(\f)\c))\}; $	&
$\mH_{\u\u}= \D_0(\f)+\lambda\diag(\v\odot\v./\big|\c\big|^{3}); $
\\
$\mH_{\f\u}= \R\{-\mathbf{\Lambda}^H\D_1(\f)+\diag(\D_1(\f, \f^\star)\c^\star)-\diag(\D_1(\f)\c)\}$; &
$\mH_{\v\v}= \D_0(\f)+\lambda\diag(\u\odot\u./\big|\c\big|^{3})$; 
\\
$\mH_{\f\v}= \I\{\mathbf{\Lambda}^H\D_1(\f)
+\diag(\D_1(\f, \f^\star)\c^\star)-\diag(\D_1(\f)\c)\}$; 
&
$\mH_{\u\v}= -\lambda\diag(\u\odot\v./\big|\c\big|^{3})$,
\end{tabular}
\end{center}
where $\mathbf{\Lambda}=\diag(\c)$.

Next we compute the weighed $\ell_{\hinfty, \hinfty}$ norm of the Jacobian of the fixed point map $\Theta^\lambda$:
\begin{align*}
\|\nabla\Theta^\lambda(\btheta)\|_{\hinfty, \hinfty}
&\stack{\ding{172}}{=}\| \W^\star\nabla^2G^\lambda(\btheta)-\eye\|_{\hinfty, \hinfty}
\stack{\ding{173}}{=}  \|\Upsilon(\nabla^2G^\lambda(\btheta)-{\W^\star}^{-1})  \|_{\infty, \infty}
\stack{\ding{174}}{=}   \|\Upsilon(\nabla^2G^\lambda(\btheta))-\eye\|_{\infty, \infty}, 
\end{align*}
where \ding{172} follows from~\eqref{eqn:grad:theta},  \ding{173} follows from~\eqref{eqn:gamma} by noting that $\W^\star\nabla^2G^\lambda(\btheta)-\eye=\W^\star(\nabla^2G^\lambda(\btheta)-{\W^\star}^{-1})$ and \ding{174} from  the linearity of $\Upsilon(\cdot)$ and $\Upsilon({\W^\star}^{-1})={\W^\star}^{\frac{1}{2}}{\W^\star}^{-1}{\W^\star}^{\frac{1}{2}}=\eye.$ Direct computation gives
\begin{align*}
  \Upsilon(\nabla^2G^\lambda(\btheta))-\eye
=& \begin{bmatrix}
\frac{-1}{{\tau}}\R\{\mathbf{\Phi}^H\D_2(\f)\mathbf{\Phi}\}-\eye\!\!\!&
\frac{-1}{\sqrt{\tau}}\R\{\mathbf{\Phi}\}\D_1(\f)\!\!\!&\frac{-1}{\sqrt{\tau}}\I\{\mathbf{\Phi}\}\D_1(\f)\\
\frac{1}{\sqrt{\tau}}\D_1(\f)\R\{\mathbf{\Phi}\}\!\!\!&\D_0(\f)-\eye\!\!\!& \\
\frac{1}{\sqrt{\tau}}\D_1(\f)\I\{\mathbf{\Phi}\}\!\!\!& \!\!\!&\D_0(\f)-\eye
\end{bmatrix}
+
\begin{bmatrix}
\diag(\d_{\f\f})\!\!\!&\diag(\d_{\f\u})\!\!\!&\diag(\d_{\f\v})\\
\diag(\d_{\f\u})\!\!\!&\diag(\d_{\u\u})\!\!\!&\diag(\d_{\u\v})\\
\diag(\d_{\f\v})\!\!\!&\diag(\d_{\u\v})\!\!\!&\diag(\d_{\v\v})
\end{bmatrix} 
\end{align*}
where $\mathbf{\Phi}:=\diag(\c./|\c^\star|)$ and
\begin{center}
	\begin{tabular}{ll}
$\d_{\f\f}= -\R\{\diag(\c./|\c^\star|^2)^H[ \D_2(\f, \f^\star)\c^\star-\D_2(\f)\c]/\tau\};$
	&  
$\d_{\u\u}= \lambda\diag(\u\odot\u./\big|\c\big|^{3});$
	\\ 
$\d_{\f\u}= \R\{\diag(1./|\c^\star|)[ \D_1(\f, \f^\star)\c^\star-\D_1(\f)\c]/\sqrt\tau\};$		
	&  
$\d_{\u\v}= \lambda\diag(\u\odot\v./\big|\c\big|^{3});	$
	\\ 
$\d_{\f\v}= \I\{\diag(1./|\c^\star|)[\D_1(\f, \f^\star)\c^\star-\D_1(\f)\c]/\sqrt\tau\};$	
	&  
$\d_{\v\v}= \lambda\diag(\v\odot\v./\big|\c\big|^{3}).$	
\end{tabular} 
\end{center}
Clearly, 
\begin{align*}
\|\Upsilon(\nabla^2G^\lambda(\btheta))-\eye\|_{\infty, \infty}=\max\left\{\Pi^\lambda_1,  \Pi^\lambda_2,  \Pi^\lambda_3\right\}
\end{align*}
with $\Pi^\lambda_1, \Pi^\lambda_2, \Pi^\lambda_3$ being the first,  second and third absolute row sums of $ \Upsilon(\nabla^2G^\lambda(\btheta))-\eye$,  respectively.

\bigskip
\noindent{\bf Bounding $\Pi^\lambda_1$.}
\begin{align}
\Pi^\lambda_1
\leq&  \left\|-\R\{\diag(\c./|\c^\star|)^H\D_2(\f)/\tau\diag(\c./|\c^\star|)\}-\eye\right\|_{\infty,\infty}
 +2  \left\|\diag(\c./|\c^\star|)\D_1(\f)/\sqrt\tau\right\|_{\infty,\infty}
\nn\\
& +2\left\| \diag(1./|\c^\star|)[\D_1(\f^\lambda, \f^\star)\c^\star-\D_1(\f)\c]/\sqrt\tau\right\|_{\infty}
+\left\| \diag(\c./|\c^\star|^2)[\D_2(\f^\lambda, \f^\star)\c^\star-\D_2(\f)\c]/\tau\right\|_{\infty}
\nn\\
\stack{\ding{172}}{\leq}& (0.05610n^2/\tau+2.12X^\star\gamma)+2(1+X^\star\gamma)(0.01236n/\sqrt\tau)
+2(0.75038B^\star X^\star\gamma)+1.14168B^\star X^\star\gamma
\nn\\
\leq& 0.08561\label{eqn:PI:1}, 
\end{align}
where \ding{172} follows from  Eq.~\eqref{eqn:kernel:matrix:small:norm}, ~\eqref{eqn:A:DC-DC1}, ~\eqref{eqn:A:DC-DC2} and the following bound 
\begin{equation}\label{eqn:A:CDC-I}
\begin{aligned}
\left\|-\R\{\diag(\c./|\c^\star|)^H\D_2(\f)/\tau\diag(\c./|\c^\star|)\}-\eye\right\|_{\infty,\infty}
\leq&  \max_i\bigg|\frac{|c_i|^2}{|c^\star_i|^2}-1\bigg|+(0.05610n^2/\tau) \max_{i, j}\frac{|c_i||c_j|}{|c^\star_i||c^\star_j|}
\\
\leq&   X^\star \gamma(2+X^\star \gamma)+(0.05610n^2/\tau) (1+X^\star \gamma)^2\\
\leq&  1.05610 (X^\star \gamma)^2+2.113 X^\star \gamma+ 0.05610n^2/\tau\\
\leq&   0.05610n^2/\tau+2.12X^\star\gamma.
\end{aligned}
\end{equation}

\bigskip
\noindent{\bf Bounding $\Pi^\lambda_2$ and $\Pi^\lambda_3$.}
\\
Note $\Pi^\lambda_2$ and $\Pi^\lambda_3$ are of the same form. Thus we can bound them together:
 \begin{align*}
\max\{\Pi^\lambda_2, \Pi^\lambda_3\}
\leq& \left\| \D_1(\f)\R\{\diag(\c./|\c^\star|)\}\right\|_{\infty,\infty}/\sqrt\tau
+\left\|{\diag(1./|\c^\star|)}[ \D_1(\f, \f^\star)\c^\star-\D_1(\f)\c]\right\|_{\infty,\infty}/\sqrt\tau
\nn\\
&+\|\D_0(\f)-\eye\|_{{\infty,\infty}}+ 2\|\lambda\diag(\u\odot\v ./\big|\c\big|^{3})\|_{\infty,\infty}
\\
\stack{\ding{172}}{\leq}& (1+X^\star\gamma)(0.01236n/\sqrt\tau)+(0.75038B^\star X^\star\gamma) +(0.00755)+2(0.646 X^\star\gamma)\\
<&\Pi_1^\lambda \text{ (since $B^\star X^\star\gamma\leq 10^{-3}$)}, 
\end{align*}
where \ding{172} follows from Eq.~\eqref{eqn:kernel:matrix:small:norm}, ~\eqref{eqn:A:DC-DC1}-\eqref{eqn:A:DC-DC2} and   $\lambda\leq 0.646X^\star\gamma_0.$
Therefore, 
\begin{align}\label{eqn:norm:Theta:lambda}
\maximize_{\btheta\in\NN^\star}\left\|\Upsilon(\nabla^2G^\lambda(\btheta))-\eye\right\|_{\infty, \infty} \leq 0.08561<1, 
\end{align}
implying the contraction property of $\Theta^\lambda(\btheta)$.

\subsection{ Showing the Non-escaping Property}
 
By the definition of the neighborhood $\NN^\star$,   it suffices to bound the distance between $\Theta^\lambda(\btheta)$
and $\btheta^\star$:
\begin{align*}
\|\Theta^\lambda(\btheta) - \btheta^\star\|_\hinfty
\stack{\ding{172}}{\leq}&  \|\Theta^\lambda(\btheta) - \Theta^\lambda(\btheta^\star)\|_\hinfty + \| \Theta^\lambda(\btheta^\star)-\btheta^\star\|_\hinfty
\\
\stack{\ding{173}}{=}&          \|\int_{0}^1 [\nabla_{\btheta}\Theta^\lambda((1-t)\btheta^\star+t\btheta)](\btheta-\btheta^\star)\dif t\|_\hinfty+ \| \Theta^\lambda(\btheta^\star)-\btheta^\star\|_\hinfty
\\
\stack{\ding{174}}{\leq}& \maximize_{{\z\in\NN^\star}}\|\nabla_{\btheta} \Theta^\lambda(\z)\|_{\hinfty, \hinfty}  \|\btheta-\btheta^\star\|_\hinfty
+\|  \W^\star  \nabla G^\lambda(\btheta^\star)\|_\hinfty
\\
\stack{\ding{175}}{\leq}&  (0.08561)(  X^\star\gamma_0/{\sqrt 2} )+ \lambda
\stack{\ding{176}}{\leq}   X^\star\gamma_0/{\sqrt 2}, 
\end{align*}
where \ding{172} follows from the triangle inequality,  \ding{173} follows from the integral form of the mean value theorem for vector-valued functions (see~\cite[Eq. (A.57)]{nocedal2006numerical}),  \ding{174} follows from sub-multiplicative property of $\|\cdot\|_{\hinfty, \hinfty}$ and the fact that $(1-t)\btheta^\star+t\btheta)\in\N^\star$ for $t\in[0, 1]$,  \ding{175} follows from
\begin{align*}
\|\W^\star \nabla G^\lambda(\btheta^\star)\|_\hinfty
&=\left\|
\begin{bmatrix}
0\\
\R\{\lambda \c^\star./\big|\c^\star\big|\}\\
\I\{\lambda \c^\star./\big|\c^\star\big|\}
\end{bmatrix}
\right\|_{\hinfty}
\leq\lambda, 
\end{align*}
and \ding{176} holds for $ \lambda\leq  0.646X^\star  \gamma_0$ since $(0.08561)(  X^\star\gamma_0/{\sqrt 2} )+ 0.646X^\star  \gamma_0\leq 0.9992X^\star\gamma_0/{\sqrt 2}$.

In sum,  $\Theta^\lambda$ satisfies both the contraction and the non-escaping properties in $\N^\star$. Therefore,   by the contraction mapping theorem, 
the map $\Theta^\lambda$ has a unique fixed point $\btheta^\lambda\in \NN^\star$ satisfying $\Theta^\lambda(\btheta^\lambda) = \btheta^\lambda$.

We continue to show that $\btheta^\lambda$ is a differentiable function of $\lambda$.
Define a function $F:\R^{3k}\times \R\mapsto \R^{3k}$ as $F(\btheta, \lambda)=\nabla G^\lambda(\btheta)$ and recognize $F(\btheta, \lambda)$ is continuously differentiable since it has a continuous Jacobian given by
$$\partial F(\btheta, \lambda)=\begin{bmatrix}\frac{\partial}{\partial\btheta}F(\btheta, \lambda) & \frac{\partial}{\partial\lambda} F(\btheta, \lambda)  \end{bmatrix}=\left[\nabla^2 G^\lambda(\btheta)~~\begin{bmatrix} \zero\\ \R\{\c./|\c|\}\\ \I\{\c./|\c|\}\end{bmatrix}\right], $$
with   $\frac{\partial}{\partial\btheta} F(\btheta, \lambda)$  nonsingular in $\N^\star$ by~\eqref{eqn:norm:Theta:lambda}. Then according to the implicit function theorem (see~\cite[Proposition A.25]{bertsekas1999nonlinear}),  there is a continuously differentiable function $\g(\cdot)$ such that $F(\g(\lambda),  \lambda) = \nabla G^\lambda(\g(\lambda)) = \zero$ and
\begin{align}\label{eqn:gdiff}
\frac{\mathrm{d}}{\mathrm{d} \lambda} \g(\lambda) = - (\frac{\partial}{\partial\btheta}F(\g(\lambda), \lambda))^{-1}\frac{\partial }{\partial \lambda}F(\g(\lambda),  \lambda) = - (\nabla^2 G^\lambda(\g(\lambda)))^{-1}\frac{\partial }{\partial \lambda}\nabla G^\lambda(\g(\lambda)).
\end{align}
Since $\nabla G^\lambda(\g(\lambda)) = \zero$ is equivalent to $\Theta^\lambda(\g(\lambda)) = \g(\lambda)$,  we conclude that $\btheta^\lambda = \g(\lambda)$ due to the uniqueness of the fixed point of $\Theta^\lambda$. Therefore,  $\btheta^\lambda$ is a differentiable function of $\lambda$ and
\begin{align}
\frac{\mathrm{d}}{\mathrm{d} \lambda} \btheta^\lambda = - (\nabla^2 G^\lambda(\btheta^\lambda))^{-1}\frac{\partial }{\partial \lambda}\nabla G^\lambda(\btheta^\lambda).
\end{align}

Finally,  let $\lim_{\lambda \rightarrow 0} \btheta^\lambda = \btheta^0$. Taking limit as $\lambda$ goes to $0$ in the equation $\nabla G^\lambda(\btheta^\lambda) = \zero$  yields $\nabla G^0(\btheta^0) = \zero$ due to the continuity of $\nabla G^\lambda(\btheta)$ in $\lambda$ and $\btheta$ and the continuity of $\btheta^\lambda$. Since $\nabla G^0(\btheta^\star) = \zero$ by direct computation and the solution is unique in $\N^\star$,  we conclude that $\lim_{\lambda \rightarrow 0} \btheta^\lambda = \btheta^0 = \btheta^\star$.
\end{proof}

\section{Proof of Lemma~\ref{lem:fix2}}\label{sec:E}

\begin{proof}
The main idea is again to apply the contraction mapping theorem~\ref{thm:fix} to the fixed point map:
\begin{align*}
\Theta(\btheta) = \btheta-\W^\star\nabla {G}(\btheta), 
\end{align*}
where $G$ is the objective function of~\eqref{eqn:pdw}:
$${G}(\btheta) =\frac{1}{2}\|\mA(\f)\c-\y\|_{\mZ}^2+\lambda\|\c\|_1$$
with $\lambda=0.646X^\star  \gamma_0$.
By Theorem~\ref{thm:fix},   showing the existence of a unique point $\hat{\btheta}\in\NN^\lambda$ such that $\Theta(\hat{\btheta})=\hat{\btheta}$ can be reduced to showing that $\Theta$ satisfies both  the non-escaping property and the contraction properties:
\\[1ex]
\boxed{\vbox{\vspace*{-0.2cm}
\begin{description}
\item[(i)] $\Theta(\NN^\lambda)\subset\NN^\lambda$;
\item[(ii)] There exists $\rho\in(0, 1)$ such that $\|\Theta(\v)-\Theta(\w)\|_\hinfty\leq\rho\|\v-\w\|_\hinfty$ for any $\v, \w\in\NN^\lambda$.
\end{description}
\vspace*{-0.2cm}
}}

\subsection{ Showing the Contraction Property}
 
Recall that $\NN^\star$  is a neighborhood centered at $\btheta^\star$ and $\NN^\lambda$ is a neighborhood centered at $\btheta^\lambda$ defined respectively via
\begin{align*}
\NN^\star&= \left\{\btheta:\|\btheta-\btheta^\star\|_{\hinfty}\leq \frac{X^\star}{\sqrt2}\gamma_0\right\}
\quad
\text{ and }
\quad
\NN^\lambda= \left\{\btheta:\|\btheta-\btheta^\lambda\|_{\hinfty}\leq \frac{35.2}{\sqrt2}\gamma_0\right\}.
\end{align*}
Keep in mind that $\btheta^\lambda$ is the unique point in $\NN^\star$ that satisfies $\nabla G^\lambda(\btheta^\lambda)=\zero.$  To show the contraction of $\Theta$ in
$\NN^\lambda$,  our strategy is to show $\Theta$ is contractive in a larger set $\hat{\N}$ that contains $\NN^\lambda$:
\begin{align*}
 \hat{\NN}=\left\{\btheta:\|\btheta-\btheta^\star\|_{\hinfty}\leq \frac{X^\star+35.2}{\sqrt2}\gamma_0:=\frac{\hat{X}}{\sqrt2}\gamma_0\right\}.
\end{align*}
Recognize that $\hat{\N}$ is  a neighborhood centered at $\btheta^\star$ but with a radius  $35.2\gamma_0/\sqrt{2}$ larger than that of $\N^\star$. Such a choice is made for
the purpose of showing the closeness between the final fixed point solution $\hat{\btheta}$ and $\btheta^\star$. We remark that the quantity $35.2\gamma_0/\sqrt{2}$ corresponds to the dual atomic norm of the weighted Gaussian noise.
Adding such a noise norm term to the radius of the original neighborhood $\N^\star$ ensures that the region $\hat{\N}$ is large enough for $\Theta(\btheta)$ to be non-escaping. This is reasonable because the second fixed point map~\eqref{eqn:Map:2} involves an additive Gaussian noise and we have shown that the first fixed point map~\eqref{eqn:Map:1} (the one constructed in the noise-free setting) satisfies the non-escaping property in $\N^\star$.

Next,  we apply arguments similar to those of showing the contraction of $\Theta^\lambda$ in $\NN^\star$. In particular,  we first compute the expression of
$\Upsilon(\nabla^2{G}(\btheta))-\eye$:
\begin{align*}
\Upsilon(\nabla^2{G}(\btheta))-\eye
=&\begin{bmatrix}
\frac{-1}{{\tau}}\R\{\mathbf{\Phi}^H\D_2(\f)\mathbf{\Phi}\}-\eye\!\!\!&
\frac{-1}{\sqrt{\tau}}\R\{\mathbf{\Phi}\}\D_1(\f)\!\!\!&\frac{-1}{\sqrt{\tau}}\I\{\mathbf{\Phi}\}\D_1(\f)\\
\frac{1}{\sqrt{\tau}}\D_1(\f)\R\{\mathbf{\Phi}\}\!\!\!&\D_0(\f)-\eye\!\!\!&\zero \\
\frac{1}{\sqrt{\tau}}\D_1(\f)\I\{\mathbf{\Phi}\}\!\!\!&\zero \!\!\!&\D_0(\f)-\eye
\end{bmatrix}
+
\begin{bmatrix}
\diag(\hat\d_{\f\f})\!\!\!&\diag(\hat\d_{\f\u})\!\!\!&\diag(\hat\d_{\f\v})\\
\diag(\hat\d_{\f\u})\!\!\!&\diag(\hat\d_{\u\u})\!\!\!&\diag(\hat\d_{\u\v})\\
\diag(\hat\d_{\f\v})\!\!\!&\diag(\hat\d_{\u\v})\!\!\!&\diag(\hat\d_{\v\v})
\end{bmatrix}  
\end{align*}
with $\mathbf{\Phi}=\diag(\c./|\c^\star|)$ and
\begin{center}
\begin{tabular}{ll}
$\hat\d_{\f\f}= - \R\{\diag(\c./|\c^\star|^2)^H[ \mA''(\f)^H\mZ\w + \D_2(\f, \f^\star)\c^\star-\D_2(\f)\c]/\tau\}$;
&  
$\hat\d_{\u\u}= \lambda\diag(\u\odot\u./\big|\c\big|^{3});$
\\ 
$\hat\d_{\f\u}=  \R\{\diag(1./|\c^\star|)[-\mA'(\f)^H\mZ\w+ \D_1(\f, \f^\star)\c^\star-\D_1(\f)\c]/\sqrt\tau\};$
&  
$\hat\d_{\u\v}= \lambda\diag(\u\odot\v./\big|\c\big|^{3});$
\\ 
$\hat\d_{\f\v}=  \I\{\diag(1./|\c^\star|)[-\mA'(\f)^H\mZ\w+\D_1(\f, \f^\star)\c^\star-\D_1(\f)\c]/\sqrt\tau\}$;
&  
$\hat\d_{\v\v}= \lambda\diag(\v\odot\v./\big|\c\big|^{3}).$
\end{tabular} 
\end{center}

Comparing the expressions for $[\Upsilon(\nabla^2 G^\lambda(\btheta))-\eye]$ and $[\Upsilon(\nabla^2{G}(\btheta))-\eye]$ shows that the latter differs in have additional noise terms in the first row and the first column blocks. We have  shown that the first absolute row sum $\Pi_1^\lambda$ of $[\Upsilon(\nabla^2 G^\lambda(\btheta))-\eye]$ dominates the other row sums. Having additional noise terms will only increase the final bounds due to the application of the triangle inequality. Therefore,  the first absolute row sum (denoted by $\hat\Pi_1$) of $[\Upsilon(\nabla^2{G}(\btheta))-\eye]$ also dominates and hence achieves the $\ell_{\infty, \infty}$ norm. Direct computation gives
\begin{align*}
\hat\Pi_1
&\leq\Pi_1^\lambda+2 \|\diag(1./|\c^\star|)\mA'(\f)^H\mZ\w\|_\infty/\sqrt\tau
+ \|\diag(\c./|\c^\star|^2)\mA''(\f)^H\mZ\w\|_\infty/\tau
\\
&\stack{\ding{172}}{\leq} 0.08561+  2(22.64\gamma)+78.43(1+\hat{X}\gamma) \gamma
\\
&\stack{\ding{173}}{\leq}   0.08563, 
\end{align*}
where \ding{172} follows from $\Pi^\lambda_1\leq0.08561$ and Eq.~\eqref{eqn:noise:bound1}-\eqref{eqn:noise:bound2.2},  \ding{173} follows from $\hat{X}=X^\star+35.2$ and
the SNR condition~\eqref{eqn:snr} that $X^\star {B^\star}\gamma  \leq 10^{-3} \text{\ and\ }  B^\star /{X^\star }\leq 10^{-4}$ hence $ 2(22.64\gamma)+78.43(1+\hat{X}\gamma) \gamma\leq0.00002.$
%
%\note{Qiuwei: can you double check the argument in the previous sentence to make sure it is precise?}
% I double checked it.
Hence, 
\begin{align}
\maximize_{\btheta\in \hat{\NN}} \|\nabla\Theta(\btheta)\|_{\hinfty, \hinfty} \leq 0.08563<1\label{eqn:D:Pi:lambda}.
\end{align}
This implies the contraction of  $\Theta$ in $\NN^\lambda$,  since
\begin{align*}
\maximize_{\btheta\in \NN^\lambda} \|\nabla\Theta(\btheta)\|_{\hinfty, \hinfty} \leq\maximize_{\btheta\in \hat{\NN}} \|\nabla\Theta(\btheta)\|_{\hinfty, \hinfty}.
\end{align*}

\subsection{ Showing the Non-escaping Property}
\begin{align*}
\|\Theta(\btheta) - \btheta^\lambda\|_\hinfty
=& \|(\Theta(\btheta) - \Theta(\btheta^\lambda))+(\Theta(\btheta^\lambda) - \btheta^\lambda)\|_\hinfty
\\
\stack{\ding{172}}{\leq}& \|\nabla \Theta(\tilde\btheta)^T(\btheta-\btheta^\lambda)\|_{\hinfty} + \|  \W^\star  \nabla {G}(\btheta^\lambda)\|_\hinfty
\\
\leq& \max_{\tilde\btheta\in\hat{\NN}}\|\nabla \Theta(\tilde\btheta)\|_{\hinfty, \hinfty}  \|\btheta-\btheta^\lambda\|_\hinfty+ \|  \W^\star  \nabla {G}(\btheta^\lambda)\|_\hinfty
\\
\stack{\ding{173}}{\leq}&  (0.08563) \left( {35.2\gamma_0}/{\sqrt 2}  \right)+ 22.7\gamma_0
\\
\leq&  {35.117\gamma_0}/{\sqrt 2}
<{35.2\gamma_0}/{\sqrt 2}, 
\end{align*}
where \ding{172} follows from the mean value theorem for some $\tilde{\btheta}$ on the line segment joining $\btheta$ and $\btheta^\lambda$
and \ding{173} follows from~\eqref{eqn:D:Pi:lambda} and~\eqref{eqn:E:final}.  Eq.~\eqref{eqn:E:final} is given as follows
\begin{equation}
\begin{aligned}
\left\|\W^\star\nabla {G}(\btheta^\lambda)\right\|_\hinfty
&=\left\|\W^\star
\begin{bmatrix*}[l]
\R\{-\diag({\c^\lambda})^H(\mA'(\f^\lambda)^H\mZ\w+\D_1(\f^\lambda, \f^\star)\c^\star-\D_1(\f^\lambda)\c^\lambda)\}\\
\R\{-\mA(\f^\lambda)^H\mZ\w-\D_0(\f^\lambda, \f^\star)\c^\star+\D_0(\f^\lambda)\c^\lambda+\lambda \c^\lambda./\big|\c^\lambda\big|\}\\
\I\{-\mA(\f^\lambda)^H\mZ(^\lambda\w-\D_0(\f^\lambda, \f^\star)\c^\star+\D_0(\f^\lambda)\c^\lambda+\lambda \c^\lambda./\big|\c^\lambda\big|\}
\end{bmatrix*}
\right\|_\hinfty
\\
&\stack{\ding{172}}{=}\left\|
\begin{bmatrix*}[l]
\R\{-\diag({\c^\lambda}./|\c^\star|)^H\mA'(\f^\lambda)^H\mZ\w\}/\sqrt\tau\\
\R\{-\mA(\f^\lambda)^H\mZ\w\}\\\I\{-\mA(\f^\lambda)^H\mZ\w\}
\end{bmatrix*}
\right\|_\infty
\\
&\stack{\ding{173}}{\leq}\left\|
\begin{bmatrix*}[l]
41.052n/\sqrt{\tau}(1+X^\star\gamma)\gamma_0
\\
6.534\gamma_0
\\
6.534\gamma_0
\end{bmatrix*}
\right\|_\infty
\\
&\leq22.7\gamma_0, 
\end{aligned} \label{eqn:E:final}
\end{equation}
where \ding{172} holds since $\nabla G^\lambda(\btheta)$ vanishes at $\btheta^\lambda$:
\begin{align*}
\nabla G^\lambda(\btheta^\lambda)
=\begin{bmatrix*}[l]
\R\{-\diag({\c^\lambda})^H(\D_1(\f^\lambda, \f^\star)\c^\star-\D_1(\f^\lambda)\c^\lambda)\}\\
\R\{-\D_0(\f^\lambda, \f^\star)\c^\star+\D_0(\f^\lambda)\c^\lambda+\lambda \c^\lambda./\big|\c^\lambda\big|\}\\
\I\{-\D_0(\f^\lambda, \f^\star)\c^\star+\D_0(\f^\lambda)\c^\lambda+\lambda \c^\lambda./\big|\c^\lambda\big|\}
\end{bmatrix*}=\zero.
\end{align*}
\ding{173} holds with probability at least $1- \frac{1}{n^2}$ by~\eqref{eqn:noise:bound1}-\eqref{eqn:noise:bound2.2}.

Hence both the contraction and the non-escaping properties are satisfied by $\Theta$ in $\N^\lambda$. Then by the contraction mapping theorem, 
we conclude the proof of Lemma~\ref{lem:fix2}.
\end{proof}

\section{Proof of Lemma~\ref{lem:q0:dual:certificate}}\label{sec:F}
\begin{proof}
To show that $\q^\star$ is a valid dual certificate,  it is instructive to first relate $\q^\star$ to the derivative of $\x^\lambda$ with respect to $\lambda$ (where we treat $\x^\lambda$ as a function of $\lambda$):
\begin{align}
\q^\star
=\lim_{\lambda\to0}\q^\lambda
=\lim_{\lambda\to0}\frac{\x^\star-\x^\lambda}{\lambda}
=-\frac{\mathrm{d} }{\mathrm{d} \lambda} \x^\lambda\big|_{\lambda=0},  \label{eqn:sec:4:1}
\end{align}
where we used the fact that $\lim_{\lambda \rightarrow 0} \x^\lambda = \lim_{\lambda \rightarrow 0} \mA(\f^\lambda)\c^\lambda = \mA(\f^\star)\c^\star = \x^\star$ by Lemma~\ref{lem:fix1}.
Since $\x^\lambda=\mA(\f^\lambda)\c^\lambda=\sum_\ell c_\ell^\lambda \a(f_\ell^\lambda)$,  we compute the derivative $\frac{\mathrm{d} }{\mathrm{d} \lambda} \x^\lambda$ using the chain rule as:
\begin{align}
\frac{\mathrm{d} }{\mathrm{d} \lambda} \x^\lambda
&= \sum_\ell \left(\frac{\mathrm{d} }{\mathrm{d} \lambda}u_\ell^\lambda
+ i \frac{\mathrm{d} }{\mathrm{d} \lambda}v_\ell^\lambda\right)\a(f_\ell^\lambda)
+ \sum_\ell c_\ell^\lambda \left(\frac{\mathrm{d} f_\ell^\lambda}{\mathrm{d} \lambda} \a'(f_\ell^\lambda)\right)
= \left[\mA'(\f^\lambda)\diag(\c^\lambda)~~\mA(\f^\lambda)~~i\mA(\f^\lambda)\right] \frac{\mathrm{d} }{\mathrm{d} \lambda}\btheta^\lambda\label{eqn:sec:4:2}, 
\end{align}
where $\mA'(\f) = \begin{bmatrix} \a'(f_1) & \cdots & \a'(f_k)\end{bmatrix}$.
Therefore,  using Eq.~\eqref{eqn:sec:4:1} and~\eqref{eqn:sec:4:2} we obtain:
\begin{align}
\q^\star
&=-\lim_{\lambda\to 0}\left[\mA'(\f^\lambda)\diag(\c^\lambda)~~\mA(\f^\lambda)~~i\mA(\f^\lambda)\right] \frac{\mathrm{d} }{\mathrm{d} \lambda}\btheta^\lambda
\nn\\
&= - \left[\mA'(\f^\star)\diag(\c^\star)~~\mA(\f^\star)~~i\mA(\f^\star)\right] \lim_{\lambda\to 0}\frac{\mathrm{d} }{\mathrm{d} \lambda}\btheta^\lambda\nonumber\\
& = \left[\mA'(\f^\star)\diag(\c^\star)~~\mA(\f^\star)~~i\mA(\f^\star)\right] (\nabla^2 G^0(\btheta^\star))^{-1}\frac{\partial }{\partial \lambda}\nabla G^0(\btheta^\star)
\label{eqn:sec:4:3}, 
\end{align}
where in the second line we again used the fact that $\lim_{\lambda\to0}\btheta^\lambda=\btheta^\star$ by Lemma~\ref{lem:fix1},  and in the last line we used the expression for ${\mathrm{d}\btheta^\lambda }/{\mathrm{d} \lambda}$ given in~\eqref{eqn:thetadiff}.

We next compute $\frac{\partial }{\partial \lambda}\nabla G^0(\btheta^\star)$ explicitly. Let $K^{(\ell)}(\cdot)$ denote the $\ell$-order derivative of the Jackson kernel $K(\cdot)$ (see Appendix~\ref{sec:A} for more details).
Recall that $\D_\ell(\f^1, \f^2) := [K^{(\ell)}(f_m^2-f_n^1)]_{1\leq n\leq k, 1\leq m\leq k}$ and $\D_\ell(\f):=\D_\ell(\f, \f)$ are matrices formed by sampling the Jackson kernel and its derivatives. Then we have the following expression for $ \nabla G^\lambda(\btheta)$ (see Appendix~\ref{sec:C} for more details)
\begin{align}\label{eqn:nabla:G:lambda}
&\nabla{G}^\lambda(\btheta)
=\begin{bmatrix*}[l]
\R\{ \diag(\c)(\D_1(\f, \f^\star)\c^\star-\D_1(\f)\c)\}
\\
\R\{ -\D_0(\f, \f^\star)\c^\star+\D_0(\f)\c+\lambda \c./\big|\c\big|\}
\\
\I\{ -\D_0(\f, \f^\star)\c^\star+\D_0(\f)\c+\lambda \c./\big|\c\big|\}
\end{bmatrix*}.
\end{align}
Therefore,  the partial derivative of~\eqref{eqn:nabla:G:lambda} with respect to $\lambda$ is the expanded complex sign vector:
\begin{align}
\frac{\partial }{\partial \lambda}\nabla G^\lambda(\btheta^\lambda)
=
\begin{bmatrix}
\zero\\ \R\{\sign(\c^\lambda)\}\\ \I\{\sign(\c^\lambda)\}
\end{bmatrix} := \begin{bmatrix}\zero \\ \s^\lambda_R \\ \s^\lambda_I\end{bmatrix}
\quad
\Longrightarrow
\quad
\frac{\partial }{\partial \lambda}\nabla G^0(\btheta^\star)
=
\begin{bmatrix}
\zero\\ \R\{\sign(\c^\star)\}\\ \I\{\sign(\c^\star)\}
\end{bmatrix} := \begin{bmatrix}\zero \\ \s^\star_R \\ \s^\star_I\end{bmatrix}.
\label{eqn:sec:4:5}
\end{align}
Here $\s^\lambda = \c^\lambda./|\c^\lambda|,  \s^\star = \c^\star./|\c^\star|$ and the subscript $R$ and $I$ indicate the real and imaginary parts respectively.

Combining Eq.~\eqref{eqn:sec:4:3} and~\eqref{eqn:sec:4:5},  we get
\begin{align}
\q^\star&=
[\mA'(\f^\star)~\mA(\f^\star)~i\mA(\f^\star)] \underbrace{\begin{bmatrix}\diag(\c^\star)& &\\ &\eye&\\ & &\eye  \end{bmatrix}(\nabla^2 G^0(\btheta^\star))^{-1}
\begin{bmatrix}\zero \\ \s^\star_R \\ \s^\star_I\end{bmatrix}}_{:=[\bbeta^T~\balpha_R^T~\balpha_I^T]^T}\label{eqn:sec:4:alpha:beta:gamma}, 
\end{align}
where  we have defined the coefficient vectors $\balpha_R,  \balpha_I$ and $\bbeta$ in~\eqref{eqn:sec:4:alpha:beta:gamma}. These coefficient vectors satisfy
\begin{align}\label{eqn:coef}
\nabla^2 G^0(\theta^\star) \begin{bmatrix}
\diag(\c^\star)^{-1}\bbeta\\
\balpha_R\\
\balpha_I
\end{bmatrix} = \begin{bmatrix}\zero \\ \s^\star_R \\ \s^\star_I\end{bmatrix}.
\end{align}
By denoting $\balpha=\balpha_R+i\balpha_I$ and $\balpha=[\alpha_1, \ldots, \alpha_k]^T$,  $\bbeta=[\beta_1, \ldots, \beta_k]^T$,  we obtain an explicit form for the dual polynomial $Q^\star(f)$:
\begin{align}
Q^\star(f)&=\a(f)^H\mZ\q^\star
%\nn\\
%&
=\sum_{\ell=1}^k \alpha_\ell K(f^\star_\ell-f)+\sum_{\ell=1}^k\beta_\ell K'(f^\star_\ell-f)\label{eqn:sec:4:Q:star}.
\end{align}

To show that $\q^\star$ certifies the atomic decomposition $\x^\star = \sum_{\ell = 1}^k c_\ell^\star \a(f_\ell^\star)$,  we need to establish that
\begin{enumerate}[label=\arabic*)]
	\item  $Q^\star(f)$ satisfies $Q^\star(f^\star_\ell) = \sign(c^\star_\ell),  \ell = 1,  \ldots,  k$ (Interpolation);
	\item  $|Q^\star(f)| < 1,  \forall f \notin T^\star$ (Boundedness). 
\end{enumerate}  

\subsection{Showing the Interpolation Property}
The Interpolation property follows from the construction process and is also easy to verify directly by noting
\begin{align*}
\nabla^2 G^0(\btheta^\star)=
\begin{bmatrix}
-\R\{\diag(\c^\star)^H\D_2(\f^\star)\diag(\c^\star)\}&\R\{-\diag(\c^\star)^H\D_1(\f^\star)\}&\I\{\diag(\c^\star)^H\D_1(\f^\star)\}\\
-\R\{{\D_1(\f^\star)}^H\diag(\c^\star)\}&\D_0(\f^\star)&0\\
-\I\{{\D_1(\f^\star)}^H\diag(\c^\star)\}&0&\D_0(\f^\star)
\end{bmatrix}.
\end{align*}
Indeed,  the Interpolation property is a result of~\eqref{eqn:coef}: since $\D_1(\f^\star)\in\R^{k\times k}$ and $\D_1(\f^\star)^T=-\D_1(\f^\star)$ (see Appendix~\ref{sec:A}),   the last two row blocks in~\eqref{eqn:coef} read
\begin{align}\label{eqn:int}
& \begin{bmatrix}
\D_1(\f^\star)\R\{\diag(\c^\star)\}&\D_0(\f^\star)& 0\\
\D_1(\f^\star)\I\{\diag(\c^\star)\}&0 &\D_0(\f^\star)
\end{bmatrix}
 \begin{bmatrix} \diag(\c^\star)^{-1} \bbeta \\ \balpha_R \\ \balpha_I \end{bmatrix}
 =\begin{bmatrix} \s^\star_R \\ \s^\star_I\end{bmatrix}
\nn\\
\iff
&\ \D_1(\f^\star) (\R\{\diag(\c^\star)\}+i\I\{\diag(\c^\star)\})\diag(\c^\star)^{-1} \bbeta+
\D_0(\f^\star)(\balpha_R+i\balpha_I)=\R\{\sign(\c^\star)\}+i\I\{\sign(\c^\star)\}\nn\\
\iff
&\ \D_1(\f^\star)  \bbeta+\D_0(\f^\star)\balpha=\sign(\c^\star)\nn\\
\iff
&\ Q^\star(f^\star_\ell) = \sign(c^\star_\ell),  \ell = 1,  \ldots,  k.
\end{align}
Furthermore,  the first row block of~\eqref{eqn:coef} is equivalent to
\begin{align}\label{eqn:gradientzero}
& -\R\{\diag(\c^\star)^H\D_2(\f^\star)\diag(\c^\star)\} \diag(\c^\star)^{-1} \bbeta +
\R\{-\diag(\c^\star)^H\D_1(\f^\star)\}\balpha_R +
\I\{\diag(\c^\star)^H\D_1(\f^\star)\}\balpha_I = \zero\nonumber\\
\iff &\ \R\{\diag(\c^\star)^H\left(\D_2(\f^\star) \bbeta + \D_1(\f^\star) \balpha\right)\}  = \zero\nonumber \\
 \iff &\ \R\{c^{\star H}_\ell Q^\star(f_\ell)'\} = 0,  \ell = 1,  \ldots,  k.
\end{align}

\subsection{Showing the Boundedness Property}
It remains to show that $Q^\star(f)$ satisfies the Boundedness property,  for which we follow the arguments of~\cite{Candes:2014br}.
%In particular,  fix an arbitrary point $f_0^\star\in T^\star$ as the reference point and let $f_{-1}^\star$ be the first frequency in $T^\star$ that lies on the left of $f_0^\star$ while $f_{1}^\star$ be the first frequency in $T^\star$ that lies on the right. Here ``left'' and ``right'' are directions on the complex circle $\TT$. We remark that the analysis depends only on the relative locations of $\{f_\ell^\star\}$. Hence,  to simplify the arguments,  we assume that the reference point $f_0^\star$ is at $0$ by shifting the frequencies if necessary.
%Then we divide the region between $f_0^\star = 0$ and $f_{1}^\star/2$ into three parts: Near Region $\N:=[0, 0.24/n]$,  Middle Region $\M:=[0.24/n, 0.75/n]$ and Far Region $\F:=[0.75/n,  f_{1}^\star/2]$. Also their symmetric counterparts: $-\N:=[-0.24/n,  0]$,  $-\M:=[-0.75/n, -0.24/n]$,  and $-\F:=[f_{-1}^\star/2,  -0.75/n]$.  We first show that the dual polynomial $|Q^\star(f)|\leq 1$ in $\N\cup\M\cup\F=[0,  f_{1}^\star/2]$. Then using the same symmetric arguments in~\cite{Candes:2014br},  we claim that $|Q^\star(f)|\leq 1$ also holds in
%$(-\N)\cup(-\M)\cup(-\F)=[f_{-1}^\star/2,  0]$. Combining these two results with the fact that the reference point $f_0^\star$ is chosen arbitrarily from $T^\star$ (and shifted to $0$),  we establish that the Boundedness property of $Q^\star(f)$ holds in the entire domain $\TT$.
We start with estimating the coefficient vectors $\balpha$ and $\bbeta$ by rewriting~\eqref{eqn:sec:4:alpha:beta:gamma} as
\begin{align}
\begin{bmatrix}
\diag(\c^\star)& &\\ &\eye&\\ & &\eye
\end{bmatrix}
\Phi
\left( \Phi
\nabla^2 G^0(\btheta^\star)
\Phi
\right)^{-1} \Phi
\begin{bmatrix}
\zero \\ \s_R^\star \\ \s_I^\star
\end{bmatrix}
=
\begin{bmatrix}
\bbeta \\ \balpha_R \\ \balpha_I
\end{bmatrix}\label{eqn:sec:F:QW}, 
\end{align}
where   $\Phi=\diag\left(\left[\diag\left( \frac{1}{|\c^\star|} \right), \eye, \eye  \right]\right).$
%We then get a normalized Hessian matrix $\Phi\nabla^2 G^0(\btheta^\star)\Phi$,  which is closer to the identity -- introducing $\Phi$  can cancel the effects of $\c^\star$ on the original Hessian $\nabla^2 G^0(\btheta^\star)$. Hence
%\begin{align*}
%\begin{bmatrix}
%\diag(\c^\star)& &\\ &\eye&\\ & &\eye
%\end{bmatrix}
%\begin{bmatrix}
%\diag\left( \frac{1}{|\c^\star|} \right)& &\\ &\eye&\\ & &\eye
%\end{bmatrix}
%=
%\begin{bmatrix}
%\diag(\s^\star)& &\\ &\eye&\\ & &\eye
%\end{bmatrix}, 
%\end{align*}
%where $\s^\star=\c^\star./|\c^\star|.$
Denoting $\mathbf{\Phi}:=\diag(\s^\star)$,  we further simplify~\eqref{eqn:sec:F:QW} as
\begin{align}\label{eqn:kernel:coefficient}
\begin{bmatrix}
-\R\{\mathbf{\Phi}^H\D_2(\f^\star)\mathbf{\Phi}\}&\R\{-\mathbf{\Phi}^H\D_1(\f^\star)\}&\I\{\mathbf{\Phi}^H\D_1(\f^\star)\}
\\
-\R\{{\D_1(\f^\star)}^H\mathbf{\Phi}\}&\D_0(\f^\star)&0
\\
-\I\{{\D_1(\f^\star)}^H\mathbf{\Phi}\}&0&\D_0(\f^\star)
\end{bmatrix}
\begin{bmatrix}
\mathbf{\Phi}^{-1} \bbeta \\ \balpha_R \\ \balpha_I
\end{bmatrix}
&=
\begin{bmatrix}0 \\ \s_R^\star \\ \s_I^\star
\end{bmatrix}.
\end{align}
Denote
\begin{align*}
\tilde{\D}_2&= -\R\{\diag(\s^\star)^H\D_2(\f^\star)\diag(\s^\star)\};
\\
\tilde{\D}_1&= \diag(\s^\star)^H\D_1(\f^\star);
\\
\tilde{\bbeta}&= \diag(\s^\star)^{-1} \bbeta.
\end{align*}
The last two row blocks of~\eqref{eqn:kernel:coefficient} give
\begin{align*}
\balpha_R&= {\D_0(\f^\star)}^{-1}[\s_R^\star+ \R\{{\D_1(\f^\star)}^H\diag(\s^\star)\}\tilde{\bbeta}];
\\
\balpha_I&= {\D_0(\f^\star)}^{-1}[\s_I^\star+\I\{{\D_1(\f^\star)}^H\diag(\s^\star)\}\tilde{\bbeta}]
\\
\end{align*}
implying
\begin{align}
\balpha&= {\D_0(\f^\star)}^{-1}[\s^\star +{\D_1(\f^\star)}^H\diag(\s^\star)\tilde{\bbeta}]\label{eqn:alpha}
\\
&= {\D_0(\f^\star)}^{-1}[\s^\star +{\D_1(\f^\star)}^H \bbeta ]
\nn\\
&=  \s^\star-(\eye-{\D_0(\f^\star)}^{-1})\s^\star+{\D_0(\f^\star)}^{-1}{\D_1(\f^\star)}^H\bbeta\nn.
\end{align}
Without loss of generality,  we assume  $\e_1^T\s^\star=1$. Then
\begin{align}\label{eqn:alpha1}
\alpha_1=1 -\left[(\eye-{\D_0(\f^\star)}^{-1})\s^\star-{\D_0(\f^\star)}^{-1}{\D_1(\f^\star)}^H\bbeta \right]_1, 
\end{align}
where $[\cdot]_1$ stands for the first entry of a vector.
The first row block of~\eqref{eqn:kernel:coefficient} leads to
\begin{align*}
\tilde{\D}_2\tilde{\bbeta}&= \R\{\tilde{\D}_1\balpha_R\}-\I\{\tilde{\D}_1\balpha_I\}
= \R\{\tilde{\D}_1(\balpha_R+i\balpha_I)\}
=  \R\{\tilde{\D}_1\balpha\}.
\end{align*}
Combining this with~\eqref{eqn:alpha},  we get
\begin{align*}
\tilde{\D}_2\tilde{\bbeta}&=  \R\{\tilde{\D}_1{\D_0(\f^\star)}^{-1}[\s^\star+{\D^\star}_1^H{\bbeta}]\}
\\
&=  \R\{\tilde{\D}_1{\D_0(\f^\star)}^{-1} \s^\star  \}+\R\{\tilde{\D}_1{\D_0(\f^\star)}^{-1}\}{\D^\star}_1^H {\bbeta}
\\
&=  \R\{\tilde{\D}_1{\D_0(\f^\star)}^{-1} \s^\star  \}+\R\{\tilde{\D}_1{\D_0(\f^\star)}^{-1} \}{\D^\star}_1^H\diag(\s^\star)\tilde{\bbeta}
\\
&=  \R\{\tilde{\D}_1{\D_0(\f^\star)}^{-1} \s^\star  \}+\R\{\tilde{\D}_1{\D_0(\f^\star)}^{-1} \}\tilde{\D}_1^H\tilde{\bbeta}.
\end{align*}
This implies
\begin{align}\label{eqn:beta:a}
(\tilde{\D}_2-\R\{\tilde{\D}_1{\D_0(\f^\star)}^{-1} \}\tilde{\D}_1^H)\tilde{\bbeta}=\R\{\tilde{\D}_1{\D_0(\f^\star)}^{-1}\s^\star\}.
\end{align}

\subsubsection{\bf Bounding \texorpdfstring{$\|\tilde{\bbeta}\|_\infty$}{Lg}}
First invoke~\eqref{eqn:kernel:matrix:small:norm} to get
\begin{equation}\label{F:eqn:norm:kernel:matrix}
\begin{aligned}
\|{\D_0(\f^\star)}^{-1}\|_{\infty, \infty}& \leq \frac{1}{1-0.00755 }, 
\\
\{\|\D_1(\f^\star)\|_{\infty, \infty}, \|\tilde{\D}_1\|_{\infty, \infty} \}/\sqrt{\tau}&\leq 0.01236n/\sqrt{\tau}\leq0.00682, %, 0.01236n;
\\
\|\eye-\tilde{\D}_2/\tau\|_{\infty, \infty}&\leq0.0171.%0.05610 n^2.
\end{aligned}
\end{equation}
These inequalities~\eqref{F:eqn:norm:kernel:matrix} immediately lead to
\begin{equation}\label{eqn:F:I-X}
\begin{aligned}
%&
\|\tau\eye-\tilde{\D}_2+\R\{\tilde{\D}_1{\D_0(\f^\star)}^{-1} \}\tilde{\D}_1^H\|_{\infty, \infty}
%\\
\stack{\ding{172}}{\leq}& \tau\left(\|\eye-\tilde{\D}_2/\tau\|_{\infty, \infty}+\|\tilde{\D}_1/\sqrt{\tau}\|_{\infty, \infty}^2\|{\D_0(\f^\star)}^{-1}\|_{\infty, \infty}\right)
\\
\stack{\ding{173}}{\leq}& \tau\left(0.0171 + 0.00682^2/(1-0.00755)\right)
\leq  0.01715\tau
{<}\tau, 
\end{aligned}
\end{equation}
where \ding{172} follows from the triangle inequality and the sub-multiplicative property of $\ell_{\infty, \infty}$ norm and \ding{173} follows from~\eqref{F:eqn:norm:kernel:matrix}. This implies that $ \tilde{\D}_2-\R\{\tilde{\D}_1{\D_0(\f^\star)}^{-1} \}\tilde{\D}_1^H $ is nonsingular and well-conditioned. In particular, 
\begin{align*}
\|(\tilde{\D}_2-\R\{\tilde{\D}_1{\D_0(\f^\star)}^{-1} \}\tilde{\D}_1^H)^{-1}\|_{\infty, \infty}
&\stack{\ding{172}}{\leq} \frac{1}{\tau(1-0.01715)}\leq \frac{1.0175}{\tau}, 
\end{align*}
where \ding{172} follows from~\eqref{eqn:F:I-X}. Then from~\eqref{eqn:beta:a},  we have
\begin{equation}\label{eqn:beta:norm}
\begin{aligned}
\left\|\tilde{\bbeta}\right\|_\infty
\leq& \|(\tilde{\D}_2-\R\{\tilde{\D}_1{\D_0(\f^\star)}^{-1} \}\tilde{\D}_1^H)^{-1}\|_{\infty, \infty}\|\R\{\tilde{\D}_1{\D_0(\f^\star)}^{-1}\s^\star\}\|_\infty
\\
\stack{\ding{172}}{\leq}& \|(\tilde{\D}_2-\R\{\tilde{\D}_1{\D_0(\f^\star)}^{-1} \}\tilde{\D}_1^H)^{-1}\|_{\infty, \infty}
\|\tilde{\D}_1\|_{\infty, \infty}\|{\D_0(\f^\star)}^{-1}\|_{\infty, \infty}\|\s^\star\|_\infty
\\
\stack{\ding{173}}{\leq}&  \frac{1.0175}{\tau}\frac{0.00682\sqrt\tau}{1-0.00755}
\leq  \frac{0.00700}{\sqrt{\tau}}, 
\end{aligned}
\end{equation}
where \ding{172} follows from sub-multiplicative property of the operator norm $\|\cdot\|_{\infty, \infty}$,  and \ding{173} follows from  Eq.~\eqref{F:eqn:norm:kernel:matrix} and $\|\s^\star\|_\infty=1$. This indicates that
 \begin{align}
\|{\bbeta}\|_\infty
\leq  \left\|\diag(\s^\star)\right\|_{\infty, \infty}\|\tilde{\bbeta}\|_\infty\leq 0.00700/\sqrt{\tau} \leq 0.00386/n := \beta^\infty,  \label{F:eqn:beta}
\end{align}
where the last inequality follows because $\tau\geq 3.289n^2$ for $n\geq130$ by~\eqref{eqn:sec:A:tau}.

\subsubsection{Bounding \texorpdfstring{$\|{\balpha}\|_\infty$ and $\R\{\alpha_1\}$ and $\left|\I\{\alpha_1\}\right|$}{Lg}}
From~\eqref{eqn:alpha},  we have
\begin{equation}\label{eqn:alpha:norm}
\begin{aligned}
\|\balpha\|_\infty
\stack{\ding{172}}{\leq}& \|{\D_0(\f^\star)}^{-1}\|_{\infty, \infty}\|\s^\star\|_\infty+\|{\D_0(\f^\star)}^{-1}\|_{\infty, \infty}\|{\D_1(\f^\star)}\|_{\infty, \infty}\|\bbeta\|_\infty
\\
\stack{\ding{173}}{\leq}& \frac{1}{1-0.00755 }+\frac{0.00682\sqrt{\tau}}{1-0.00755}\frac{0.00700}{\sqrt{\tau}}
\\\leq&1.00766 := \alpha^\infty, 
\end{aligned}
\end{equation}
where \ding{172} follows from the triangle inequality and the fact that $\|\mA\mB\x\|_\infty\leq\|\mA\|_{\infty, \infty}\|\mB\|_{\infty, \infty}\|\x\|_\infty$. \ding{173} holds since $\|\s^\star\|_\infty=1.$

Second,  recognizing that $\alpha_1= 1 -[(\eye-{\D_0(\f^\star)}^{-1})\s^\star-{\D_0(\f^\star)}^{-1}{\D_1(\f^\star)}^H\bbeta ]_1$ by Eq.~\eqref{eqn:alpha1},  we have
$
\R\{\alpha_1\}= 1 -[\R\{(\eye-{\D_0(\f^\star)}^{-1})\s^\star-{\D_0(\f^\star)}^{-1}{\D_1(\f^\star)}^H\bbeta \}]_1.
$
We further get an upper bound as follows 
\begin{align*}
&\big|\big[\R\big\{(\eye-{\D_0(\f^\star)}^{-1})\s^\star-{\D_0(\f^\star)}^{-1}{\D_1(\f^\star)}^H\bbeta \big\}\big]_1\big|\\
\stack{\ding{172}}{\leq}& \|(\eye-{\D_0(\f^\star)}^{-1})\s^\star-{\D_0(\f^\star)}^{-1}{\D_1(\f^\star)}^H\bbeta\|_\infty
\\
\stack{\ding{173}}{\leq}& \|{\D_0(\f^\star)}^{-1}\|_{\infty, \infty}\|\eye-{\D_0(\f^\star)} \|_{\infty, \infty}\|\s^\star\|_\infty
+ \|{\D_0(\f^\star)}^{-1}\|_{\infty, \infty} \|\D_1(\f^\star)\|_{\infty, \infty} \|\bbeta\|_\infty
\\
\leq& \frac{0.00755}{1-0.00755 }+  \frac{0.00682\sqrt{\tau}}{1-0.00755}\frac{0.00700}{\sqrt{\tau}}
\\
\leq& 0.00766, 
\end{align*}
where \ding{172} follows from the real part of the first entry of a vector is no larger than the infinity norm of this vector and \ding{173} follows from the triangle inequality and the sub-multiplicative property of infinity operator norm that $\|\mA\mB\x\|_\infty\leq\|\mA\|_{\infty, \infty}\|\mB\|_{\infty, \infty}\|\x\|_\infty$. The last inequality follows from Eq.~\eqref{F:eqn:norm:kernel:matrix} and~\eqref{F:eqn:beta}. Combining the above arguments yields
\begin{equation}\label{eqn:alpha1:norm}
\begin{aligned}
\R\{\alpha_1\}
\geq& 1 -0.00766
\quad\text{ and }\quad
\left|\I\{\alpha_1\}\right|
\leq0.00766.
\end{aligned}
\end{equation}

We are ready to show the Boundedness property following the simplifications used in~\cite{Candes:2014br}. In particular,  fix an arbitrary point $f_0^\star\in T^\star$ as the reference point and let $f_{-1}^\star$ be the first frequency in $T^\star$ that lies on the left of $f_0^\star$ while $f_{1}^\star$ be the first frequency in $T^\star$ that lies on the right. Here ``left'' and ``right'' are directions on the complex circle $\TT$. We remark that the analysis depends only on the relative locations of $\{f_\ell^\star\}$. Hence,  to simplify the arguments,  we assume that the reference point $f_0^\star$ is at $0$ by shifting the frequencies if necessary.
Then we divide the region between $f_0^\star = 0$ and $f_{1}^\star/2$ into three parts: Near Region $\N:=[0, 0.24/n]$,  Middle Region $\M:=[0.24/n, 0.75/n]$ and Far Region $\F:=[0.75/n,  f_{1}^\star/2]$. Also their symmetric counterparts: $-\N:=[-0.24/n,  0]$,  $-\M:=[-0.75/n, -0.24/n]$,  and $-\F:=[f_{-1}^\star/2,  -0.75/n]$.  We first show that the dual polynomial has strictly negative curvature $|Q^\star(f)|''< 0$ in $\N=[0, 0.24/n]$ and $|Q^\star(f)|< 1$ in $\M\cup\F=[0.24/n,  f_{1}^\star/2]$,  implying $|Q^\star(f)|< 1$ in $\N\cup\M\cup\F\backslash\{f_0^\star\}$ by exploiting $|Q^\star(f_0^\star)|=1$ and $|Q^\star(f_0^\star)|'=0$.  Then using the same symmetric arguments in~\cite{Candes:2014br},  we claim that  $|Q^\star(f)|< 1$ in $(-\N)\cup(-\M)\cup(-\F)\backslash\{f_0^\star\}$. Combining these two results with the fact that the reference point $f_0^\star$ is chosen arbitrarily from $T^\star$ (and shifted to $0$),  we establish that the Boundedness property of $Q^\star(f)$ holds in the entire $\TT \backslash T^\star$.

\subsubsection{Controlling \texorpdfstring{$Q^\star(f)$}{Lg} in Near Region}
 
For $f\in\N$,  the second-order Taylor expansion of $|Q^\star(f)|$ at ${f}^\star_0 = 0$ states
\begin{align}\label{eqn:taylor}
|Q^\star(f)| &= |Q^\star(f^\star_0)| + (f-f^\star_0) |Q^\star(f^\star_0)|' + \frac{1}{2}(f-f^\star_0)^2 |Q^\star(\xi)|''\nonumber\\
& = 1 + (f-f^\star_0) |Q^\star(f^\star_0)|' + \frac{1}{2}(f-f^\star_0)^2 |Q^\star(\xi)|''\text{\ for\ some\ } \xi \in \N, 
\end{align}
with the second line following from the Interpolation property.
We argue that
\begin{align*}
|Q^\star(f^\star_0)|' &= \frac{Q^\star_R(f^\star_0)Q^\star_R(f^\star_0)'+Q^\star_I(f^\star_0)Q^\star_I(f^\star_0)'}{|Q^\star(f^\star_0)|} = \frac{\R\{c^\star_0\}Q^\star_R(f^\star_0)'+\I\{c^\star_0\}Q^\star_I(f^\star_0)'}{|c^\star_0||Q^\star(f^\star_0)|}
 = \frac{\R\{c^{\star H}_0 Q^\star(f^\star_0)'\}}{|c^\star_0||Q^\star(f^\star_0)|} = 0.
\end{align*}
The last equality is due to~\eqref{eqn:gradientzero}. Together with~\eqref{eqn:taylor},  to bound $|Q^\star(f)|$ strictly below 1,   we only need to show  the concavity of $|Q^\star(f)|$ in Near Region (i.e.,  $|Q^\star(f)|''<0$ for $f\in\N$). Since
\begin{align*}
|Q^\star(f)|'' = -\frac{(Q^\star_R(f)Q^\star_R(f)'+Q^\star_I(f)Q^\star_I(f)')^2}{|Q^\star(f)|^3} + \frac{Q^\star_R(f){Q_R^\star}(f)''+|{Q^\star}(f)'|^2+|Q^\star_I(f)||{Q^\star_I}(f)''|}{|Q^\star(f)|}, 
\end{align*}
we only need to show that
\begin{align*}
Q^\star_R(f){Q_R^\star}(f)''+|{Q^\star}(f)'|^2+|Q^\star_I(f)||{Q^\star_I}(f)''|<0.
\end{align*}

Recall the expression for $Q^\star(f)$ given in Eq.~\eqref{eqn:sec:4:Q:star}
 \begin{align*}
Q^\star(f) =\sum_{f_\ell^\star \in T^\star} \alpha_\ell K(f^\star_\ell-f)+\sum_{f_\ell^\star \in T^\star}\beta_\ell K'(f^\star_\ell-f).
\end{align*}
To bound the real part of $Q^\star(f)$ in $\N = [0, 0.24/n]$,  we observe
\begin{align*}
{Q_R^\star}(f)& \geq \R\{\alpha_1 K(f)\} - \alpha^\infty \sum_{f_\ell^\star \in T^\star\backslash\{0\}}|K(f-f_\ell^\star)| - \beta^\infty |K'(f)| - \beta^\infty \sum_{f_\ell^\star \in T^\star\backslash\{0\}}|K'(f-f_\ell^\star)|\\
& \geq \R\{\alpha_1\}\min_{f\in\N}K(f)-\alpha^\infty F_0(2.5/n, f)
-\beta^\infty(\max_{f\in\N}|K'(f)|+F_1(2.5/n, f))
\\
& \geq (1-0.00766)(0.905252)- (1.00766)0.00757
-( 0.00386/n)(\bl{0.789569}n+0.01241n)
\\
& \geq \bl{0.887594}, 
\end{align*}
where the first inequality follows from an application of the triangle inequality,  and the second is from Lemma~\ref{lem:bounds:sum:K}. The third inequality follows from evaluating $F_0(2.5/n,  f)$ and $F_1(2.5/n,  f)$ at $f = 0.24/n$,  the numerical bounds in Tables~\ref{tbl:bounds:1} and~\ref{tbl:bounds:2} of Appendix~\ref{sec:numerical:bound} and Eq.~\eqref{F:eqn:beta}, ~\eqref{eqn:alpha:norm}, ~\eqref{eqn:alpha1:norm},  as well as  $\min_{f\in\N}K(f)\geq 0.905252.$ This last bound follows from~\cite[Eq. (2.20),  set $f_c=n-2$]{Candes:2014br}  that
$K(f)\geq 1- \frac{\pi^2}{6}(n-2)(n+2)f^2.$
Hence
\begin{align*}
\min_{f\in\N}K(f)&\geq\min_{f\in\N} 1- \frac{\pi^2}{6}(n-2)(n+2)f^2
\geq 1-\frac{\pi^2}{6}(n-2)(n+2)(0.24/n)^2
\geq 0.905252.
\end{align*}

{\allowdisplaybreaks
Similarly,  combining  Eq.~\eqref{F:eqn:beta}, ~\eqref{eqn:alpha:norm}, ~\eqref{eqn:alpha1:norm},  the upper bounds on $F_\ell(2.5/n, 0.24/n)$
in Table~\ref{tbl:bounds:1} and the upper bounds for   $\max_{f\in\N}|K^{(\ell)}(f)|$ and $\max_{f\in\N}K''(f)$ in Table~\ref{tbl:bounds:2},  we get
\begin{align*}
{Q_R^\star}''(f)
\leq& \R\{\alpha_1\}\max_{f\in\N}K''(f)+\alpha^\infty F_2(2.5/n, 0.24/n)
+\beta^\infty(\max_{f\in\N}|K'''(f)|+F_3(2.5/n, 0.24/n))
\\
\leq& (1-0.00766)(-2.35084n^2)+ (1.00766)(0.05637n^2)
+( 0.00386/n)(7.79273n^3+0.28838n^3)
\\
\leq& -2.24483 n^2;
\\[1ex]
|{Q^\star_I}(f)|
\leq& |\I\{\alpha_1\}|\max_{f\in\N}K(f)+\alpha^\infty F_0(2.5/n, 0.24/n)
+\beta^\infty(\max_{f\in\N}|K'(f)|+F_1(2.5/n, 0.24/n))
\\
\leq&(0.00766)\times 1+(1.00766)0.00757
+( 0.00386/n)(\bl{0.789569}n+0.01241n)
\\
\leq&\bl{0.0183836};
\\[1ex]
|{Q^\star_I}''(f)|
\leq& |\I\{\alpha_1\}|\max_{f\in\N}|K''(f)|+\alpha^\infty F_2(2.5/n, 0.24/n)
+\beta^\infty(\max_{f\in\N}|K'''(f)|+F_3(2.5/n, 0.24/n))
\\
\leq& (0.00766)(3.290n^2)+ (1.00766)(0.05637n^2)
+( 0.00386/n)(7.79273n^3+0.28838n^3)
\\
\leq& 0.113197  n^2;
\\[1ex]
|{Q^\star}'(f)|
\leq& \alpha^\infty(\max_{f\in\N}|K'(f)|+F_1(2.5/n, 0.24/n))
+\beta^\infty(\max_{f\in\N}|K''(f)|+F_2(2.5/n, 0.24/n))
\\
\leq& ( 1.00766)(0.789569n+0.01241n)
+( 0.00386/n)(3.290n^2+0.05637n^2)
\\
\leq& 0.821039 n;
\\[1ex]
|{Q^\star}''(f)|
\leq& \alpha^\infty(\max_{f\in\N}|K''(f)|+F_2(2.5/n, 0.24/n))
+\beta^\infty(\max_{f\in\N}|K'''(f)|+F_3(2.5/n, 0.24/n))
\\
\leq& ( 1.00766)(3.290n^2+0.05637n^2)
+( 0.00386/n)(7.79273n^3+0.28838n^3)
\\
\leq& 3.40320 n^2.
\end{align*}
Combining the lower bound on ${Q}^\star_R(f)$ and the upper bounds on ${Q_R^\star}(f)''$,  $|{{Q}^\star}(f)'|$,  $|{Q^\star_I}(f)|$ and
$|{{Q}^\star_I}(f)''|$,  we arrive at
\begin{align*}
&|{Q}^\star(f)|''
={Q}^\star_R(f){Q_R^\star}(f)''+|{{Q}^\star}(f)'|^2+|{Q^\star_I}(f)||{{Q}^\star_I}(f)''|
\leq -1.316313n^2 <0 \ \text{in }\N.
\end{align*}
}

\subsubsection{ Bounding \texorpdfstring{$|Q^\star(f)|$}{Lg}   in   Middle Region}
 
For upperbounding $|Q^\star(f)|$ for  $f \in \M  = [0.24/n,  0.75/n]$,  we firstly apply the triangle inequality 
\begin{align}
|Q^\star(f)|
& = |\sum_{f_\ell^\star \in T^\star} \alpha_\ell K(f^\star_\ell-f)+\sum_{f_\ell^\star \in T^\star}\beta_\ell K'(f^\star_\ell-f)|\nonumber\\
& \leq \|\balpha\|_\infty \bigg(|K(f)| + \sum_{f_\ell^\star \in T^\star\backslash\{0\}}|K(f-f_\ell^\star)|\bigg) + \|\bbeta\|_\infty\bigg( |K'(f)| + \sum_{f_\ell^\star \in T^\star\backslash\{0\}}|K'(f-f_\ell^\star)|\bigg)\nonumber\\
& \leq \alpha^\infty |K(f)| +  \beta^\infty |K'(f)| + \alpha^\infty F_0(2.5/n,  f)  + \beta^\infty F_1(2.5/n,  f), 
\end{align}
where the last inequality is from Lemma~\ref{lem:bounds:sum:K}.  
We then follow \cite[Eq. (2.29)]{Candes:2014br} to upperbound the first two terms in the last line
\begin{align*}
|K(f)| & \leq 1 - \frac{\pi^2(n^2-4)f^2}{6} + \frac{\pi^4n^4 f^4}{72} \qquad\text{and}\qquad
|K'(f)|  \leq \frac{\pi^2 (n^2-4) f}{3}, \qquad \text{for $f \in [-1/2,  1/2]$.}
\end{align*}
The rest of argument consists of defining
\begin{align*}
L_1(f) & =\alpha^\infty\left(1-\frac{1}{6} \pi ^2 (n^2-4) f^2+\frac{1}{72} \pi ^4 n^4 f^4\right)+\beta^\infty\frac{1}{3} \pi^2  (n^2-4) f; 
\\
L_2(f)&=\alpha^\infty F_0(2.5/n, f)+\beta^\infty F_1(2.5/n, f)
\end{align*}
with the derivative of $L_1(f)$ given by
\begin{align*}
L_1'(f) = -\alpha^\infty \left(\frac{\pi^2(n^2-4)f}{3}-\frac{\pi^4 n^4 f^3}{18}\right) + \beta^\infty \frac{\pi^2 (n^2-4)}{3}<0, \quad\text{for }f\in\M,
\end{align*}
implying that $L_1(f)$ is decreasing. Also,   $L_2(f)$ is increasing in $\M$ by Lemma~\ref{lem:bounds:sum:K}. Hence,  by the monotonic property,  we have
\begin{align*}
|Q^\star(f)| \leq&L_1(0.24/n)+ L_2(0.75/n)
\leq  0.919779+0.007836 = 0.927615 < 1.
\end{align*}

\subsubsection{ Bounding \texorpdfstring{$|Q^\star(f)|$}{Lg}  in Far Region. }
 
Recall that $f_0^\star = 0$ is the reference point. To simplify notation,  we re-index the frequencies such that
$\ldots\leq f_{-1}^\star<f_0^\star=0<f_1^\star<\ldots$.  For $f\in\F=[0.75/n,  f_1^\star/2] = [0.75/n, f_1^\star-f_1^\star/2]$,  by Lemma~\ref{lem:bounds:sum:K:far},   we have
\begin{align}
\sum_j |K^{(\ell)}(f-f_j^\star)|
\leq W_\ell (0.75/n, f_1^\star/2)\nn
& \stack{\ding{172}}{=} \sum_{j\geq0}B_\ell(j ({2.5/n})+0.75/n) + \sum_{j\geq0}B_\ell(j ({2.5/n}) +f_1^\star/2)
\nn\\
&\stack{\ding{173}}{\leq} \sum_{j\geq0}B_\ell(j ({2.5/n})+0.75/n) + \sum_{j\geq0}B_\ell(j ({2.5/n}) +1.25/n)\nn\\
&\stack{\ding{174}}{=}W_\ell(0.75/n, 1.25/n), \label{eqn:far:b}
\end{align}
where \ding{172} follow from the definition of $W(\underline f, \bar f)$ in Lemma~\ref{lem:bounds:sum:K:far}, 
\ding{173} follows $f_1^\star/2=(f_1^\star-f_0^\star)/2\geq\Delta_{\min}/2=1.25/n$ and decreasing property of $B_\ell(\cdot), $ and
\ding{174} follows from the definition of $W(\underline f, \bar f)$.

Finally,  applying~\eqref{eqn:far:b}, ~\eqref{F:eqn:beta} and~\eqref{eqn:alpha:norm}  to~\eqref{eqn:sec:4:Q:star},  we arrive at
\begin{align*}
|Q^\star(f)|
\leq& \alpha^\infty \sum_\ell |K(f-f_\ell^\star)|+\beta^\infty \sum_\ell |K'(f-f_\ell^\star)|
\\\leq&   1.00766W_0(0.75/n, 1.25/n) +(0.00386/n)W_1(0.75/n, 1.25/n)
\\
\leq&  1.00766(0.70859)+  (0.00386/n)(5.2084n)
\\=& 0.734123.
\end{align*}
This concludes the proof of Lemma~\ref{lem:q0:dual:certificate}.
\end{proof}

\section{Proof of Lemma~\ref{lem:Q:lambda:closeto:Q:0}} \label{sec:G}

\begin{proof}
We exploit the closeness of $\btheta^\star$ and $\btheta^\lambda$ (see Lemma~\ref{lem:fix1}) to bound the pointwise distance between $Q^\star(f)$ and $Q^\lambda(f)$. Note
\begin{align*}
Q^\lambda (f) - Q^\star(f) = \a(f)^H\mZ(\q^\lambda - \q^\star) = \a(f)^H\mZ\left(\frac{\x^\star - \x^\lambda}{\lambda} + \frac{\mathrm{d} }{\mathrm{d} \lambda} \x^\lambda\big|_{\lambda=0}\right)  =
\frac{1}{\lambda}\int_0^\lambda \a(f)^H\mZ\left( \frac{\mathrm{d}}{\mathrm{d} t}\x^\star- \frac{\mathrm{d}}{\mathrm{d} t}\x^t \right)\mathrm{d}t, 
\end{align*}
which implies that
\begin{align}\label{eqn:diffQlambdaQstar}
|Q^\lambda(f)-Q^\star(f)| \leq  \max_{0\leq t\leq \lambda}
\left|\a(f)^H\mZ( \frac{\mathrm{d}}{\mathrm{d} t}\x^\star-\a(f)^H\mZ( \frac{\mathrm{d}}{\mathrm{d} t}\x^t\right|.
\end{align}
We can also obtain similar bounds on the pointwise distances between derivatives of $Q^\lambda(f)$ and $Q^\star(f)$.

Recall from Eq.~\eqref{eqn:sec:4:2}, ~\eqref{eqn:thetadiff},  and~\eqref{eqn:sec:4:5} that
\begin{align}\label{sec:G:partial:x}
 \frac{\mathrm{d}}{\mathrm{d} \lambda} \x^\lambda
&=-[\mA'(\f^\lambda)\diag(\c^\lambda)~~~\mA(\f^\lambda)~~~i\mA(\f^\lambda)](\nabla^2 G^\lambda(\btheta^\lambda))^{-1} \brho^\lambda, \\
 \frac{\mathrm{d}}{\mathrm{d} \lambda} \x^\star
&=-[\mA'(\f^\star)\diag(\c^\star)~~~\mA(\f^\star)~~~i\mA(\f^\star)](\nabla^2 G^0(\btheta^\star))^{-1} \brho^\star, \label{sec:G:partial:xstar}
\end{align}
where $\brho^\star = \begin{bmatrix}
\zero^T & \R\{\sign(\c^\star)\}^T & \I\{\sign(\c^\star)\}^T
\end{bmatrix}^T$
and $\brho^\lambda = \begin{bmatrix}
\zero^T & \R\{\sign(\c^\lambda)\}^T & \I\{\sign(\c^\lambda)\}^T
\end{bmatrix}^T$.
%denoting the complex sign vector for $\c^\star$,  i.e.,  $[\s^\star]_\ell=c^\star_\ell/|c^\star_\ell|.$ \note{Stopped here.}

Multiplying both sides of  Eq.~\eqref{sec:G:partial:x} and~\eqref{sec:G:partial:xstar} by $-\a(f)^H\mZ($
and then inserting $\W^{\star \frac{1}{2}}\W^{\star-\frac{1}{2}}$ (which equals $\eye$) into the spaces before and after  $(\nabla^2G^0(\btheta^\star))^{-1}$ (and $(\nabla^2 G^\lambda(\btheta^\star))^{-1}$) yield
\begin{align*}
&-\a(f)^H\mZ( \frac{\mathrm{d}}{\mathrm{d} \lambda}\x^\lambda=\bnu^\lambda(f)\Xi^\lambda\brho^\lambda, 
\\
&-\a(f)^H\mZ( \frac{\mathrm{d}}{\mathrm{d} \lambda}\x^\star=\bnu^\star(f) \Xi^\star\brho^\star.
\end{align*}
Here
\begin{equation}\label{eqn:v}
\begin{aligned}
\bnu^\lambda(f)&:=[\D_1(f, \f^\lambda ) \diag(\c^\lambda)\mS^{-1}~~\D_0(f, \f^\lambda)~i\D_0(f, \f^\lambda)], 
\\
\bnu^\star(f)&:=[\D_1(f, \f^\star) \diag(\c^\star)\mS^{-1}~~\D_0(f, \f^\star)~i\D_0(f, \f^\star)], 
\end{aligned}
\end{equation}
with  $\D_\ell(f, \f^\lambda)$ a row vector defined by $\D_\ell(f, \f^\lambda):=[K_\ell(f^\lambda_1-f), \ldots, K_\ell(f^\lambda_k-f)]$,  and
\begin{align*}
\Xi^\lambda&:=\Upsilon(\nabla^2G^\lambda(\btheta^\lambda))^{-1}, \\
\Xi^\star&:=\Upsilon(\nabla^2G^0(\btheta^\star))^{-1}, 
\end{align*}
where $\Upsilon(\cdot):={\W}^{\star\frac{1}{2}}(\cdot){\W}^{\star\frac{1}{2}}$ is  a linear operator that normalizes the Hessian matrix so that it is close to the identity.
As a consequence,  we bound the integrand of~\eqref{eqn:diffQlambdaQstar} as follows 
\begin{equation}\label{eqn:sec:G:QW:abc}
\begin{aligned}
&|\bnu^\star(f) \Xi^\star\brho^\star-\bnu^\lambda(f)\Xi^\lambda\brho^\lambda|
\\
{\leq}& |{\bnu^\star(f)}\Xi^\star(\brho^\star-\brho^\lambda)|+|{\bnu^\star(f)}(\Xi^\star-\Xi^\lambda)\brho^\lambda|+|{(\bnu^\star(f)-\bnu^\lambda(f))}\Xi^\lambda\brho^\lambda|
\\
{\leq}&  \|\bnu^\star(f)\|_1\|\Xi^\star\|_{\infty, \infty}\|\brho^\star-\brho^\lambda\|_{\infty}+\|\bnu^\star(f)\|_1\|\Xi^\star-\Xi^\lambda\|_{\infty, \infty}\|\brho^\lambda\|_\infty
+\|\bnu^\star(f)-\bnu^\lambda(f)\|_1\|\Xi^\lambda\|_{\infty, \infty}\|\brho^\lambda\|_\infty, 
\end{aligned}
\end{equation}
where the first line follows from the triangle inequality and the second line follows from H\"{o}lder's inequality and the sub-multiplicative property of the $\ell_{\infty, \infty}$ norm. We next develop upper bounds on $\|\bnu^\star(f)\|_1$,  $\|\bnu^\star(f)-\bnu^\lambda(f)\|_1$,  $\|\Xi^\star\|_{\infty, \infty}$,  $\|\Xi^\star-\Xi^\lambda\|_{\infty, \infty}$,  $\|\Xi^\lambda\|_{\infty, \infty}$,  $\|\brho^\star-\brho^\lambda\|_1$ and $\|\brho^\lambda\|_\infty$.

\bigskip
\noindent{\bf Bounding $\|{\Xi^\star}\|_{\infty, \infty}$ and $\|\Xi^\lambda\|_{\infty, \infty}$ and $\|{\Xi^\star}-\Xi^\lambda\|_{\infty, \infty}$.}
\\
We note that both  ${\Xi^\star}^{-1} = \Upsilon(\nabla^2 G^0(\btheta^\star))$) and ${\Xi^\lambda}^{-1} = \Upsilon(\nabla^2 G^\lambda(\btheta^\lambda))$) are close to the identity matrix. More precisely,  we have
\begin{align*}
\|\eye-{\Xi^\star}^{-1}\|_\infty
&\stack{\ding{172}}{=}\|\eye-\Upsilon(\nabla^2 G^0(\btheta^\star))\|_{\infty, \infty}
\\
&\stack{\ding{173}}{\leq}  [ \|\eye-\diag\left( {\c^\star}./{|\c^\star|}\right)^H(-\D_2(\f^\star)/\tau)\diag\left( {\c^\star}./{|\c^\star|}\right) \|_{\infty, \infty}
%\\
%&~~~~~
+2\|\diag ( {\c^\star}./{|\c^\star|})\D_1(\f^\star)/\sqrt {\tau} \|_{\infty, \infty} ]
\\
&~~~~~\vee  ( \|\diag ( {\c^\star}./{|\c^\star|} )\D_1(\f^\star)/\sqrt {\tau} \|_{\infty, \infty}+\|\eye-\D_0(\f^\star)\|_{\infty, \infty}  )
\\
&\stack{\ding{174}}{\leq} \|\eye-(-\D_2(\f^\star)/\tau) \|_{\infty, \infty}+2 \|\D_1(\f^\star)/\sqrt {\tau} \|_{\infty, \infty}
\\
&\stack{\ding{175}}{\leq} 0.0171+2\times0.00682
\\
&\leq 0.03074, 
\end{align*}
where $a\vee b:=\max(a, b).$
\ding{172} follows from definition of ${\Xi^\star}$ and \ding{173} follows from applying the triangle inequality to the expression of $[\eye-\Upsilon(\nabla^2 G^0(\btheta^\star))].$ \ding{174} follows since the infinity norm of any sign vector is 1,  bounding $\left\|\diag ( {\c^\star}./{|\c^\star|}\right) {\D_1(\f^\star)}/{\sqrt {\tau}} \|_{\infty, \infty}$ is equivalent to bound $ \| \D_1(\f^\star)/\sqrt {\tau} \|_{\infty, \infty}$. Finally,  \ding{175} follows from Eq.~\eqref{eqn:kernel:matrix:small:norm}.
This leads to
\begin{equation}\label{eqn:sec:G:control:B.1}
\begin{aligned}
\|\Xi^\star\|_{\infty, \infty}
\leq& \frac{1}{1-\|\eye-{\Xi^\star}^{-1}\|_{\infty, \infty}}
\leq \frac{1}{1-0.03074}
\leq 1.03172.
\end{aligned}
\end{equation}
According to~\eqref{eqn:norm:Theta:lambda},  we have
\begin{align*}
\|\eye-{\Xi^\lambda}^{-1}\|_{\infty, \infty}
=\|I-\Upsilon(\nabla^2 G^\lambda(\btheta^\lambda))\|_{\infty, \infty}
\leq 0.08561, 
\end{align*}
yielding
\begin{equation}\label{eqn:sec:G:control:B.2}
\begin{aligned}
\|\Xi^\lambda\|_{\infty, \infty}
\leq& \frac{1}{1-\|\eye-{\Xi^\lambda}^{-1}\|_{\infty, \infty}}
\leq \frac{1}{1-0.08561}
\leq 1.09363.
\end{aligned}
\end{equation}
Next,  note
\begin{align*}
\|{\Xi^\star}^{-1}-{\Xi^\lambda}^{-1}\|_{\infty, \infty}
=&\|\Upsilon(\nabla^2 G^0(\btheta^\star))-\Upsilon(\nabla^2 G^\lambda(\btheta^\lambda))\|_{\infty, \infty}
\leq \max\{\Pi_1, \Pi_2, \Pi_3\}, 
\end{align*}
where $\Pi_1, \Pi_2, \Pi_3$ denote the first,  second and third absolute row block sums of $[\Upsilon(\nabla^2 G^0(\btheta^\star))-\Upsilon(\nabla^2 G^\lambda(\btheta^\lambda))]$.
We first bound $\Pi_1$ as follows 
\begin{align}
\Pi_1=&
\|\diag( {\c}./{|\c^\star|})^H {\D_2(\f)} \diag( {\c}./{|\c^\star|})-\diag( {\c^\star}./{|\c^\star|})^H {\D_2(\f^\star)} \diag( {\c^\star}./{|\c^\star|})\|_{\infty, \infty}/\tau
\nn\\
&+2\left\|\diag( {\c}./{|\c^\star|})^H {\D_1(\f)} -\diag( {\c^\star}./{|\c^\star|})^H {\D_1(\f^\star)}  \right\|_{\infty, \infty}/\sqrt\tau
\nn\\
&+2 \left\|\diag\left( {1}./{|\c^\star|}\right)[\D_1(\f, \f^\star)\c^\star-\D_1(\f)\c]\right\|_{\infty  }/\sqrt\tau
\nn\\
&+ \|\diag( {\c}./{|\c^\star|^2})^H[\D_2(\f, \f^\star)\c^\star-\D_2(\f)\c]\|_\infty/\tau
\nn\\
\stack{\ding{172}}{\leq}& [2.19778 X^\star\gamma + 1.14168 (X^\star\gamma)^2]+2[1.48286 X^\star\gamma + 1.47604(X^\star\gamma)^2]
+2(0.75038)X^\star B^\star\gamma+1.14168X^\star B^\star\gamma
\nn\\
\leq& 7.81004X^\star B^\star\gamma\ (\text{by }B^\star X^\star\gamma\leq10^{-3}\label{eqn:g:1}), 
\end{align}
where \ding{172} follows from combining Eq.~\eqref{eqn:A:DC-DC1}-\eqref{eqn:A:DC-DC2} and~\eqref{eqn:F:Pi:1}-\eqref{eqn:F:Pi:2},  where
\eqref{eqn:F:Pi:1}-\eqref{eqn:F:Pi:2} are given by
\begin{align}
&\|\diag( {\c}./{|\c^\star|})^H {\D_2(\f)} \diag( {\c}./{|\c^\star|})-\diag( {\c^\star}./{|\c^\star|})^H {\D_2(\f^\star)} \diag( {\c^\star}./{|\c^\star|})\|_{\infty, \infty}/\tau
\nn\\
\leq&\|\diag( {\c}./{|\c^\star|})^H {\D_2(\f)} \diag( {(\c-\c^\star)}./{|\c^\star|})\|_{\infty, \infty}/\tau
\nn\\
&+\|\diag( {\c}./{|\c^\star|})^H ({\D_2(\f)-\D_2(\f^\star)}) \diag( {\c^\star}./{|\c^\star|})\|_{\infty, \infty}/\tau
\nn\\
&+\|\diag( ({\c-\c^\star})./{|\c^\star|})^H {\D_2(\f^\star)} \diag( {\c^\star}./{|\c^\star|})\|_{\infty, \infty}/\tau
\nn\\
\leq& (1+X^\star\gamma)(1.05610)(X^\star\gamma)+(1+X^\star\gamma)( 0.08558X^\star\gamma)+(X^\star\gamma)(1.05610)\tag{by \eqref{eqn:bound:diff:kernel:b4} and~\eqref{eqn:kernel:matrix:small:norm}}\\
\leq& 2.19778 X^\star\gamma + 1.14168 (X^\star\gamma)^2
\label{eqn:F:Pi:1}
\end{align}
and
\begin{align}
&\left\|\diag\left( {\c}./{|\c^\star|}\right)^H {\D_1(\f)} -\diag\left( {\c^\star}./{|\c^\star|}\right)^H {\D_1(\f^\star)}  \right\|_{\infty, \infty}/\sqrt\tau
\nn\\
\leq& \left\|\diag\left( {\c}./{|\c^\star|}\right)^H (\D_1(\f)-\D_1(\f^\star))  \right\|_{\infty, \infty}/\sqrt\tau
+\left\|\diag\left( {(\c-\c^\star)}./{|\c^\star|}\right)^H {\D_1(\f^\star)}  \right\|_{\infty, \infty}/\sqrt\tau
\nn\\
\leq&(1+X^\star\gamma)(1.47604X^\star\gamma)+(X^\star\gamma)(0.00682)  \tag{by Eq.~\eqref{eqn:bound:diff:kernel:b3}  and~\eqref{eqn:kernel:matrix:small:norm}}
\nn\\
\leq&
1.48286 X^\star\gamma + 1.47604 (X^\star\gamma)^2.
\label{eqn:F:Pi:2}
\end{align}

We next bound $\Pi_2$ and $\Pi_3$:
\begin{align*}
\{\Pi_2, \Pi_3\}
\stack{\ding{172}}{\leq}& \left\|\D_1(\f)\diag\left(\c./|\c^\star|\right)
-\D_1(\f^\star)\diag\left(\c^\star./|\c^\star|\right)\right\|_{{\infty, \infty}}/\sqrt\tau
\\
&+\|\D_0(\f)-\D_0(\f^\star)\|_{\infty, \infty}
+ \|\diag(1./|\c^\star|)[\D_1(\f, \f^\star)\c^\star-\D_1(\f)\c]\|_\infty/\sqrt\tau
\\
&+\lambda\|\u\odot\u./|\c|^3-\u^\star\odot\u^\star./|\c^\star|^3\|_\infty
+\lambda\|\u\odot\v./|\c|^3-\u^\star\odot\v^\star./|\c^\star|^3\|_\infty
\\
\stack{\ding{173}}{\leq}& [1.48286 X^\star\gamma + 1.47604 (X^\star\gamma)^2]+0.01516X^\star\gamma
+0.75038 X^\star B^\star\gamma+ 2(0.646X^\star\gamma)(5.00701)X^\star\gamma
\\
\leq&2.25636 X^\star B\gamma
\\
<& \Pi_1 \ (\text{by }B^\star X^\star\gamma\leq10^{-3}), 
\end{align*}
where \ding{172} follows from the triangle inequality and \ding{173} follows by combining Eq.~\eqref{eqn:F:Pi:2}, ~\eqref{eqn:bound:diff:kernel:b2}, ~\eqref{eqn:A:DC-DC1}, 
and~\eqref{eqn:F:u2c3-uic3}. To show~\eqref{eqn:F:u2c3-uic3},  we assume the norm $\|\u\odot\u./|\c|^3-\u^\star\odot\u^\star./|\c^\star|^3\|_\infty$ is achieved by the $\ell$th entry and proceed as
\begin{equation}
\begin{aligned}
\left|\frac{u_\ell^2}{|c_\ell|^3}-\frac{{u_\ell^\star}^2}{|c_\ell^\star|^3}\right|
\stack{\ding{172}}{\leq}& \left|\frac{c_\ell^2}{|c_\ell|^3}-\frac{{c_\ell^\star}^2}{|c_\ell^\star|^3}\right|
\\
\stack{\ding{173}}{\leq}& \frac{|c_\ell^2-{c_\ell^\star}^2|}{|c_\ell^\star|^3}+ |c_\ell|^2\left|\frac{1}{|c_\ell|^3}-\frac{1}{|c_\ell^\star|^3}\right|
\\
\stack{\ding{174}}{\leq}& \frac{X^\star\gamma}{c^\star_{\min}}\bigg((2+X^\star\gamma)+\frac{(X^\star\gamma)^2+3(X^\star\gamma)+3}{1-X^\star\gamma}\bigg)
\leq \frac{X^\star\gamma}{c^\star_{\min}}(5.00701)\label{eqn:F:u2c3-uic3}, 
\end{aligned}
\end{equation}
where \ding{172} follows from $|\R\{a\}|\leq|a|$ for all $a\in\C$ and \ding{173} follows from the triangle inequality. \ding{174} follows from  Eq.~\eqref{eqn:F:c2-c2} and~\eqref{eqn:F:c2c3-c2c3} given below:
\begin{equation}
\begin{aligned}
\frac{|c_\ell^2-{c_\ell^\star}^2|}{|c_\ell^\star|^3}
\leq&\frac{1}{|c_\ell^\star|}\frac{|c_\ell-c_\ell^\star|}{|c_\ell^\star|}\frac{|c_\ell+c_\ell^\star|}{|c_\ell^\star|}
\leq\frac{X^\star\gamma(2+X^\star\gamma)}{c^\star_{\min}}\label{eqn:F:c2-c2}
\end{aligned}
\end{equation}
and
\begin{equation}\label{eqn:F:c2c3-c2c3}
\begin{aligned}
|c_\ell|^2\left|\frac{1}{|c_\ell|^3}-\frac{1}{|c_\ell^\star|^3}\right|
= |c_\ell|^2 \left|\frac{1}{|c_\ell|}-\frac{1}{|c_\ell^\star|}\right|\left(\frac{1}{|c_\ell|^2}+\frac{1}{|c_\ell^\star|^2}+\frac{1}{|c_\ell^\star||c_\ell|} \right)
\leq& \frac{|c_\ell-c_\ell^\star|}{|c_\ell||c_\ell^\star|} \left(1+\frac{ |c_\ell|^2}{|c_\ell^\star|^2}+\frac{ |c_\ell|}{|c_\ell^\star|} \right)
\\
\leq& \frac{1}{|c_\ell|}X^\star\gamma \left(1+\frac{|c_\ell|^2}{|c_\ell^\star|^2}+\frac{|c_\ell|}{|c_\ell^\star|} \right)
\\
\leq& \frac{1}{c^\star_{\min}(1-X^\star\gamma)} X^\star\gamma\left(1+\frac{|c_\ell|^2}{|c_\ell^\star|^2}+\frac{|c_\ell|}{|c_\ell^\star|} \right)
\\
\leq&\frac{X^\star\gamma}{c^\star_{\min}}\frac{(X^\star\gamma)^2+3(X^\star\gamma)+3}{1-X^\star\gamma}
\end{aligned}
\end{equation}
where the first line follows from $|a^3-b^3|=|(a-b)(a^2+ab+b^2)|=|a-b|(a^2+ab+b^2)$ for any positive $a, b$. The second line holds since $ |\frac{1}{|c_\ell|}-\frac{1}{|c_\ell^\star|} |=\frac{\left||c_\ell|-|c_\ell^\star|\right|}{|c_\ell||c_\ell^\star|}\leq\frac{|c_\ell-c_\ell^\star|}{|c_\ell||c_\ell^\star|}$ by the triangle inequality. The third line follows from $ {|c_\ell-c_\ell^\star|}/{|c_\ell^\star|}\leq X^\star\gamma$ by~\eqref{eqn:parameter:close}. For the fourth line to hold,  note that by~\eqref{eqn:parameter:close},  $\frac{|c_i-c_i^\star|}{|c_i^\star|}\leq X^\star\gamma$,  which implies that $|c_i|\geq |c_i^\star| - |c_i-c_i^\star|\geq (1-X^\star\gamma|){c_{\min}^\star}|$. The last line follows from $|c_\ell|/|c^\star_\ell|\leq(1+X^\star\gamma).$
Finally,  we get the bound
\begin{align*}
\|{\Xi^\star}^{-1}-{\Xi^\lambda}^{-1}\|_{\infty, \infty}=\Pi_1\leq 7.81004X^\star B^\star\gamma
\end{align*}
implying
\begin{equation}\label{eqn:sec:G:control:B.3}
\begin{aligned}
\|\Xi^\star-\Xi^\lambda\|_{\infty, \infty}
\leq&\|\Xi^\star\|_{\infty, \infty}\|{\Xi^\star}^{-1}-{\Xi^\lambda}^{-1}\|_{\infty, \infty}\|\Xi^\lambda\|_{\infty, \infty}
\leq (1.03172)(1.09363)(7.81004X^\star B^\star\gamma)
= 8.81222 X^\star  B^\star\gamma.
\end{aligned}
\end{equation}

\bigskip
\noindent{\bf Bounding $\|\brho^\star-\brho^\lambda\|_\infty$ and $\|\brho^\lambda\|_\infty$.}
\\
First recognize that $\|\brho^\lambda\|_\infty=1$ since $\brho^\lambda$ contains either signs or zeros. Assume the $\ell_\infty$ norm of $(\brho^\star-\brho^\lambda)$ is achieved by $|\sign(c^\lambda_\ell)-\sign(c^\star_\ell)|$,  then applying triangle inequalities gives
\begin{equation}\label{eqn:sec:G:control:C}
\begin{aligned}
\|\brho^\star-\brho^\lambda\|_\infty
=\bigg|\frac{c^\lambda_\ell}{|c^\lambda_\ell|}-\frac{c_\ell^\star}{|c_\ell^\star|}\bigg|
=\bigg|\frac{c^\lambda_\ell}{|c^\lambda_\ell|}-\frac{c^\lambda_\ell}{|c^\star_\ell|}+\frac{c^\lambda_\ell}{|c^\star_\ell|}-\frac{c_\ell^\star}{|c_\ell^\star|}\bigg|
\leq& \bigg|\frac{c^\lambda_\ell}{|c^\lambda_\ell|}-\frac{c_\ell^\lambda}{|c_\ell^\star|}\bigg|+\frac{|c_\ell^\star-c_\ell^\lambda|}{|c^\star_\ell|}
\\
=& |c_\ell^\lambda|\bigg|\frac{1}{|c^\lambda_\ell|}-\frac{1}{|c_\ell^\star|}\bigg|+\frac{|c_\ell^\star-c_\ell^\lambda|}{|c^\star_\ell|}
\\=&\left| c^\lambda_\ell\right|\bigg|\frac{|c^\lambda_\ell|-|c^\star_\ell|}{|c^\lambda_\ell||c^\star_\ell|}\bigg|+\frac{|c_\ell^\star-c_\ell^\lambda|}{|c^\star_\ell|}
\\\leq& 2\frac{|c_\ell^\star-c_\ell^\lambda|}{|c^\star_\ell|}
\leq  2X^\star  \gamma.
\end{aligned}
\end{equation}

\bigskip
\noindent{\bf Bounding ${\bnu^\star(f)},   {\bnu^\star(f)}',  {\bnu^\star(f)}''$ and $({\bnu^\star(f)}-{\bnu^\lambda(f)}),  ({\bnu^\star(f)}-{\bnu^\lambda(f)})',  ({\bnu^\star(f)}-{\bnu^\lambda(f)})''$.}
\\
Applying the triangle inequality and the sub-multiplicative property of the norm to~\eqref{eqn:sec:G:QW:abc} and~\eqref{eqn:v},  we get
\begin{equation}\label{eqn:w}
\begin{aligned}
\|{\bnu^\star(f)}-{\bnu^\lambda(f)}\|_1
\leq&\|[\D_1(f, \f^\lambda)-\D_1(f, \f^\star)]^T \diag(\c^\lambda)\mS^{-1}-\D_1(f, \f^\lambda)^T \mathbf{\Phi}\|_1
+2\|\D_0(f, \f^\lambda)-\D_0(f, \f^\star)\|_1
\\
\leq& \| \D_1(f, \f^\lambda)-\D_1(f, \f^\star)\|_1\| \diag(\c^\lambda)\mS^{-1}\|_{1, 1}
+\|\D_1(f, \f^\lambda)\|_1 \|\mathbf{\Phi}\|_{1, 1}
+2\|\D_0(f, \f^\lambda)-\D_0(f, \f^\star)\|_1;
\\[1ex]
\|\bnu^\star(f)\|_1 \leq& \|\D_1(f, \f^\star)\diag(\c^\star)\mS^{-1}\|_1+2\|\D_0(f, \f^\star)\|_1
\leq\|\D_1(f, \f^\star)\|_1\|\diag(\c^\star)\mS^{-1}\|_{1, 1}+2\|\D_0(f, \f^\star)\|_1, 
\end{aligned}
\end{equation}
where $\mathbf{\Phi}:= \diag(\c^\star)\mS^{-1}-\diag(\c^\lambda)\mS^{-1} $ and $\D_\ell(f, \f):=[K^{(\ell)}({f}_1-f), \ldots,  K^{(\ell)}({f}_k-f)]$. Similar bounds also apply to various derivatives of $\bnu^\star(f)$ and $\bnu^\lambda(f)$,  which we need in order to bound the distances between derivatives of $Q^\star(f)$ and $Q^\lambda(f)$. Using~\eqref{eqn:parameter:close} and $\tau\geq3.289 n^2$,  we have
\begin{equation}
\begin{aligned}
\left\|(\diag(\c^\lambda)-\diag(\c^\star))\mS^{-1}\right\|_{1, 1}
\leq&(\max_i |c_i^\lambda-c^\star_i|/|c^\star_i|)/\sqrt{{\tau}}
\leq X^\star  \gamma/\sqrt{{\tau}}
\leq {0.552 X^\star  \gamma}/{n};
\\ \left\|\diag(\c^\lambda)\mS^{-1}\right\|_{1, 1} \leq
&(1+X  \gamma)/\sqrt{{\tau}} \leq {0.552 }/{n}, \label{eqn:coeff:diff}
\end{aligned}
\end{equation}
which we need to continue the bounds in~\eqref{eqn:w}.

Since $f$ may lie in different regions: Near Region,  Middle Region,   and Far Region,  we next organize our analysis into three parts based on what region $f$ is located in.

\subsection{Near Region Analysis}
\noindent We start with controlling $\|\D_\ell(f, \f^\lambda)-\D_\ell(f, \f^\star)\|_1$ and $|\D_\ell(f, \f^\star)\|_1$ for $\ell=0, 1, 2, 3$ in Near Region. When $\ell = 0$,  we have
\begin{align}
\|\D_0(f, \f^\lambda)-\D_0(f, \f^\star)\|_1
=\sum_\ell |K(f^\lambda_\ell-f)-K(f^\star_\ell-f)|
&\stack{\ding{172}}{\leq}  \sum_\ell |K'(\tilde{f}_\ell-f)|\|\f^\lambda-\f^\star\|_\infty
\nn\\
&\stack{\ding{173}}{\leq}\left(F_1(2.5/n, 0.2404/n)+\max_{f\in \hat{\N}}|K'(f)|\right)\|\f^\lambda-\f^\star\|_\infty
\nn\\
&\stack{\ding{174}}{\leq}(0.01241n+ 0.790885 n)(0.4X^\star \gamma/n)
\nn\\
&\leq0.321318 X^\star\gamma, \label{eqn:near:a}
\end{align}
where \ding{172} is due to the mean value theorem with $\tilde{f}_\ell$ located between $f^\star_\ell$ and $f_\ell^\lambda$.
\ding{173} follows from Lemma~\ref{lem:control:middle}.
To see this,  first note that $\Delta(\{\tilde{f}_\ell\})\geq2.5/n$ by Lemma~\ref{lem:resolution}. Second,  $f\in\N=[0, 0.24/n]$ implies that
\begin{align*}
0\leq|f-\tilde{f}_0|
\leq& |f-f_0^\star|+|f_0^\star-\tilde{f}_0|
\leq 0.24/n+0.4(10^{-3})/n
=0.2404/n.
\end{align*}
We also used the definition $\hat{\N} =[0, 0.2404/n]$ in \ding{173}.
\ding{174} follows from the upper bound on $F_1(2.5/n, 0.2404/n)$ in Table~\ref{tbl:bounds:1},  the upper bound on
$\max_{f\in\hat\N}|K'(f)|$ in  Table~\ref{tbl:bounds:2},  as well as the upper bound on $\|\f^\lambda-\f^\star\|_\infty$ in Lemma~\ref{lem:fix1}.

Applying arguments similar to those for~\eqref{eqn:near:a},  we can control $\|\D_\ell(f, \f^\lambda)-\D_\ell(f, \f^\star)\|_1$ as
\begin{equation}\label{eqn:sec:G:dif:kernel}
\begin{aligned}
\|\D_\ell(f, \f^\lambda)-\D_\ell(f, \f^\star)\|_1
\leq& (F_{\ell+1}(2.5/n, 0.2404/n)+\max_{f\in\hat\N}|K^{(\ell+1)}(f)|)\|\f^\lambda-\f^\star\|_\infty.
\end{aligned}
\end{equation}
We specialize the above inequality to $\ell = 1,  2,  3$ using the upper bounds on $F_\ell(2.5/n, 0.2404/n)$ in Table~\ref{tbl:bounds:1} and those on
$\max_{f\in\hat\N}|K^{(\ell)}(f)|$ in Table~\ref{tbl:bounds:2} to obtain
{\allowdisplaybreaks
\begin{align}
\|\D_1(f, \f^\lambda)-\D_1(f, \f^\star)\|_1
\leq& (F_2(2.5/n, 0.2404/n)+\max_{f\in\hat\N}|K''(f)|)\|\f^\lambda-\f^\star\|_\infty \nn\\
\leq& (0.05637n^2+3.290n^2)(0.4X^\star\gamma/n)
= 1.338548  nX^\star\gamma;
\label{eqn:near1.0}
\\[1ex]
\|\D_2(f, \f^\lambda)-\D_2(f, \f^\star)\|_1
\leq& (F_3(2.5/n, 0.2404/n)+\max_{f\in\hat\N}|K'''(f)|)\|\f^\lambda-\f^\star\|_\infty
\nn\\
\leq& (0.28838n^3+7.80572 n^3)(0.4X^\star\gamma/n)
= 3.23764 n^2X^\star\gamma;
\label{eqn:near1.1}
\\[1ex]
\|\D_3(f, \f^\lambda)-\D_3(f, \f^\star)\|_1
\leq& (F_4(2.5/n, 0.2404/n)+\max_{f\in\hat\N}|K''''(f)|)\|\f^\lambda-\f^\star\|_\infty
\nn\\
\leq& (1.671n^4+\bl{29.2227}n^4)(0.4X^\star\gamma/n)
= \bl{12.3575}  n^3X^\star\gamma.
\label{eqn:near1.2}
\end{align}
}

Furthermore,  we can use similar arguments and Lemma~\ref{lem:control:middle} to control $\|\D_\ell(f, \f)\|_1$ for $f \in \N$:
\begin{align}\label{eqn:sec:G:control:D(f, fstar)}
\|\D_\ell(f, \f^\star)\|_1 \leq F_\ell(2.5/n, 0.2404/n)+\max_{f\in\hat\N}|K^{(\ell)}(f)|, 
\end{align}
which specializes to
\begin{align}\label{eqn:near1.3}
\|\D_0(f, \f^\star)\|_1
\leq& F_0(2.5/n, 0.2404/n)+\max_{f\in\hat\N}|K(f)|
\leq 0.00757+1  = 1.00757;
\\
\label{eqn:near1.4}
\|\D_1(f, \f^\star)\|_1
\leq& F_1(2.5/n, 0.2404/n)+\max_{f\in\hat\N}|K'(f)|
\leq 0.01241n+\bl{0.790885}n  =\bl{0.803295}  n;
\\
\label{eqn:near1.5}
\|\D_2(f, \f^\star)\|_1
\leq& F_2(2.5/n, 0.2404/n)+\max_{f\in\hat\N}|K''(f)|
\leq 0.05637n^2+3.290n^2  =3.34637  n^2;
\\
\label{eqn:near1.6}
\|\D_3(f, \f^\star)\|_1
\leq& F_3(2.5/n, 0.2404/n)+\max_{f\in\hat\N}|K'''(f)|
\leq 0.28838n^3+7.80572 n^3 =8.0941  n^3.
\end{align}

With these preparations,  we are ready to control $\|{\bnu^\star(f)^{(\ell)}}-{\bnu^\lambda(f)^{(\ell)}}\|_1$ and $\|{\bnu^\star(f)^{(\ell)}}\|_1$ for $\ell=0, 1, 2$ in  Near Region. Generalizing~\eqref{eqn:w} to the $\ell$th derivative of $\bnu^\star(f)$ and $\bnu^\lambda(f)$ to get
\begin{equation}\label{eqn:sec:G:control:nu}
\begin{aligned}
\|{\bnu^\star(f)}^{(\ell)}-{\bnu^\lambda(f)}^{(\ell)}\|_1
\leq&\|\D_{\ell+1}(f, \f^\lambda)-\D_{\ell+1}(f, \f^\star)\|_1 \|\diag(\c^\lambda)\mS^{-1}\|_1
\\
+&\|\D_{\ell+1}(f, \f^\lambda)\|_1\|\diag(\c^\star)\mS^{-1}-\diag(\c^\lambda)\mS^{-1}\|_1
+2\|\D_{\ell}(f, \f^\lambda)-\D_{\ell}(f, \f^\star)\|_1;
\\[1ex]
\|{\bnu^\star(f)}^{(\ell)}\|_1
\leq& \|\D_{\ell+1}(f, \f^\star)\diag(\c^\star)\mS^{-1}\|_1+2\|\D_\ell(f, \f^\star)\|_1.
\end{aligned}
\end{equation}
Plugging Eq.~\eqref{eqn:near:a}, ~\eqref{eqn:near1.0}, ~\eqref{eqn:near1.4} and~\eqref{eqn:coeff:diff}  into~\eqref{eqn:sec:G:control:nu},  we obtain
\begin{equation}\label{eqn:N1.1}
\begin{aligned}
\|{\bnu^\star(f)}-{\bnu^\lambda(f)}\|_1
\leq&\|\D_{1}(f, \f^\lambda)-\D_{1}(f, \f^\star)\|_1 \|\diag(\c^\lambda)\mS^{-1}\|_1
\\
&+\|\D_{1}(f, \f^\lambda)\|_1\|\diag(\c^\star)\mS^{-1}-\diag(\c^\lambda)\mS^{-1}\|_1
+2\|\D_{0}(f, \f^\lambda)-\D_{0}(f, \f^\star)\|_1
\\
\leq& 1.338548 nX^\star  \gamma\frac{0.552 }{n}+ (0.803295 n)\frac{0.552 X^\star  \gamma}{n}
+2(0.321318 X^\star  \gamma)
\leq   1.82494X^\star\gamma.
\end{aligned}
\end{equation}
Plugging Eq.~\eqref{eqn:near1.0}-\eqref{eqn:near1.1}, ~\eqref{eqn:near1.5} and~\eqref{eqn:coeff:diff}  into~\eqref{eqn:sec:G:control:nu},  we obtain
\begin{equation}\label{eqn:N1.2}
\begin{aligned}
\|{\bnu^\star(f)}'-{\bnu^\lambda(f)}'\|_1
\leq&\|\D_{2}(f, \f^\lambda)-\D_{2}(f, \f^\star)\|_1 \|\diag(\c^\lambda)\mS^{-1}\|_1
\\
&
+\|\D_{2}(f, \f^\lambda)\|_1\|\diag(\c^\star)\mS^{-1}-\diag(\c^\lambda)\mS^{-1}\|_1
+2\|\D_{1}(f, \f^\lambda)-\D_{1}(f, \f^\star)\|_1
\\
\leq& 3.23764n^2X^\star  \gamma\frac{0.552 }{n}+(3.34637 n^2)\frac{0.552 X^\star  \gamma}{n}
%\\
%&
+2(1.338548 n X^\star  \gamma)
\\
\leq&  6.31147nX^\star\gamma.
\end{aligned}
\end{equation}
Plugging Eq.~\eqref{eqn:near1.1}-\eqref{eqn:near1.2}, ~\eqref{eqn:near1.6} and~\eqref{eqn:coeff:diff}  into~\eqref{eqn:sec:G:control:nu},  we obtain
\begin{equation}\label{eqn:N1.3}
\begin{aligned}
\|{\bnu^\star(f)}''-{\bnu^\lambda(f)}''\|_1
\leq&\|\D_{3}(f, \f^\lambda)-\D_{3}(f, \f^\star)\|_1 \|\diag(\c^\lambda)\mS^{-1}\|_1
\\
&+\|\D_{3}(f, \f^\lambda)\|_1\|\diag(\c^\star)\mS^{-1}-\diag(\c^\lambda)\mS^{-1}\|_1
+2\|\D_{2}(f, \f^\lambda)-\D_{2}(f, \f^\star)\|_1
\\
\leq& \bl{12.3575}  n^3X^\star  \gamma\frac{0.552 }{n}+(8.0941 n^3)\frac{0.552 X^\star  \gamma}{n}
+2(3.23764 n^2 X^\star  \gamma)
\\
\leq&   \bl{17.7646}n^2X^\star  \gamma.
\end{aligned}
\end{equation}
Similarly,  plugging Eq.~\eqref{eqn:near1.3}-\eqref{eqn:near1.4}  and~\eqref{eqn:coeff:diff}  into~\eqref{eqn:sec:G:control:nu},  we have
\begin{align}\label{eqn:N1.4}
\|{\bnu^\star(f)}\|_1
\leq&\|\D_{ 1}(f, \f^\star)\diag(\c^\star)\mS^{-1}\|_1+2\|\D_0(f, \f^\star)\|_1
\leq (0.803295 n)\frac{0.552 }{n}+2(1.00757)
\leq  2.45856.
\end{align}
Plugging Eq.~\eqref{eqn:near1.4}-\eqref{eqn:near1.5}  and~\eqref{eqn:coeff:diff}  into~\eqref{eqn:sec:G:control:nu},  we obtain
\begin{align}\label{eqn:N1.5}
\|{\bnu^\star(f)}'\|_1
&\leq\|\D_{2}(f, \f^\star)\diag(\c^\star)\mS^{-1}\|_1+2\|\D_1(f, \f^\star)\|_1
\leq   (3.34637 n^2)\frac{0.552 }{n}+2(0.803295 n)
\leq  3.4538n.
\end{align}
Finally,  plugging Eq.~\eqref{eqn:near1.5}-\eqref{eqn:near1.6}  and~\eqref{eqn:coeff:diff}  into~\eqref{eqn:sec:G:control:nu},  we arrive at
\begin{align}\label{eqn:N1.6}
\|{\bnu^\star(f)}''\|_1
&\leq\|\D_{3}(f, \f^\star)\diag(\c^\star)\mS^{-1}\|_1+2\|\D_2(f, \f^\star)\|_1
\leq (8.0941 n^3)\frac{0.552 }{n}+2(3.34637 n^2)
\leq  11.1607n^2.
\end{align}

We are now ready to control the pointwise distance between ${Q^\star}^{(\ell)}(f)$ and ${Q^\lambda}^{(\ell)}(f)$ using
\begin{align}
|{Q^\star}^{(\ell)}(f)-{Q^\lambda}^{(\ell)}(f)|\leq|{\bnu^\lambda(f)}^{(\ell)}\Xi^\lambda\brho^\lambda-{\bnu^\star(f)}^{(\ell)}\Xi^\star\brho^\star|, ~~~\ell=0, 1, 2. \label{eqn:Qstar-Qlambda}
\end{align}
Plugging Eq.~\eqref{eqn:N1.1}-\eqref{eqn:N1.2}, ~\eqref{eqn:sec:G:control:B.1}-\eqref{eqn:sec:G:control:B.3} and
\eqref{eqn:sec:G:control:C} to~\eqref{eqn:Qstar-Qlambda} with $\ell=0$,  we obtain for $f\in\N$
\begin{align*}
&|{Q^\star}(f)-{Q^\lambda}(f)|
\\
\leq&\|{\bnu^\star(f)}\|_1\|\Xi^\star\|_{\infty, \infty}\|\brho^\star-\brho^\lambda\|_1+\|{\bnu^\star(f)}\|_1\|\Xi^\star-\Xi^\lambda\|_{\infty, \infty}\|\brho^\lambda\|_\infty
+\|{\bnu^\star(f)}-{\bnu^\lambda(f)}\|_1\|\Xi^\lambda\|_{\infty, \infty}\|\brho^\lambda\|_\infty
\\
\leq&(2.45856)(1.03172)(2X^\star  \gamma)+(2.45856)(8.81222X^\star B^\star\gamma)
+(1.82494X^\star\gamma)(1.09363)
\leq 28.7343 X^\star B^\star\gamma.
\end{align*}
Plugging Eq.~\eqref{eqn:N1.3}-\eqref{eqn:N1.4}, ~\eqref{eqn:sec:G:control:B.1}-\eqref{eqn:sec:G:control:B.3} and
\eqref{eqn:sec:G:control:C} to~\eqref{eqn:Qstar-Qlambda} with $\ell=1$,  we obtain for $f\in\N$
\begin{align*}
&|{Q^\star}(f)'-{Q^\lambda}(f)'|
\\
\leq&\|{\bnu^\star(f)}'\|_1\|\Xi^\star\|_{\infty, \infty}\|\brho^\star-\brho^\lambda\|_1+\|{\bnu^\star(f)}'\|_1\|\Xi^\star-\Xi^\lambda\|_{\infty, \infty}\|\brho^\lambda\|_\infty
+\|{\bnu^\star(f)}'-{\bnu^\lambda(f)}'\|_1\|\Xi^\lambda\|_{\infty, \infty}\|\brho^\lambda\|_\infty
\\
\leq&(3.4538n)(1.03172)(2X^\star  \gamma)+(3.4538n)(8.81222X^\star B^\star\gamma)
+(6.31147nX^\star\gamma)(1.09363)
\leq44.4648 n X^\star B^\star\gamma.
\end{align*}
Finally,  plugging Eq.~\eqref{eqn:N1.5}-\eqref{eqn:N1.6}, ~\eqref{eqn:sec:G:control:B.1}-\eqref{eqn:sec:G:control:B.3} and
\eqref{eqn:sec:G:control:C} to~\eqref{eqn:Qstar-Qlambda} with $\ell=2$,  we get for $f\in\N$
\begin{align*}
&|{Q^\star}(f)''-{Q^\lambda}(f)''|
\\
\leq&\|{\bnu^\star(f)}''\|_1\|\Xi^\star\|_{\infty, \infty}\|\brho^\star-\brho^\lambda\|_1+\|{\bnu^\star(f)}''\|_1\|\Xi^\star-\Xi^\lambda\|_{\infty, \infty}\|\brho^\lambda\|_\infty
+\|{\bnu^\star(f)}''-{\bnu^\lambda(f)}''\|_1\|\Xi^\lambda\|_{\infty, \infty}\|\brho^\lambda\|_\infty
\\
\leq&(11.1607n^2)(1.03172)(2X^\star  \gamma)+(11.1607n^2) (8.81222X^\star B^\star\gamma )
+(\bl{17.7646}n^2X^\star\gamma)(1.09363)
\leq \bl{140.808} n^2 X^\star B^\star\gamma.
\end{align*}

\subsection{Middle Region Analysis}
We continue with bounding the pointwise distance between $Q^\star(f)$ and $Q^\lambda(f)$ in Middle Region $\M =[0.24/n, 0.75/n]$.
We start with controlling $\|\D_\ell(f, \f^\lambda)-\D_\ell(f, \f^\star)\|_1$ and $|\D_\ell(f, \f^\star)\|_1$  for $\ell=0, 1$.
First note when $f\in\M=[0.24/n, 0.75/n]$,  we have
\begin{align*}
(a)~|f-\tilde{f}_0|
\leq& |f-f_0^\star|+|f_0^\star-\tilde{f}_0|
\leq 0.75/n+0.0004/n
=0.7504/n, \\
(b)~|f-\tilde{f}_0|
\geq& |f-f_0^\star|-|f_0^\star-\tilde{f}_0|
\geq0.24/n-0.0004/n
=0.2396/n.
\end{align*}
Denote $\hat\M=[0.2396/n,  0.7504/n]$.  We combine the upper bounds on $F_\ell(2.5/n, 0.7504/n)$  in Table~\ref{tbl:bounds:1}  and the upper bounds on $\max_{f\in\hat\M}|K^{(\ell)}(f)|$  in Table~\ref{tbl:bounds:2} to get
\begin{equation}\label{eqn:middleI1.1}
\begin{aligned}
\|\D_0(f, \f^\lambda)-\D_0(f, \f^\star)\|_1
\leq&(F_1(2.5/n, 0.7504/n)+\max_{f\in\hat\M}|K'(f)|)\|\f^\lambda-\f^\star\|_\infty
\\
\leq& (0.01454n+2.46872n)(0.4X^\star\gamma/n)
=0.993304  X^\star\gamma;
\end{aligned}
\end{equation}
\begin{equation}\label{eqn:middleI1.2}
\begin{aligned}
\|\D_1(f, \f^\lambda)-\D_1(f, \f^\star)\|_1
\leq&(F_2(2.5/n, 0.7504/n)+\max_{f\in\hat\M}|K''(f)|)\|\f^\lambda-\f^\star\|_\infty
\\
\leq& (0.12675n^2+3.290n^2)(0.4X^\star\gamma/n)
=1.36670 n X^\star\gamma.
\end{aligned}
\end{equation}
In a similar manner,  we use Lemma~\ref{lem:control:middle} to control $\|\D_\ell(f, \f)\|_1$ as follows 
\begin{align}
&\|\D_0(f, \f)\|_1
\leq F_0(2.5/n, 0.7504/n)+\max_{f\in\hat\M}|K(f))|
\leq 0.00772 +\bl{0.90951}
=\bl{0.91723};
\label{eqn:middleI1.3}
\\
&\|\D_1(f, \f)\|_1
\leq F_1(2.5/n, 0.7504/n)+\max_{f\in\hat\M}|K'(f))|
\leq  0.01454n+2.46872n
=2.48326  n.
\label{eqn:middleI1.4}
\end{align}

To control $\|{\bnu^\star(f)}-{\bnu^\lambda(f)}\|_1$ and $\|{\bnu^\star(f)}\|_1$ in the  Middle Region,  we plug Eq.~\eqref{eqn:middleI1.1}-\eqref{eqn:middleI1.4} into~\eqref{eqn:sec:G:control:nu} to get
\begin{align}
\|{\bnu^\star(f)}-{\bnu^\lambda(f)}\|_1
\leq&\|\D_{1}(f, \f^\lambda)-\D_{1}(f, \f^\star)\|_1 \|\diag(\c^\lambda)\mS^{-1}\|_1+\|\D_{1}(f, \f^\lambda)\|_1\|\diag(\c^\star)\mS^{-1}-\diag(\c^\lambda)\mS^{-1}\|_1
\nn\\
&
+2\|\D_{0}(f, \f^\lambda)-\D_{0}(f, \f^\star)\|_1
\nn\\
\leq&  1.36670nX^\star \gamma\frac{0.552 }{n}
+(2.48326 n)\frac{0.552 X^\star  \gamma}{n}
+2(0.993304  X^\star  \gamma)
\leq 4.11179 X^\star\gamma;\label{eqn:M.1}
\\[1ex]
\|{\bnu^\star(f)}\|_1
\leq&\|\D_{ 1}(f, \f^\star)\diag(\c^\star)\mS^{-1}\|_1+2\|\D_0(f, \f^\star)\|_1
\leq (2.48326 n)\frac{0.552 }{n}+2(\bl{0.91723})
\leq \bl{3.20522}.\label{eqn:M.2}
\end{align}

Finally,  we control $|{Q^\star}(f)-{Q^\lambda}(f)|$ in Middle Region by
plugging Eq.~\eqref{eqn:M.1}-\eqref{eqn:M.2}, ~\eqref{eqn:sec:G:control:B.1}-\eqref{eqn:sec:G:control:B.3} and
\eqref{eqn:sec:G:control:C} to~\eqref{eqn:Qstar-Qlambda} with $\ell=0$ to get
\begin{align*}
|{Q^\star}(f)-{Q^\lambda}(f)|
\leq&\|{\bnu^\star(f)}\|_1\|\Xi^\star\|_{\infty, \infty}\|\brho^\star-\brho^\lambda\|_1+\|{\bnu^\star(f)}\|_1\|\Xi^\star-\Xi^\lambda\|_{\infty, \infty}\|\brho^\lambda\|_\infty
+\|{\bnu^\star(f)}-{\bnu^\lambda(f)}\|_1\|\Xi^\lambda\|_{\infty, \infty}\|\brho^\lambda\|_\infty
\\
\leq&( \bl{3.20522})(1.03172)(2X^\star  \gamma)+( \bl{3.20522}) (8.81222X^\star B^\star\gamma )
+(4.11179 X^\star \gamma )1.09363
\\
 \leq& \bl{39.3557}X^\star B^\star \gamma, \ f\in\M.
\end{align*}

\subsection{Far Region Analysis}
Lastly,  we bound the pointwise distance between $Q^\star(f)$ and $Q^\lambda(f)$ in Far Region $\F=[0.75/n,  f_1^\star/2]$. Again,  we start with controlling $\|\D_\ell(f, \f^\lambda)-\D_\ell(f, \f^\star)\|_1$ and $|\D_\ell(f, \f^\star)\|_1$ for $\ell=0, 1$. First note when $f\in\F=[0.75/n,  f_1^\star/2]$,  we have
\begin{align*}
&(a)~f-\tilde{f}_0
\geq  f-f_0^\star -|f_0^\star-\tilde{f}_0|
\geq 0.75/n-0.0004/n
=0.74996/n, 
\\
&(b)~\tilde{f}_1-f
\geq-|\tilde{f}_1-f_1^\star|+f_1^\star-f
\geq -0.0004n+f_1^\star/2
\geq -0.0004/n+(2.5009/n)/2
\geq1.25/n.
\end{align*}
Further note that $\{\tilde{f}_\ell\}$ satisfies the separation condition that $\Delta(\{\tilde{f}_\ell\})\geq2.5/n$.
Then,    following from  Lemma~\ref{lem:bounds:sum:K:far} and the upper bounds on $W_\ell(0.74996/n, 1.25/n)$ in Table~\ref{tbl:bounds:3},  we have
\begin{align}
\|\D_0(f, \f^\lambda)-\D_0(f, \f^\star)\|_1
\leq& \|\D_1(\tilde{\f}, f)\|_1\|\f^\lambda-\f^\star\|_\infty
\nn\\
\leq& W_1(0.74996/n, 1.25/n)\|\f^\lambda-\f^\star\|_\infty
\leq5.2265n(0.4X^\star\gamma/n)
=2.0906 X^\star\gamma;
\label{eqn:far.1}
\\[1ex]
\|\D_1(f, \f^\lambda)-\D_1(f, \f^\star)\|_1
\leq& W_2(0.74996/n, 1.25/n)\|\f^\lambda-\f^\star\|_\infty
\leq 48.033n^2(0.4X^\star\gamma/n)
=19.2132 nX^\star\gamma.
\label{eqn:far.2}
\end{align}
Similarly,  we can use Lemma~\ref{lem:bounds:sum:K:far} to  control $\|\D_\ell(f, \f)\|_1$ for $\ell = 0,  1$ and $f \in \F$:
\begin{equation}\label{eqn:far.3}
\begin{aligned}
\|\D_0(f, \f^\star)\|_1
\leq& W_0(0.74996/n, 1.25/n)\leq  0.71059;
\\
\|\D_1(f, \f^\star)\|_1
\leq& W_1(0.74996/n, 1.25/n)\leq5.2265n.
\end{aligned}
\end{equation}
Directly plugging Eq.~\eqref{eqn:far.1}-\eqref{eqn:far.3} into~\eqref{eqn:sec:G:control:nu},  we arrive at
\begin{align}
\|{\bnu^\star(f)}-{\bnu^\lambda(f)}\|_1
\leq&\|\D_{1}(f, \f^\lambda)-\D_{1}(f, \f^\star)\|_1 \|\diag(\c^\lambda)\mS^{-1}\|_1+\|\D_{1}(f, \f^\lambda)\|_1\|\diag(\c^\star)\mS^{-1}-\diag(\c^\lambda)\mS^{-1}\|_1
\nn\\
&
+2\|\D_{0}(f, \f^\lambda)-\D_{0}(f, \f^\star)\|_1
\nn\\
\leq&   19.2132nX^\star  \gamma\frac{0.552 }{n}+( 5.2265n)\frac{0.552 X^\star  \gamma}{n}
+2(2.0906 X^\star  \gamma)
\leq 17.6720X^\star\gamma;
\label{eqn:F.1}
\\[1ex]
\|{\bnu^\star(f)}\|_1
\leq&\|\D_{ 1}(f, \f^\star)\diag(\c^\star)\mS^{-1}\|_1+2\|\D_0(f, \f^\star)\|_1
\leq (5.2265n)\frac{0.552 }{n}+2(0.71059)
\leq 4.30621.
\label{eqn:F.2}
\end{align}
As a final step,  we control $|{Q^\star}(f)-{Q^\lambda}(f)|$ in Far Region by plugging Eq.~\eqref{eqn:F.1}-\eqref{eqn:F.2} and~\eqref{eqn:sec:G:control:B.1}-\eqref{eqn:sec:G:control:C} to~\eqref{eqn:Qstar-Qlambda} to get
\begin{align*}|{Q^\star}(f)-{Q^\lambda}(f)|
\leq&\|{\bnu^\star(f)}\|_1\|\Xi^\star\|_{\infty, \infty}\|\brho^\star-\brho^\lambda\|_1+\|{\bnu^\star(f)}\|_1\|\Xi^\star-\Xi^\lambda\|_{\infty, \infty}\|\brho^\lambda\|_\infty
+\|{\bnu^\star(f)}-{\bnu^\lambda(f)}\|_1\|\Xi^\lambda\|_{\infty, \infty}\|\brho^\lambda\|_\infty
\\
\leq&(4.30621)(1.03172)(2X^\star  \gamma)+(4.30621) (8.81222X^\star B^\star\gamma )
+ (17.6720X^\star \gamma )(1.09363)
\\
 \leq& 66.1596  X^\star {B^\star}\gamma, \ f\in\F.
\end{align*}
This concludes the proof of Lemma~\ref{lem:Q:lambda:closeto:Q:0}.
\end{proof}

\section{Proof of Lemma~\ref{lem:Q:hat:lambda:closeto:Q:lambda}}\label{sec:H}

\begin{proof}
The expressions
$\q^\lambda=\frac{\x^\star-\x^\lambda}{\lambda}$ and $\hat{\q}= \frac{\y-\hat{\x}}{\lambda}$ lead to
\begin{align*}
\hat{\q}-\q^\lambda
=\frac{(\y-\hat{\x})-(\x^\star-\x^\lambda)}{\lambda}
=\frac{\w}{\lambda}+\frac{\x^\lambda-\hat{\x}}{\lambda}, 
\end{align*}
implying
\begin{equation}
\begin{aligned}
|Q^\lambda(f)-\hat{Q}(f)|
&\leq \underbrace{\frac{|\a(f)^H\mZ\w|}{\lambda}}_{\Pi_1(f)} +\underbrace{\frac{| \a(f)^H\mZ(\x^\lambda-\hat{\x})| }{\lambda}}_{\Pi_2(f)}\label{eqn:Qlamba-Qhat}.
\end{aligned}
\end{equation}
This separates the distance between $Q^\lambda(f)$ and $\hat{Q}(f)$ into two parts: one is $\Pi_1(f)$ associated with the dual atomic norm of the Gaussian noise $\w$ whose upper bounds were developed in Appendix~\ref{sec:B}; the other is $\Pi_2(f)$ corresponding to the dual atomic norm of
$\x^\lambda-\hat{\x}$. The latter quantity can be bounded by similar arguments as controlling $|{ \a(f)^H\mZ(\x^\lambda-\hat{\x}) }|$ in Lemma~\ref{lem:Q:lambda:closeto:Q:0}.

\bigskip
\noindent{\bf Bounding $\Pi_1(f)$.}
\\
Combining Eq.~\eqref{eqn:noise:bound1}-\eqref{eqn:noise:bound2.2},   we can upperbound $\Pi_1(f), \Pi(f)'$ and $\Pi_1(f)''$ with high probability (at least $1-1/n^2$) for all $f\in\TT$:
\begin{equation}
\begin{aligned}
&\Pi_1(f)  \leq{6.534\gamma_0}/\lambda\leq 10.115/X^\star;
\\
&\Pi_1(f)'  \leq  {41.052n\gamma_0}/\lambda\leq63.458n/X^\star;
\\
&\Pi_1(f)''  \leq  {257.94n^2\gamma_0}/\lambda\leq 399.288 n^2/X^\star \label{eqn:Pi1}, 
\end{aligned}
\end{equation}
where we used $\lambda={0.646X^\star\gamma_0}$.

\bigskip\noindent
{\bf Bounding $\Pi_2(f)$.}
\begin{align*}
\Pi_2(f)
&=\frac{1}{\lambda}| \D_0(f, \f^\lambda)\c^\lambda-\D_0(f, \hat{\f})\hat{\c} |
\\
&\leq\frac{1}{\lambda}\|\D_0(f, \f^\lambda)-\D_0(f, \hat{\f}) \|_1\|\c^\lambda\|_\infty+
\frac{1}{\lambda}\| \D_0(f, \hat{\f}) \|_1\|\hat{\c}-\c^\lambda\|_\infty
\\
&\stack{\ding{172}}{\leq} \frac{1}{\lambda}\|\D_1(f, \tilde{\f} ) \|_1 \|\f^\lambda-\hat{\f}\|_\infty\|\c^\lambda\|_\infty+\frac{1}{\lambda}\| \D_0(f, \hat{\f}) \|_1\|\hat{\c}-\c^\lambda\|_\infty
\\
&\stack{\ding{173}}{\leq} \frac{c_{\max}^\lambda}{0.646X^\star  \gamma_0}\left(\frac{0.4(35.2)\gamma}{n}\|\D_1(f, \tilde{\f})\|_1 +35.2\gamma\|\D_0(f, \hat{\f})\|_1 \right)
\\
&\stack{\ding{174}}{\leq} \frac{ B^\star (1+X^\star \gamma)}{0.646X^\star}\left(\frac{14.08}{n}\|\D_1(f, \tilde{\f})\|_1 +35.2\|\D_0(f, \hat{\f})\|_1  \right), 
\end{align*}
where \ding{172} follows from the mean value theorem.  For \ding{173} to hold,  first note that $\lambda=0.646X^\star\gamma_0$   and  $\hat{\btheta}\in N^\lambda$  by    Lemma~\ref{lem:fix2}. Then,  we can upperbound $\|\hat{\c}-\c^\lambda\|_\infty$ as
\begin{align*}
\|\hat{\c}-\c^\lambda\|_\infty =  \frac{|\hat{c}_j-c^\lambda_j|}{|c^\lambda_j|}|c^\lambda_j| {\leq} (35.2 \gamma) c^\lambda_{\max}, 
\end{align*}
where the equality follows by assuming  the $\ell_\infty$ norm is achieved by the $j$th row  and the inequality follows by changing $X^\star$ to 35.2 in~\eqref{eqn:parameter:close}  and defining  $c^\lambda_{\max}:=\max_j |c^\lambda_j|$.
\ding{174} follows from $\gamma_0 =\gamma {c_{\min}^\star}$ and
\begin{align*}
\frac{c^\lambda_{\max}}{c^\star_{\min}} =  B^\star \frac{|c_j^\lambda|}{c^\star_{\max}}
\leq  B^\star \frac{|c_j^\lambda|}{|c^\star_{j}|}
\leq B^\star (1+X^\star  \gamma).
\end{align*}

As a consequence,  to control $\Pi_2(f)$,  it reduces to bounding $\|\D_\ell(f, \tilde{\f})\|_1$ and $\|\D_\ell(f, \hat{\f})\|_1$. For this purpose,  we first note that
$\{\Delta(\tilde T), \Delta(\hat T)\}\geq2.5/n$  by Lemma~\ref{lem:resolution}. Second,  by
\begin{align*}
\|\tilde{\f}-\f^\star\|_\infty
&\stack{\ding{172}}{\leq}\|\hat{\f}-\f^\star\|_\infty
\stack{\ding{173}}{\leq} 0.4(X^\star+35.2)\gamma
\stack{\ding{174}}{\leq} 0.0004/n + 1.408\times 10^{-6}/n=0.000401408/n, 
\end{align*}
where \ding{172} follows from the length of subinterval is no larger than the whole one. \ding{173} follows from Eq.~\eqref{eqn:parameter:close} and \ding{174} follows from the SNR condition~\eqref{eqn:snr}.
Thus,  we can follow the same arguments that lead to Eq.~\eqref{eqn:near1.3}-\eqref{eqn:near1.4} for Near Region,  Eq.~\eqref{eqn:middleI1.3}-\eqref{eqn:middleI1.4} for Middle Region,  and Eq.~\eqref{eqn:far.3} for Far Region to develop bounds on $\|\D_\ell(f, \hat{\f})\|_1$.

To have a concrete idea,  we first show how to control  $\|\D_\ell(f, \hat{\f})\|_1$ since  the upper bounds for $\|\D_\ell(f, \tilde{\f})\|_1$ then follows by $\|\tilde{\f}-\f^\star\|_\infty \leq\|\hat{\f}-\f^\star\|_\infty$.
First,  consider $f\in\N$. Then we have
\begin{align*}
0\leq|\hat{f}_0-f|
\leq&|\hat{f}_0-f_0^\star|+|f_0^\star-f|
\leq 0.000401408/n+  0.24/n
\leq 0.240401408/n.
\end{align*}
With some abuse of notation,   we denote $\hat\N:=[0, 0.240401408/n]$.
Second,  consider $f\in\M$. Then we have
\begin{align*}
(a)~|f-\hat{f}_0|
\leq& |f-f_0^\star|+|f_0^\star-\hat{f}_0|
\leq0.75/n+0.000401408/n
=0.750401408/n;
\\ 
(b)~|f-\hat{f}_0|
\geq& |f-f_0^\star|-|f_0^\star-\hat{f}_0|
\geq0.24/n-0.000401408/n.
=0.239598592/n.
\end{align*}
Denote $\hat\M=[0.2396/n,  0.7504/n]$.
At last,  we consider $f\in\F=[0.75/n,  f_1^\star/2]$:
\begin{align*}
(a)~f-\hat{f}_0
\geq&  f-f_0^\star -|f_0^\star-\hat{f}_0|
\geq 0.75/n-0.000401408/n
=0.749598592/n;
\\ 
(b)~\hat{f}_1-f \geq&-|\hat{f}_1-f_1^\star|+f_1^\star-f
\geq -0.000401408/n+f_1^\star/2
\geq -0.000401408/n+(2.5009/n)/2
\geq1.25/n.
\end{align*}
Hence we can define $\hat\F:=[0.749598592/n,  1.25/n].$
Furthermore,  we remark that those numerical upper bounds in Table~\ref{tbl:bounds:1}-\ref{tbl:bounds:2} do not change when evaluated for the newly defined intervals $\hat\N$,  $\hat\M$ and $\hat\F$.

Finally,  by directly plugging the upper bounds of $\|\D_\ell(f, \f)\|_\infty$ in~\eqref{eqn:near1.3}-\eqref{eqn:near1.4} for Near Region, ~\eqref{eqn:middleI1.3}-\eqref{eqn:middleI1.4} for Middle Region,   and  equation\eqref{eqn:far.3} for Far Region,   it follows that
\begin{align}
\Pi_2(f)\leq&\frac{ B^\star}{ X^\star}
\begin{cases}
\frac{1.001}{0.646}(\frac{14.08}{n}(0.803295 n)+35.2(1.00757))
\leq 72.4825\frac{ B^\star}{ X^\star},  \    f\in\N;
\\[1ex]
\frac{1.001}{0.646}(\frac{14.08}{n}(2.48326 n)+35.2(\bl{0.91723}))
\leq \bl{104.208}\frac{ B^\star}{ X^\star},  \    f\in\M;
\\[1ex]
\frac{1.001}{0.646}(\frac{14.08}{n}(5.2265n)+35.2(0.71059))
\leq152.788\frac{ B^\star}{ X^\star},  \        f\in\F.
\end{cases}
\label{eqn:Pi2}
\end{align}

 Similarly,  from~\eqref{eqn:near1.4}-\eqref{eqn:near1.6},  we have an upper bound on  $\Pi_2(f)'$ and $\Pi_2(f)''$ as follows 
\begin{align}
\Pi_2(f)'
\leq& \frac{ B^\star (1+X^\star \gamma)}{0.646 X^\star}\left(\frac{14.08}{n}\|\D_2(f, \tilde{\f})\|_1 +35.2\|\D_1(f, \hat{\f})\|_1  \right)
\nn\\
\leq& \frac{ B^\star  }{  X^\star}
\frac{1.001}{0.646}\left(\frac{14.08}{n}(3.34637 n^2)+35.2(0.803295 n)\right)
\leq 116.825 n \frac{ B^\star}{ X^\star}, \ f\in\N;
\label{eqn:Pi2:b.1}
\\[1ex]
\Pi_2(f)''
\leq& \frac{ B^\star (1+X^\star \gamma)}{0.646 X^\star}\left(\frac{14.08}{n}\|\D_3(f, \tilde{\f})\|_1 +35.2\|\D_2(f, \hat{\f})\|_1  \right)
\nn\\
\leq& \frac{ B^\star}{  X^\star}
\frac{1.001}{0.646}\left(\frac{14.08}{n}(8.0941 n^3)+35.2(3.34637 n^2)\right)
\leq 359.116   n^2 \frac{ B^\star}{ X^\star}, \ f\in\N.
\label{eqn:Pi2:b.2}
\end{align}

Combining~\eqref{eqn:Pi1}-\eqref{eqn:Pi2:b.2} for $\Pi_1(f)$ and $\Pi_2(f)$,  we can control
$|{\hat{Q}}^{(\ell)}(f)-{Q^\lambda}^{(\ell)}(f)|$ in Near region $f\in\N$ as follows 
\begin{align*}
&|{\hat{Q}}(f)-{Q^\lambda}(f)|
\leq(10.115+72.4825)B^\star/{X^\star}
=    82.5975  B^\star/{X^\star}, \ f\in\N;
\\
&|{\hat{Q}}(f)'-{Q^\lambda}'(f)|
\leq(63.458n+116.825n)B^\star/{X^\star}
=    180.283 n  B^\star/{X^\star}, \ f\in\N;
\\
&|{\hat{Q}}(f)''-{Q^\lambda}''(f)|\leq(399.288n^2+359.116 n^2)B^\star/{X^\star}
=    758.404n^2  B^\star/{X^\star}, \ f\in\N.
\end{align*}

For the case of Middle Region and Far Region,  we can upperbound them as:
\begin{align*}
|{\hat{Q}}(f)-{Q^\lambda}(f)|
&\leq(10.115+ \bl{104.208})B^\star/{X^\star}
=      \bl{114.323}  B^\star/{X^\star}, ~f\in\M;
\\
|{\hat{Q}}(f)-{Q^\lambda}(f)|
&\leq(10.115+ 152.788)B^\star/{X^\star}
=   162.903  B^\star/{X^\star}, ~f\in\F.
\end{align*}
This completes the proof of  Lemma~\ref{lem:Q:hat:lambda:closeto:Q:lambda}.
\end{proof}

\section{Proof of Proposition~\ref{pro:bip}}\label{sec:proof:optimalitycondition}
\begin{proof}
The uniqueness follows from the strongly convex quadratic term in~\eqref{eqn:primal}. We next show the primal optimality of $\hat{\x}$ and the dual optimality of $\hat\q$ by establishing strong duality. First,  $\hat{\q}$ is feasible to the dual program~\eqref{eqn:dual} because of the BIP property.  Second,  we have the following chain of inequalities:
\begin{align*}
\text{value of~\eqref{eqn:dual}}
&=\frac{1}{2}\|\y\|_{\mZ}^2-\frac{1}{2}\|\y-\lambda\hat{\q}\|_{\mZ}^2\\
%&=\frac{1}{2}\|\lambda\hat{\q}+\hat{\x}\|_{\mZ}^2-\frac{1}{2}\|\hat{\x}\|_{\mZ}^2\\
&=\frac{1}{2}\|\lambda\hat{\q}\|_{\mZ}^2+\lambda\R\{\hat{\x}^H\mZ\hat{\q}\}\\
&=\frac{1}{2}\|\y-\hat{\x}\|_{\mZ}^2 + \lambda\|\hat{\c}\|_1\\
&\geq\frac{1}{2}\|\y-\hat{\x}\|_{\mZ}^2 +\lambda\|\hat{\x}\|_\A
=\text{value of  }\eqref{eqn:primal}, 
\end{align*}
where the second line follows by plugging  $\y=\hat{\x}+\lambda\hat{\q}$; the third line holds due to the Interpolation property; and the last line holds
since $\|\hat{\x}\|_\A\leq \|\hat{\c}\|_1$ by~\eqref{eqn:def:atomic:norm}.
Since the weak duality theorem ensures that the other direction of the inequality always holds,    we obtain strong duality. As a consequence,  $\hat{\x}$ and $\hat{\q}$ achieve primal optimality and dual optimality,  respectively. This means $\hat{\x} = {\x}^{\mathrm{glob}},  \hat{\q} = {\q}^{\mathrm{glob}}$ due to uniqueness of the solutions.
\end{proof}

\section{Proof of Corollary~\ref{cor:main2}}\label{sec:connection}
\begin{proof}
Denote by $F(\x)$ the objective functions for~\eqref{eqn:primal} and $G(\f, \c)$ for~\eqref{eqn:pdw},  respectively. Assume $(\f^{\mathrm{non}}, \c^{\mathrm{non}})$ is a global optimum for~\eqref{eqn:pdw} with $\x^{\mathrm{non}} = \mA(\f^{\mathrm{non}})\c^{\mathrm{non}}$,  and $\x^{\mathrm{glob}}=\mA(\f^{\mathrm{glob}})\c^{\mathrm{glob}}$ is the global optimum of~\eqref{eqn:primal}.
Then
\begin{align}\label{eqn:proof:global}
F({\x}^{\mathrm{glob}}) \leq F({\x}^{\mathrm{non}}) \leq G({\f}^{\mathrm{non}}, {\c}^{\mathrm{non}}) \leq G({\f}^{\mathrm{glob}}, {\c}^{\mathrm{glob}}), 
\end{align}
where the first inequality uses the optimality of $\x^{\mathrm{glob}}$ to~\eqref{eqn:primal}; the second inequality follows from $\|{\x}^{\mathrm{non}}\|_{\A}\leq \|{\c}^{\mathrm{non}}_\ell\|_1$ by~\eqref{eqn:def:atomic:norm}; and the last inequality follows from the optimality of $(\f^{\mathrm{non}}, \c^{\mathrm{non}})$ to~\eqref{eqn:pdw}. On the other hand,   recognize that $\|{\x}^{\mathrm{glob}}\|_{\A}= \|{\c}^{\mathrm{glob}}_\ell\|_1$ since $\{f^{\mathrm{glob}}_\ell\}$ satisfies the separation condition (revealed by Lemma~\ref{lem:resolution} in Appendix~\ref{sec:A}). This leads to
$G({\f}^{\mathrm{glob}}, {\c}^{\mathrm{glob}})=F({\x}^{\mathrm{glob}}).$ Therefore,  all inequalities in~\eqref{eqn:proof:global} become equalities and hence
$
 G({\f}^{\mathrm{non}}, {\c}^{\mathrm{non}})= G({\f}^{\mathrm{glob}}, {\c}^{\mathrm{glob}}).
$
This implies the global optimality of $({\f}^{\mathrm{glob}}, {\c}^{\mathrm{glob}})$ for the  nonconvex program~\eqref{eqn:pdw}.

\end{proof}

\section{Proof of Lemma~\ref{lem:resolution}}\label{sec:proof:resolution}

\begin{proof}
First of all,  from Lemma~\ref{lem:fix1},  we have $\btheta^\lambda\in\N^\star$,  implying $\|\f^\lambda-\f^\star\|_\infty\leq 0.4X^\star B^\star \gamma/n$ by Eq.~\eqref{eqn:parameter:close} and by Lemma~\ref{lem:fix2},  we obtain that $\hat{\btheta}\in\N^\lambda$,  which implies $\|\hat{\f}-\f^\lambda\|_\infty\leq 0.4(35.2) B^\star \gamma/n$ by Eq.~\eqref{eqn:parameter:close}. More precisely,  we bound $\Delta(T^\lambda)$ as
\begin{align*}
\Delta(T^\lambda)
&\stack{\ding{172}}{=}\min_{i\neq j}|f^\lambda_i-f^\lambda_j|
\\
&=\min_{i\neq j}|f^\lambda_i-f_i^\star+f_i^\star-f_j^\star+f_j^\star-f^\lambda_j|
\\
&\stack{\ding{173}}{\geq} \min_{i\neq j}|f_i^\star-f_j^\star| - \max_{i}|f^\lambda_i-f_i^\star|-\max_{j}|f^\lambda_j-f_j^\star|
\\
&\stack{\ding{174}}{\geq} \Delta(T^\star)-0.8X^\star B^\star\gamma/n
\\
&\stack{\ding{175}}{\geq} 2.5009/n-0.0008/n=2.5001/n >2.5/n, 
\end{align*}
where \ding{172} follows from the definition of the separation distance and \ding{173} follows from the triangle inequality. \ding{174} follows from
that $\btheta^\lambda$ is the fixed point solution of the contraction map~\eqref{eqn:Map:1}. Thus,  $\btheta^\lambda\in\N^\star$ following from the non-escaping property by the contraction mapping theorem. This further implies that $\|\f^\lambda-\f^\star\|_\infty\leq 0.4X^\star B^\star \gamma/n$ by~\eqref{eqn:parameter:close}.
Finally,  \ding{175} follows from that $T^\star$ satisfies the separation condition~\eqref{eqn:separation}: $\Delta(T^\star)\geq 2.5009/n.$

For bounding $\Delta(\tilde{T})$,  first identify that
$-\max_{i}|\tilde{f}_i-f_i^\star|\geq -\max_{i}|f^\lambda_i-f_i^\star|, $
since the inner point $\tilde{f}_i$ is included in the interval $[f_i^\star,  f^\lambda_i]$ and hence the length of the  $[\tilde{f}_i,  f^\star_i]$ is less than the entire interval $[f^\star_i,  f^\lambda_i]$. Then we immediately arrive at $\Delta(\tilde{T})>2.5/n$.

For $\Delta(\hat{T})$,  we have
\begin{align*}
\Delta(\hat{T})
&=\min_{i\neq j}|\hat{f}_i-\hat{f}_j|
\\
&=\min_{i\neq j}|\hat{f}_i-f_i^\lambda+f_i^\lambda-f_j^\lambda+f_j^\lambda-\hat{f}_j|
\\
&\stack{\ding{172}}{\geq} \min_{i\neq j}|f_i^\lambda-f_j^\lambda| - \max_{i}|\hat{f}_i-f_i^\lambda|-\max_{j}|\hat{f}_j-f_j^\lambda|
\\
&\stack{\ding{173}}{\geq} \Delta(T^\lambda)-2\|\hat{\f}-\f^\lambda\|_\infty
\\
&\stack{\ding{173}}{\geq} \Delta(T^\lambda)-2(14.08) B^\star\gamma/n, 
\end{align*}
where \ding{172} follows from the triangle inequality and \ding{173} follows from the definition of $\|\hat{\f}-\f^\lambda\|_\infty.$ \ding{174} follows from that $\|\hat{\f}-\f^\lambda\|_\infty\leq 0.4(35.2)B^\star\gamma/n=14.08 B^\star\gamma/n$ by~\eqref{eqn:parameter:close}. Finally following from the SNR condition~\eqref{eqn:snr} and $\Delta(T^\lambda)\geq 2.5001/n$,  we then have $\Delta(\hat{T})\geq 2.5001/n-2(14.08)\times10^{-7}/n>2.5/n$.

$\Delta(\tilde{T}^\lambda)\geq2.5/n$ holds by the same strategy as $\Delta(\tilde{T})>2.5/n$.
\end{proof}

\end{appendices}

\bibliography{support}

\begin{thebibliography}{10}
\expandafter\ifx\csname url\endcsname\relax
  \def\url#1{\texttt{#1}}\fi
\expandafter\ifx\csname urlprefix\endcsname\relax\def\urlprefix{URL }\fi
\expandafter\ifx\csname href\endcsname\relax
  \def\href#1#2{#2} \def\path#1{#1}\fi

\bibitem{Candes:2014br}
E.~J. Cand{\`e}s, C.~Fernandez-Granda, {Towards a mathematical theory of
  super-resolution}, Communications on Pure and Applied Mathematics 67~(6)
  (2014) 906--956.

\bibitem{candes2013super}
E.~J. Cand{\`e}s, C.~Fernandez-Granda, Super-resolution from noisy data,
  Journal of Fourier Analysis and Applications 19~(6) (2013) 1229--1254.

\bibitem{fernandez2016super}
C.~Fernandez-Granda, Super-resolution of point sources via convex programming,
  Information and Inference: A Journal of the IMA 5~(3) (2016) 251--303.

\bibitem{Tang:2013fo}
G.~Tang, B.~N. Bhaskar, P.~Shah, B.~Recht, {Compressed Sensing Off the Grid},
  Information Theory, IEEE Transactions on 59~(11) (2013) 7465--7490.

\bibitem{Tang:2013gd}
G.~Tang, B.~N. Bhaskar, B.~Recht, {Near minimax line spectral estimation.},
  IEEE Transactions on Information Theory 61~(1) (2015) 499--512.

\bibitem{Tang:2014outlier}
G.~Tang, P.~Shah, B.~N. Bhaskar, B.~Recht, Robust line spectral estimation, in:
  2014 48th Asilomar Conference on Signals, Systems and Computers, 2014, pp.
  301--305.
\newblock \href {http://dx.doi.org/10.1109/ACSSC.2014.7094450}
  {\path{doi:10.1109/ACSSC.2014.7094450}}.

\bibitem{fernandez2016demixing}
C.~Fernandez-Granda, G.~Tang, X.~Wang, L.~Zheng, Demixing sines and spikes:
  Robust spectral super-resolution in the presence of outliers, Information and
  Inference: A Journal of the IMA (2016) iax005.

\bibitem{Stoica:1989dn}
P.~Stoica, N.~Arye, {MUSIC, maximum likelihood, and Cramer-Rao bound},
  Acoustics, Speech and Signal Processing, IEEE Transactions on 37~(5) (1989)
  720--741.

\bibitem{Chandrasekaran:2010hl}
V.~Chandrasekaran, B.~Recht, P.~A. Parrilo, A.~S. Willsky, {The convex geometry
  of linear inverse problems}, Foundations of Computational Mathematics 12~(6)
  (2012) 805--849.

\bibitem{bhaskar2013atomic}
B.~N. Bhaskar, G.~Tang, B.~Recht, Atomic norm denoising with applications to
  line spectral estimation, IEEE Transactions on Signal Processing 61~(23)
  (2013) 5987--5999.

\bibitem{rao2015forward}
N.~Rao, P.~Shah, S.~Wright, Forward--backward greedy algorithms for atomic norm
  regularization, IEEE Transactions on Signal Processing 63~(21) (2015)
  5798--5811.

\bibitem{tewari2011greedy}
A.~Tewari, P.~K. Ravikumar, I.~S. Dhillon, Greedy algorithms for structurally
  constrained high dimensional problems, in: Advances in Neural Information
  Processing Systems, 2011, pp. 882--890.

\bibitem{boyd2015alternating}
N.~Boyd, G.~Schiebinger, B.~Recht, The alternating descent conditional gradient
  method for sparse inverse problems, in: Computational Advances in
  Multi-Sensor Adaptive Processing (CAMSAP), 2015 IEEE 6th International
  Workshop on, IEEE, 2015, pp. 57--60.

\bibitem{eftekhari2013greed}
A.~Eftekhari, M.~B. Wakin, Greed is super: A new iterative method for
  super-resolution, in: Global Conference on Signal and Information Processing
  (GlobalSIP), 2013 IEEE, IEEE, 2013, pp. 631--631.

\bibitem{tang2015resolution}
G.~Tang, Resolution limits for atomic decompositions via markov-bernstein type
  inequalities, in: Sampling Theory and Applications (SampTA), 2015
  International Conference on, IEEE, 2015, pp. 548--552.

\bibitem{fernandez2013support}
C.~Fernandez-Granda, Support detection in super-resolution, in: Sampling Theory
  and Applications (SampTA), 2013 International Conference on, IEEE, 2013, pp.
  145--148.

\bibitem{Duval:2015gk}
V.~Duval, G.~Peyr{\'e}, {Exact support recovery for sparse spikes
  deconvolution}, Foundations of Computational Mathematics 15~(5) (2015)
  1315--1355.

\bibitem{Denoyelle2017}
Q.~Denoyelle, V.~Duval, G.~Peyr{\'e}, Support recovery for sparse
  super-resolution of positive measures, Journal of Fourier Analysis and
  Applications 23~(5) (2017) 1153--1194.

\bibitem{morgenshtern2016super}
V.~I. Morgenshtern, E.~J. Candes, Super-resolution of positive sources: The
  discrete setup, SIAM Journal on Imaging Sciences 9~(1) (2016) 412--444.

\bibitem{kay1988modern}
S.~M. Kay, Modern spectral estimation, Pearson Education India, 1988.

\bibitem{stoica1997introduction}
P.~Stoica, R.~L. Moses, Introduction to spectral analysis, Vol.~1, Prentice
  Hall Upper Saddle River, 1997.

\bibitem{prony1795}
R.~de~Prony, {Essai experimental et analytique}, J. Ec. Polytech.(Paris) 2
  (1795) 24--76.

\bibitem{kahn1992consistency}
M.~Kahn, M.~Mackisack, M.~Osborne, G.~Smyth, On the consistency of prony's
  method and related algorithms, Journal of Computational and Graphical
  Statistics 1~(4) (1992) 329--349.

\bibitem{Hua:1990kx}
Y.~Hua, T.~K. Sarkar, {Matrix pencil method for estimating parameters of
  exponentially damped/undamped sinusoids in noise}, Acoustics, Speech and
  Signal Processing, IEEE Transactions on 38~(5) (1990) 814--824.

\bibitem{music}
R.~Schmidt, {Multiple emitter location and signal parameter estimation}, IEEE
  Trans. on Antennas and Propagation 34~(3) (1986) 276--280.

\bibitem{esprit}
R.~Roy, T.~Kailath, {ESPRIT - estimation of signal parameters via rotational
  invariance techniques}, IEEE Trans. on Acoustics, Speech and Signal
  Processing 37~(7) (1989) 984--995.

\bibitem{fri}
M.~Vetterli, P.~Marziliano, T.~Blu, {Sampling signals with finite rate of
  innovation}, IEEE Trans. on Signal Processing 50~(6) (2002) 1417--1428.

\bibitem{donoho2006compressed}
D.~L. Donoho, Compressed sensing, IEEE Transactions on information theory
  52~(4) (2006) 1289--1306.

\bibitem{candes2006compressive}
E.~J. Cand{\`e}s, et~al., Compressive sampling, in: Proceedings of the
  international congress of mathematicians, Vol.~3, Madrid, Spain, 2006, pp.
  1433--1452.

\bibitem{baraniuk2007compressive}
R.~G. Baraniuk, Compressive sensing, IEEE signal processing magazine 24~(4).

\bibitem{chi2011sensitivity}
Y.~Chi, L.~L. Scharf, A.~Pezeshki, A.~R. Calderbank, Sensitivity to basis
  mismatch in compressed sensing, IEEE Transactions on Signal Processing 59~(5)
  (2011) 2182--2195.

\bibitem{tang2013sparse}
G.~Tang, B.~N. Bhaskar, B.~Recht, Sparse recovery over continuous
  dictionaries-just discretize, in: 2013 Asilomar Conference on Signals,
  Systems and Computers, IEEE, 2013, pp. 1043--1047.

\bibitem{duval2015sparse}
V.~Duval, G.~Peyr{\'e}, Sparse spikes deconvolution on thin grids, arXiv
  preprint arXiv:1503.08577.

\bibitem{li2016off}
Y.~Li, Y.~Chi, Off-the-grid line spectrum denoising and estimation with
  multiple measurement vectors, IEEE Transactions on Signal Processing 64~(5)
  (2016) 1257--1269.

\bibitem{yang2014exact}
Z.~Yang, L.~Xie, Exact joint sparse frequency recovery via optimization
  methods, IEEE Transactions on Signal Processing 64~(19) (2014) 5145--5157.

\bibitem{li2018atomic}
S.~Li, D.~Yang, G.~Tang, M.~B. Wakin, Atomic norm minimization for modal
  analysis from random and compressed samples, IEEE Transactions on Signal
  Processing 66~(7) (2018) 1817--1831.

\bibitem{azais2015spike}
J.-M. Azais, Y.~De~Castro, F.~Gamboa, Spike detection from inaccurate
  samplings, Applied and Computational Harmonic Analysis 38~(2) (2015)
  177--195.

\bibitem{wainwright2009sharp}
M.~J. Wainwright, Sharp thresholds for high-dimensional and noisy sparsity
  recovery using-constrained quadratic programming (lasso), IEEE transactions
  on information theory 55~(5) (2009) 2183--2202.

\bibitem{foucart2013mathematical}
S.~Foucart, H.~Rauhut, A mathematical introduction to compressive sensing,
  Springer, 2013.

\bibitem{nocedal2006numerical}
J.~Nocedal, S.~Wright, Numerical optimization, Springer Science \& Business
  Media, 2006.

\bibitem{bertsekas1999nonlinear}
D.~P. Bertsekas, Nonlinear programming, Athena Scientific Belmont, 1999.

\end{thebibliography}

\end{document}